\documentclass[a4paper,11pt,english,onecolumn,accepted=2023-04-19]{quantumarticle}

\def\papershowtoc{}
\pdfoutput=1
\usepackage[final]{hyperref}
\hypersetup{
           breaklinks=true,   
           colorlinks=true,   
           citecolor=red,
        }
\usepackage{natbib} 
\setcitestyle{square,compress,numbers,comma}

\usepackage[utf8]{inputenc}
\usepackage[a4paper,bindingoffset=0cm,left=2.0cm,right=2.0cm,top=2.5cm,bottom=2.5cm,footskip=1.0cm]{geometry}
\usepackage[english]{babel}
\usepackage[T1]{fontenc}
\usepackage{amsmath,amssymb}
\usepackage{amsthm}
\usepackage{thmtools}
\usepackage{thm-restate}
\usepackage{physics}
\usepackage{mathtools}
\usepackage{float}

\usepackage{tikz}
\usepackage{lipsum}
\usepackage[center]{caption}

\newtheorem{theorem}{Theorem} 
\newtheorem{definition}[theorem]{Definition}
\newtheorem{lemma}[theorem]{Lemma}

\newtheorem{corollary}[theorem]{Corollary}

\allowdisplaybreaks

\usepackage{braket}

\usepackage{xargs,xparse,xifthen,xstring,ulem,xspace}
\makeatletter
\font\uwavefont=lasyb10 scaled 652
\newcommand\colorwave[1][blue]{\bgroup\markoverwith{\lower3\p@\hbox{\uwavefont\textcolor{#1}{\char58}}}\ULon}
\makeatother

\ifdefined\papershowcomments
\usepackage[textsize=tiny]{todonotes}
\usepackage{marginnote}

\newcommand\createtodoauthor[2]{%
\def\tmpdefault{emptystring}
\expandafter\newcommand\csname #1\endcsname[2][\tmpdefault]{\def\tmp{##1}\ifthenelse{\equal{\tmp}{\tmpdefault}}
   {\todo[linecolor=#2,backgroundcolor=#2,bordercolor=#2]{\textbf{#1:} ##2}}
   {\ifthenelse{\equal{##2}{}}{\colorwave[#2]{##1}\xspace}{\todo[linecolor=#2,backgroundcolor=#2,bordercolor=#2]{\textbf{#1:} ##2}\colorwave[#2]{##1}}}}
\expandafter\newcommand\csname #1f\endcsname[2][\tmpdefault]{
	\marginnote{
		\todo[inline,linecolor=#2,backgroundcolor=#2,bordercolor=#2]{\textbf{#1 (Figure):} ##2}}
   }
}

  \paperwidth=\dimexpr \paperwidth + 6cm\relax
  \oddsidemargin=\dimexpr\oddsidemargin + 3cm\relax
  \evensidemargin=\dimexpr\evensidemargin + 3cm\relax
  \marginparwidth=\dimexpr \marginparwidth + 3cm\relax
\else
\newcommand\createtodoauthor[2]{%
\def\tmpdefault{emptystring}
\expandafter\newcommand\csname #1\endcsname[2][\tmpdefault]{\def\tmp{##1}\ifthenelse{\equal{\tmp}{\tmpdefault}}
   {}
   {\ifthenelse{\equal{##2}{}}{##1\xspace}{##1}}}
\expandafter\newcommand\csname #1f\endcsname[2][\tmpdefault]{
   }
   }
\fi

\definecolor{pairedNegOneLightGray}{HTML}{cacaca}
\definecolor{pairedNegTwoDarkGray}{HTML}{827b7b}
\definecolor{pairedOneLightBlue}{HTML}{a6cee3}
\createtodoauthor{aditya}{pairedOneLightBlue}
\createtodoauthor{anurudh}{pairedTwoDarkBlue}
\createtodoauthor{shantanav}{pairedThreeLightGreen}

\newcommand{\rr}{\mathbb{R}}
\newcommand{\nn}{\mathbb{N}}
\newcommand{\cc}{\mathbb{C}}
\newcommand{\paren}[1]{\left( #1 \right)}

\newcommand{\logp}[1]{\log \left( #1 \right)}

\newcommand{\normalized}[1]{\frac{#1}{\norm{#1}}}
\newcommand{\ohtilde}[1]{\widetilde{\mathcal{O}}\paren{ #1 }}
\newcommand{\polylog}[1]{\mathrm{polylog}\paren{#1}}
\newcommand{\smalloh}[1]{o\paren{#1}}
\newcommand{\controlled}[1]{C\text{-}#1}

\newcommand{\eps}{\varepsilon}
\newcommand{\Om}{\Omega}
\newcommand{\KO}{{\kappa_{\Omega}}}
\newcommand{\NO}{{\norm{\Omega}}}

\newcommand{\sigmamax}[1]{\sigma_{\max}\paren{#1}}
\newcommand{\sigmamin}[1]{\sigma_{\min}\paren{#1}}
\newcommand{\sigmaminnz}[1]{\widetilde\sigma_{\min}\paren{#1}}
\newcommand{\ot}{\otimes}

\def\wmax{{w_\text{max}}}
\def\wmin{{w_\text{min}}}

\def\OLSstate{\normalized{\paren{A^TA + \lambda L^TL}^{-1} A^T \ket b}}
\def\OLSKappa{
    \kappa_L
    \paren{1 +
    \frac{\norm{A}}{\sqrt\lambda \norm{L}}
    }
}
\newcommand{\OLScomplexity}[2]{
    \order{
        \kappa \log\kappa \paren{
            \paren{\frac{#1 + \sqrt\lambda #2}{\norm{A} + \sqrt\lambda \norm{L}}}
            \log\paren{\frac{\kappa}{\delta}}
            \paren{T_A + T_L}
            +
            T_b
        }
    }
}

\newcommand{\RidgeComplexity}{
    \order{
        \log\kappa
        \paren{
            \frac{\alpha_A}{\sqrt\lambda}
            \logp{\frac\kappa\delta}
            T_A
            +
            \kappa T_b
        }
    }
}

\newcommand{\GLSKappa}{
    \kappa_L \paren{1 + \frac{\sqrt{\KO}\norm{A}}{\sqrt{\lambda\norm{\Omega}}\norm{L}}}
}
\newcommand{\GLSComplexity}{
    \kappa \sqrt{\KO} \log\kappa
    \paren{
        \paren{
          \frac{\alpha_A}{\norm{A}} T_A
          + \frac{\alpha_L}{\norm{L}} T_L
          + \frac{\alpha_\Om \KO}{\NO} T_\Om
        }
        \log^3\paren{\frac{\kappa \KO \norm{A}\norm{L}}{\delta \NO}} 
        + T_b
    }
}

\newcommand{\UPDATED}{\color{blue}}
\renewcommand{\UPDATED}{}

\begin{document}

\title{Quantum Regularized Least Squares}

\newcommand{\CSTAR}{Center for Security, Theory and Algorithmic Research, IIIT Hyderabad, Telangana 500032, India}
\newcommand{\CQST}{Center for Quantum Science and Technology, IIIT Hyderabad, Telangana 500032, India}
\newcommand{\CCNSB}{Center for Computational Natural Sciences and Bioinformatics, IIIT Hyderabad, Telangana 500032, India}
\newcommand{\RUB}{Faculty of Computer Science, Ruhr University Bochum, 44801 Bochum, Germany}

\author{Shantanav Chakraborty}
\affiliation{\CQST}
\affiliation{\CSTAR}
\email{shchakra@iiit.ac.in}

\author{Aditya Morolia}
\affiliation{\CQST}
\affiliation{\CCNSB}
\email{aditya.morolia@research.iiit.ac.in}

\author{Anurudh Peduri}
\affiliation{\RUB}
\affiliation{\CQST}
\affiliation{\CSTAR}
\email{anurudh.peduri@rub.de}
\orcid{0000-0002-6523-7098}

\maketitle

\begin{abstract}
Linear regression is a widely used technique to fit linear models and finds widespread applications across different areas such as machine learning and statistics. In most real-world scenarios, however, linear regression problems are often \textit{ill-posed} or the underlying model suffers from \textit{overfitting}, leading to erroneous or trivial solutions. This is often dealt with by adding extra constraints, known as regularization. In this paper, we use the frameworks of block-encoding and quantum singular value transformation (QSVT) to design the first quantum algorithms for quantum least squares with general $\ell_2$-regularization. These include regularized versions of quantum ordinary least squares, quantum weighted least squares, and quantum generalized least squares. Our quantum algorithms substantially improve upon prior results on \textit{quantum ridge regression} (polynomial improvement in the condition number and an exponential improvement in accuracy), which is a particular case of our result.

To this end, we assume approximate block-encodings of the underlying matrices as input and use robust QSVT algorithms for various linear algebra operations. In particular, we develop a variable-time quantum algorithm for matrix inversion using QSVT, where we use quantum singular value discrimination as a subroutine instead of gapped phase estimation. This ensures that substantially fewer ancilla qubits are required for this procedure than prior results. Owing to the generality of the block-encoding framework, our algorithms are applicable to a variety of input models and can also be seen as improved and generalized versions of prior results on standard (non-regularized) quantum least squares algorithms.
\end{abstract}

\ifdefined\papershowtoc
\newpage
\tableofcontents
\newpage
\fi

\section{Introduction}
\label{sec:intro}

The problem of fitting a theoretical model to a large set of experimental data appears across various fields ranging from the natural sciences to machine learning and statistics \cite{murphy2012machine}. Linear regression is one of the most widely used procedures for achieving this. By assuming that, for the underlying model, there exists a linear relationship between a dependent variable and one or more explanatory variables, linear regression constructs the best linear fit to the series of data points. Usually, it does so while minimizing the sum of squared errors - known as the least squares method.

In other words, suppose that we are given $N$ data points $ \{ ( a_i, b_i ) \}_{i=1}^{N}$ where $ \forall i : a_i \in \rr^d, \forall i : b_i \in \rr$. The assumption is that each $b_i$ is linearly dependent on $a_i$ up to some random noise of mean $0$. Suppose $A$ is the data matrix of dimension $N\times d$, such that its $i^{\mathrm{th}}$ row is the vector $a_i$ and $b\in \rr^N$ such that $b=(b_1,\cdots, b_N)^T$. Then the procedure, known as ordinary least squares, obtains a vector $x\in \rr^d$ that minimizes the objective function $\left|\left|Ax-b\right|\right|^2_2$. 
This problem has a closed-form solution given by $x=(A^T A)^{-1} A^T b = A^+ b$, where $A^+$ denotes the Moore-Penrose inverse of the matrix $A$. Thus computationally, finding the best fit by linear regression reduces to finding the pseudoinverse of a matrix that represents the data, a task that is expensive for classical machines for large data sets.

In practice, however, least squares regression runs into problems such as \textit{overfitting}. For instance, the solution might fit most data points, even those corresponding to random noise. 
Furthermore, the linear regression problem may also be \textit{ill-posed}, for instance, when the number of variables exceeds the number of data points rendering it impossible to fit the data. These issues come up frequently with linear regression models and result in erroneous or trivial solutions. Furthermore, another frequent occurrence is that the data matrix $A$ has linearly dependent columns. In this scenario, the matrix $A^TA$ is not full rank and therefore is not invertible. 

Regularization is a widely used technique to remedy these problems, not just for linear regression but for inverse problems, in general \cite{engl1996regularization}. In the context of linear regression, broadly, this involves adding a penalty term to the objective function, which constrains the solution of the regression problem. For instance, in the case of $\ell_2$-regularization, the objective is to obtain $x$ that minimizes 
\begin{equation}
\norm{Ax - b}^2_2 + \lambda \norm{Lx}^2_2
\end{equation}
where $L$ is an appropriately chosen penalty matrix (or regularization matrix) of dimension $N\times d$ and $\lambda>0$ is the regularization parameter, an appropriately chosen constant. This regularization technique is known as \textit{general $\mathit{\ell_2}$-regularization} or \textit{Tikhonov regularization} in the literature \cite{hemmerle1975explicit, hanke1993regularization, bishop1995training, golub1999tikhonov, vanwieringen2021lecture}. It is a generalization of \textit{ridge regression} which corresponds to the case when $L$ is the identity matrix \cite{hoerl1970ridge, marquaridt1970generalized, vinod1978survey}. The closed-form solution of the general $\ell_2$-regularized ordinary least squares problem is given by 
\begin{equation}
x=\paren{A^TA+\lambda L^TL}^{-1}A^Tb.
\end{equation}
A straightforward observation is that even when $A^TA$ is singular, a judicious choice of the penalty matrix $L$ can ensure that the effective condition number (ratio of the maximum and the minimum non-zero singular values) of the overall matrix is finite and $A^TA+\lambda L^TL$ is invertible.
 
In this paper, we develop quantum algorithms for linear regression with general $\ell_2$-regularization. If the optimal solution is $x=(x_1,\cdots,x_d)^T$, then our quantum algorithm outputs a quantum state that is $\delta$-close to $\ket{x}=\sum_{j=1}^d x_j\ket{j}/\norm{x}$, assuming access to the matrices $A, L$, and the quantum state $\ket{b}$ via general quantum input models. 

In several practical scenarios, depending on the underlying theoretical model, generalizations of the ordinary least squares (OLS) technique are more useful to fit the data.
For instance, certain samples may be of more importance (and therefore have more weight) than the others, in which case weighted least squares (WLS) is preferred.
Generalized least squares (GLS) is used when the underlying samples obtained are correlated.
These techniques also suffer from the issues commonplace with OLS, warranting the need for regularization \cite{vanwieringen2021lecture}.
Consequently, we also design algorithms for quantum WLS with general $\ell_2$-regularization and quantum GLS with general $\ell_2$-regularization.

\textbf{Organization of the paper:} In the remainder of \autoref{sec:intro}, we formally describe $\ell_2$-regularized versions of OLS, WLS, and GLS (\autoref{subsec:intro-l2}), discuss prior and related work (\autoref{subsec:prior-work}), and outline our contributions and results (\autoref{subsec:our-contributions}). In \autoref{sec:prelims}, we briefly outline the framework of block-encoding and quantum input models that are particular instances of it (\autoref{ssec:quantum_input_models}). We also briefly introduce quantum singular value transformation (QSVT) (\autoref{ssec:qsvt}) and variable time amplitude amplification (VTAA) (\autoref{ssec:vtaa}). Following this, in \autoref{sec:algorithmic_primitives}, we develop several algorithmic primitives involving arithmetic of block-encodings (\autoref{subsec:arith-w-block}), quantum singular value discrimination (\autoref{subsec:qevd}) and quantum linear algebra using QSVT (\autoref{ssec:mi_using_qsvt}). These are the technical building blocks for designing our quantum regularized regression algorithms. Using these algorithmic primitives, we design quantum algorithms for the quantum least squares with $\ell_2$-regularization in \autoref{sec:main_proof}. Finally, we conclude by discussing some possible future research directions in \autoref{sec:discussion}.

\subsection{Linear regression with $\ell_2$-regularization} 
\label{subsec:intro-l2}
Suppose we are given data points $ \{ ( a_i, b_i ) \}_{i=1}^{N} $, where $ \forall i : a_i \in \rr^d, \forall i : b_i \in \rr $ such that $ (a_i, b_i) \sim_{i.i.d} \mathcal{D} $, i.e. they are sampled i.i.d. from some unknown distribution $\mathcal{D}$, assumed to be linear. We want to find a vector $x\in \rr^d $ such that the inner product $  x^T a_j$ is a good predictor for the target $b_j $ for some unknown $a_j$. 
This can be done by minimizing the total squared loss over the given data points,

\begin{equation}
    \mathcal{L}_O := \sum_j (x^T a_j - b_j)^2,
\end{equation}

leading to the ordinary least squares (OLS) optimization problem. The task then is to find $x\in \rr^d$ that minimizes $\norm{ Ax - b}_2^2$, where $A$ is the $N\times d$ data matrix such that the $i^{th}$ row of $ A $ is $ a_i $, and the $i^{th}$ element of the vector $b$ is $b_i$. Assuming that $A^TA$ is non-singular, the optimal $x$ satisfies
\begin{equation}
  x = (A^TA)^{-1} A^T b=A^{+}b,
\end{equation}
which corresponds to solving a linear system of equations. 

Suppose that out of the samples present in the data, we have higher confidence in some of them than others. In such a scenario, the $i^{\mathrm{th}}$ observation can be assigned a weight $w_i\in \rr$. This leads to a generalization of the OLS problem to \textit{weighted least squares} (WLS). In order to obtain the best linear fit, the task is now to minimize the weighted version of the loss 
\begin{equation}
    \label{eqn:wls_loss}
    \mathcal{L}_W := \sum_j w_j (x^T a_j - b_j)^2.
\end{equation}

As before, assuming $A^TWA$ is non-singular, the above loss function has the following closed-form solution:

\begin{equation}
    \label{eqn:wls_soln}
    x = (A^TWA)^{-1} A^T W b,
\end{equation}
where $W$ is a diagonal matrix with $w_i$ being the $i^{\mathrm{th}}$ diagonal element. 

There can arise scenarios where there exists some correlation between any two samples. For \textit{generalized least squares} (GLS), the presumed correlations between pairs of samples are given in a symmetric, non-singular covariance matrix $\Omega$. This objective is to find the vector $x$ that minimizes

\begin{equation}
    \mathcal{L}_{\Omega} := \sum_{i, j} (\Omega^{-1})_{i, j} (x^T a_i - b_i) (x^T a_j - b_j).
\end{equation}

Similarly, the closed-form solution for GLS is given by 

\begin{equation}
    \label{eqn:gls_soln}
    x = (A^T\Omega^{-1}A)^{-1} A^T \Omega^{-1} b.
\end{equation}

As mentioned previously, in several practical scenarios, the linear regression problem may be \textit{ill-posed} or suffer from \textit{overfitting}. Furthermore, the data may be such that some of the columns of the matrix $A$ are linearly dependent. This shrinks the rank of $A$, and consequently of the matrix $A^TA$, rendering it singular and, therefore non-invertible. Recall that the closed-form solution of OLS exists only if $A^TA$ is non-singular, which is no longer the case. Such scenarios arise even for WLS and GLS problems \cite{vanwieringen2021lecture}.  

In such cases, one resorts to \textit{regularization} to deal with them. Let $\mathcal{L}$ be the loss function to be minimized for the underlying least squares problem (such as OLS, WLS, or GLS). Then general $\ell_2$-regularization (\textit{Tikhonov regularization}) involves an additional penalty term so that the objective now is to find the vector $x\in\rr^d$ that minimizes
\begin{equation}
    \mathcal{L}+\lambda \norm{Lx}^2_2.
\end{equation}
Here $\lambda$, known as the regularization parameter, is a positive constant that controls the size of the vector $x$, while $L$ is known as the penalty matrix (or regularization matrix) that defines a (semi)norm on the solution through which the size is measured. The solution to the Tikhonov regularization problem also has a closed-form solution. For example, in the OLS problem, when $\mathcal{L}=\mathcal{L}_O$, we have that 
\begin{equation}
    x=(A^TA+\lambda L^TL)^{-1}A^T b.
\end{equation}
It is worth noting that when $L=I$, the $\ell_2$-regularized OLS problem is known as \textit{ridge regression}.
For the unregularized OLS problem, the singular values of $A$, $\sigma_j$ are mapped to $1/\sigma_j$.
The penalty term due to $\ell_2$-regularization, results in a shrinkage of the singular values.
This implies that even in the scenario where $A$ has linearly dependent columns (some $\sigma_j=0$) and $(A^TA)^{-1}$ does not exist, the inverse $(A^TA+\lambda L^TL)^{-1}$ is well defined for $\lambda>0$ {\UPDATED and any positive-definite $L$}. {\UPDATED Throughout this article, we refer to such an $L$ (which is positive definite) as a \textit{good regularizer}}. 
The penalty matrix $L$ allows for penalizing each regression parameter differently and leads to joint shrinkage among the elements of $x$. It also determines the rate and direction of shrinkage. In the special case of ridge regression, as $L=I$, the penalty shrinks each element of $x$ equally along the unit vectors $e_j$. {\UPDATED Also note that by definition, $I$ is a \textit{good regularizer}}.

Closed-form expressions can also be obtained for the WLS and the GLS problem ($\mathcal{L}=\mathcal{L}_W, \mathcal{L}_\Omega$ respectively), and finding the optimal solution $x$ reduces to solving a linear system. The quantum version of these algorithms output a quantum state that is $\epsilon$-close $\ket{x}=\sum_j x_j\ket{j}/\norm{x}$.

Throughout this work, while designing our quantum algorithms, we shall assume  access (via a block-encoding) to the matrices $A$, $W$, $\Omega$, and $L$ and knowledge of the parameter $\lambda$. Classically, however, the regularization matrix $L$ and the optimal parameter $\lambda$ are obtained via several heuristic techniques \cite{hanke1993regularization, golub1999tikhonov, vanwieringen2021lecture}.

\subsection{Prior work} 
\label{subsec:prior-work}

Quantum algorithms for (unregularized) linear regression was first developed by Wiebe et al. \cite{wiebe2012quantum}, wherein the authors made use of the HHL algorithm for solving a linear system of equations \cite{HHL2009}. Their algorithm assumes query access to a sparse matrix $A$ (\textit{sparse-access-model}) and to a procedure to prepare $\ket{b}=\sum_i b_i\ket{i}$. They first prepare a quantum state proportional to $A^T\ket{b}$, and then use the HHL algorithm to apply the operator $(A^TA)^{-1}$ to it. Overall the algorithm runs in a time scaling as $\kappa^6_A$ (the condition number of $A$) and inverse polynomial in the accuracy $\delta$. Subsequent results have considered the problem of obtaining classical outputs for linear regression. For instance, in Ref.~\cite{wang_2017_linear_regression}, $A^+$ is directly applied to the quantum state $\ket{b}$, followed by amplitude estimation to obtain the entries of $x$. On the other hand, Ref.~\cite{schuld2016predictionByRegression} used the techniques of quantum principal component analysis in \cite{lloyd2014quantum} to predict a new data point for the regression problem. These algorithms also work in the \textit{sparse access model} and run in a time that scales as $\mathrm{poly}\left(\kappa,1/\delta\right)$. Kerenidis and Prakash \cite{KP2020_iterative_quantum_gradient} provided a quantum algorithm for the WLS problem wherein they used a classical data structure to store the entries of $A$ and $W$. Furthermore, they assumed QRAM access to this data structure \cite{prakash2014quantum, kerenidis2016quantum} that would allow the preparation of quantum states proportional to the entries of $A$ and $W$ efficiently. They showed that in this input model (\textit{quantum data structure model}), an iterative quantum linear systems algorithm can prepare $\ket{x}$ in time $\widetilde{O}(\mu\kappa^3/\delta)$, where $\kappa$ is the condition number of the matrix $A^T\sqrt{W}$ while $\mu=\norm{\sqrt{W}A}_F$. Chakraborty et al. \cite{CGJ19} applied the framework of block-encoding along with (controlled) Hamiltonian simulation of Low and Chuang \cite{LC16} to design improved quantum algorithms for solving linear systems. Quantum algorithms developed in the block-encoding framework are applicable to a wide variety of input models, including the sparse access model and the quantum data structure model of \cite{KP2020_iterative_quantum_gradient}. They applied their quantum linear systems solver to develop quantum algorithms for quantum weighted least squares and generalized least squares. Their quantum algorithm for WLS has a complexity that is in $\widetilde{O}\left(\alpha\kappa\mathrm{polylog}(Nd/\delta)\right)$, where $\alpha=s$, the sparsity of the matrix $A^T\sqrt{W}$ in the sparse access model while $\alpha=\norm{\sqrt{W}A}_F$, for the quantum data structure input model. For GLS, their quantum algorithm outputs $\ket{x}$ in cost $\widetilde{O}\left(\kappa_A \kappa_\Omega (\alpha_A+\alpha_\Omega\kappa_\Omega)\mathrm{polylog}(1/\delta)\right)$, where $\kappa_A$ and $\kappa_\Omega$ are the condition numbers of $A$ and $\Omega$ respectively while $\alpha_A$ and $\alpha_\Omega$ are parameters that depend on how the matrices $A$ and $\Omega$ are accessed in the underlying input model. 

While quantum linear regression algorithms have been designed and subsequently improved over the years, quantum algorithms for regularized least squares have not been developed extensively. Yu et al. \cite{yu2019improved} developed a quantum algorithm for \textit{ridge regression} in the sparse access model using the LMR scheme \cite{lloyd2014quantum} for Hamiltonian simulation and quantum phase estimation, which they then used to determine the optimal value of the parameter $\lambda$. Their algorithm to output $\ket{x}$ has a cubic dependence on both $\kappa$ and $1/\delta$. They use this as a subroutine to determine a good value of $\lambda$. A few other works \cite{Shao2020, chen2022faster} have considered the quantum ridge regression problem in the sparse access model, all of which can be implemented with $\mathrm{poly}(\kappa,1/\delta)$ cost.

Recently, Chen and de Wolf designed quantum algorithms for lasso ($\ell_1$-regularization) and ridge regressions from the perspective of empirical loss minimization \cite{chen2021quantum}. For both lasso and ridge, their quantum algorithms output a classical vector $\widetilde{x}$ whose loss (mean squared error) is $\delta$-close to the minimum achievable loss. In this context, they prove a quantum lower bound of $\Omega(d/\delta)$ for ridge regression which indicates that in their setting, the dependence on $d$ cannot be improved on a quantum computer (the classical lower bound is also linear in $d$ and there exists a matching upper bound). Note that $\widetilde{x}$ is not necessarily close to the optimal solution $x$ of the corresponding least squares problem, even though their respective loss values are. Moreover, their result (of outputting a classical vector $\widetilde{x}$) is incomparable to our objective of obtaining a quantum state encoding the optimal solution to the regularized regression problem.

Finally, Gily\'{e}n et al. obtained a ``dequantized'' classical algorithm for ridge regression assuming norm squared access to input data similar to the quantum data structure input model \cite{Gilyen2020AnIQ}. Furthermore, similar to the quantum setting where the output is the quantum state $\ket{x}=\sum_j x_j\ket{j}/\norm{x}$ instead of $x$ itself, their algorithm obtains samples from the distribution $x_j^2/\norm{x}^2$. For the regularization parameter $\lambda=O\left(\norm{A}\norm{A}_F\right)$, the running time of their algorithm is in $\widetilde{O}\left(\kappa^{12} r_A^3/\delta^4\right)$, where $r_A$ is the rank of $A$. Their result (and several prior results) does not have a polynomial dependence on the dimension of $A$ and therefore rules out the possibility of generic exponential quantum speedup (except in $\delta$) in the quantum data structure input model. 

\subsection{Our contributions}
\label{subsec:our-contributions}

In this work, we design the first quantum algorithms for OLS, WLS, and GLS with general $\ell_2$-regularization.
We use the Quantum Singular Value Transformation (QSVT) framework introduced by Gily\'{e}n et al \cite{GSLW2019}.
We assume that the relevant matrices are provided as input in the \textit{block-encoding} model,
in which access to an input matrix $A$ is given by a unitary $U_A$ whose top-left block is (close to) $A/\alpha$.
The parameter $\alpha$ takes specific values depending on the underlying input model.
QSVT then allows us to implement nearly arbitrary polynomial transformations to a block of a unitary matrix using a series of parameterized, projector-controlled rotations (quantum signal processing \cite{LC16b}).

More precisely, given approximate block-encodings of the data matrix $A$ and the regularizing matrix $L$, and a unitary procedure to prepare the state $\ket{b}$, our quantum algorithms output a quantum state that is $\delta$-close to $\ket{x}$, the quantum state proportional to the $\ell_2$-regularized ordinary least squares (or weighted least squares or generalized least squares problem). We briefly summarize the query complexities of our results in \autoref{table:summary-of-results}.

For the OLS problem with general $\ell_2$-regularization (\autoref{subsec:ols-gls-l2}, \autoref{thm:quantum_least_squares_gen_tik_reg}), 
we design a quantum algorithm which
given an $(\alpha_A, a_A, \varepsilon_A)$-block-encoding of $A$ (implemented in cost $T_A$), an $(\alpha_L, a_L, \varepsilon_L)$-block-encoding of $L$ (implemented in cost $T_L$), a parameter $\lambda>0$, and a procedure to prepare $\ket{b}$ (in cost $T_b$), outputs a quantum state which is $\delta$-close to $\ket{x}$. The algorithm has a cost
\begin{equation*}
\UPDATED
    \OLScomplexity{\alpha_A}{\alpha_L}
\end{equation*}
where $\kappa$ can be thought of as a \textit{modified condition number}, related to the effective condition numbers of $A$ and $L$. When $L$ is a good regularizer, this is given by the expression
\[ \UPDATED \kappa = \OLSKappa, \]
Notice that $\kappa$ is independent of $\kappa_A$, the condition number of the data matrix $A$, which underscores the advantage of regularization.  
The parameters $\alpha_A$ and $\alpha_L$ take specific values depending on the underlying input model.
For the sparse access input model, $\alpha_A=s_A$ and $\alpha_L=s_L$, the respective sparsities of the matrices $A$ and $L$.
On the other hand for the quantum data structure input model,
$\alpha_A=\norm{A}_F$ and $\alpha_L=\norm{L}_F$.
Consequently, the complexity of \textit{Quantum Ridge Regression} can be obtained by substituting $L=I$ in the above complexity as
\[ \UPDATED \RidgeComplexity \]
where $\kappa = 1 + \norm{A}/\sqrt\lambda$, by noting that the block-encoding of $I$ is trivial while the norm and condition number of the identity matrix is one.
For this problem of \textit{quantum ridge regression}, our quantum algorithms are substantially better than prior results \cite{Shao2020, yu2019improved, chen2022faster}, exhibiting a polynomial improvement in $\kappa$ and an exponential improvement in $1/\delta$.

For the $\ell_2$-regularized GLS problem (\autoref{sec:main_proof}, \autoref{thm:gls_l2}),
    we design a quantum algorithm that along with approximate block-encodings of $A$ and $L$,
    takes as input an $(\alpha_\Omega, a_\Omega, \varepsilon_\Omega)$-lock-encoding of the matrix $\Omega$
    (implementable at a cost of $T_\Omega$)
    to output a state $\delta$-close to $\ket{x}$
    at a cost of \[ \UPDATED \order{\GLSComplexity} \]
In the above complexity, when $L$ is a good regularizer, the modified condition number $\kappa$ is defined as \[ \UPDATED \kappa = \GLSKappa \]

The WLS problem is a particular case of GLS, wherein the matrix $\Omega$ is diagonal. However, we show that better complexities for the $\ell_2$-regularized WLS problem can be obtained if we assume QRAM access to the diagonal entries of $W$ (\autoref{sec:main_proof}, \autoref{thm:wls_l2_qram} and \autoref{thm:wls_l2_sparse}).
 
\autoref{table:summary-of-results} summarizes the complexities of our algorithms for quantum linear regression with general $\ell_2$-regularization.
For better exposition, here we assume that $\norm{A}, \norm{L}, \norm{\Omega}$ and $\lambda = \Theta\paren{1}$.
For the general expression of the complexities, we refer the readers to \autoref{sec:main_proof}.
\begin{table}[H]
\centering
    \def\arraystretch{2}
    \begin{tabular}{ |c|c|c|}
        \hline
        \textbf{Problem} & \textbf{Unregularized} & \textbf{$\mathbf{\ell_2}$-Regularized} \\ 
        \hline\hline
        Quantum OLS 
        & $\ohtilde{\alpha_A\kappa_A \logp{1/\delta}}$ 
        & {\UPDATED $\ohtilde{(\alpha_A + \alpha_L) \kappa_L \logp{1/\delta}}$} \\[10pt]
        \hline
        Quantum GLS 
        & {\UPDATED $\ohtilde{\paren{\alpha_A + \alpha_\Omega \kappa_\Omega} \kappa_A \sqrt{\kappa_\Omega} \log^3\paren{1/\delta}}$ }
        & {\UPDATED $\ohtilde{\paren{\alpha_A + \alpha_L + \alpha_\Omega \kappa_\Omega} \kappa_L\sqrt{\kappa_\Omega}\log^3\paren{1/\delta}}$} \\[10pt]
        \hline
    \end{tabular}
\caption{Complexity of quantum linear regression algorithms with and without general $\ell_2$-regularization. All of these algorithms require only $\Theta({\log\kappa})$ additional qubits.}
	\label{table:summary-of-results}
\end{table}

In order to derive our results, we take advantage of the ability to efficiently perform arithmetic operations on block-encoded matrices, as outlined in \autoref{sec:algorithmic_primitives}. Along with this, we use QSVT to perform linear algebraic operations on block-encoded matrices. To this end, adapt the results in Refs.~\cite{GSLW2019, Grand_Uni_2021} to our setting. One of our contributions is that we work with robust versions of many of these algorithms. In prior works, QSVT is often applied to block-encoded matrices, assuming perfect block-encoding. For the quantum algorithms in this paper, we rigorously obtain the precision $\varepsilon$ required to obtain a $\delta$-approximation of the desired output state.  

For instance, a key ingredient of our algorithm for regularized least squares is to make use of QSVT to obtain $A^+$, given an $\varepsilon$-approximate block-encoding of $A$.  In order to obtain a (near) optimal dependence on the condition number of $A$ by applying variable-time amplitude amplification (VTAA) \cite{ambainis2010variable}, we recast the standard QSVT algorithm as a variable stopping-time quantum algorithm. Using QSVT instead of controlled Hamiltonian simulation ensures that the variable-time quantum procedure to prepare $A^+\ket{x}$ has a slightly better running time (by a log factor) and considerably fewer additional qubits than Refs.~\cite{CKS17,CGJ19}.

Furthermore, for the variable time matrix inversion algorithm, a crucial requirement is the application of the inversion procedure to the portion of the input state that is spanned by singular values larger than a certain threshold. In order to achieve this, prior results have made use of \textit{Gapped Phase Estimation} (GPE), which is a simple variant of the standard phase estimation procedure that decides whether the eigenvalue of a Hermitian operator is above or below a certain threshold \cite{ambainis2010variable, CKS17, CGJ19}. However, GPE can only be applied to a Hermitian matrix and requires additional registers that store the estimates of the phases, which are never used for variable-time amplitude amplification. In this work, instead of GPE, we develop a robust version of \textit{quantum singular value discrimination} (QSVD) using QSVT, which can be directly applied to non-Hermitian matrices. This algorithm decides whether some singular value of a matrix is above or below a certain threshold without storing estimates of the singular values. This leads to a \textit{space-efficient} variable time quantum algorithm for matrix inversion by further reducing the number of additional qubits required by a factor of $O(\log^2(\kappa/\delta))$ as compared to prior results \cite{CKS17, CGJ19}. Consequently, this also implies that in our framework, quantum algorithms for (unregularized) least squares (which are special cases of our result) have better complexities than those of Ref.~\cite{CGJ19}.

\section{Preliminaries}
\label{sec:prelims}

This section lays down the notation, and introduces the quantum singular value transformation (QSVT) and block-encoding frameworks, which are used to design the algorithm for quantum regression.

\subsection{Notation}
\label{ssec:notation}

For a matrix $ A \in \rr^{N \times d} $,
$A_{i, .}$ denotes the $i^{th}$ row of $A$,
and $\norm{A_{i,\cdot}}$ denotes the vector norm of $A_{i, .}^T$.
$s^A_r$ and $s^A_c$ denote the row and column sparsity of the matrix, which is the maximum number of non-zero entries in any row and any column, respectively.
\\~\\
\textbf{Singular Value Decomposition.} The decomposition $A = W\Sigma V^{\dagger}$, where $W$ and $V$ are unitary and $\Sigma$ is a diagonal matrix, represents the \textit{singular value decomposition} (SVD) of $A$. All matrices can be decomposed in this form. The diagonal entries of $\Sigma$, usually denoted by $\sigma(A) = \{\sigma_j\}$, is the multiset of all \textit{singular values} of $A$, which are real and non-negative. $\sigma_{\max}$ and $\sigma_{\min}$ denote the maximum and minimum singular values of $A$. $r(A) = \mathrm{rank}(A)$ is the number of non-zero singular values of $A$. The columns of $W,~V$ (denoted by $\{\ket{w_j}\}$ and $\{\ket{v_j}\}$) are the left and right \textit{singular vectors} of $A$. Thus $A = \sum_j^r \sigma_j \ket{w_j}\bra{v_j}$.
The singular vectors of $A$ can be computed as the positive square roots of the eigenvalues of $A^\dagger A$ (which is positive semi-definite and therefore has non-negative real eigenvalues.)
\\~\\
{\UPDATED
\textbf{Effective Condition Number.}
$\kappa_A$ denotes (an upper bound on) the effective condition number of $A$,
defined as the ratio of the maximum and minimum non-zero singular values of $A$.
Let $\sigmamax{A}$ be the largest singular value of $A$,
and $\sigmamin{A}$ be the smallest singular value of $A$.
Additionally, let $\sigmaminnz{A}$ be the smallest non-zero singular value of $A$.
Then

\[ \kappa_A 
\ge \frac{\sigmamax{A}}{\sigmaminnz{A}}
= \sqrt{\frac{\lambda_{\max}(A^\dagger A)}{\widetilde\lambda_{\min}(A^\dagger A)}}
\]

If $A$ is full-rank, then $\sigmaminnz{A} = \sigmamin{A}$, and $\kappa_A$ becomes the condition number of the matrix.
In this text, unless stated otherwise, we always refer to $\kappa_A$ as (an upper bound on) effective condition number of a matrix, and not the true condition number.
}
\\~\\
\textbf{Norm.} Unless otherwise specified, $\norm{A}$ denotes the spectral norm of $A$, while $\norm{A}_F$ denotes the Frobenius norm of $A$, defined as
\begin{align*}
    &\norm{A} := \max_{x \neq 0} \frac{\norm{Ax}}{\norm{x}} = \sigma_{\max}(A) \\
    &\norm{A}_F := \sqrt{\sum_{j=1}^{r} \sigma_j^2}
\end{align*}
Unless otherwise specified, when $A$ is assumed to be normalized, it is with respect to the spectral norm. 
\\~\\
\textbf{Soft-O Complexity.} Finally, we use
$f = \ohtilde{g}$ to denote $f = \order{g\cdot\mathrm{polylog}(g)}$.
\\~\\
\textbf{Controlled Unitaries}. If $U$ is a $s$-qubit unitary, then $\controlled{U}$ is a $(s+1)$-qubit unitary defined by
\begin{equation*}
    \controlled{U} = \ketbra{0} \otimes I_s + \ketbra{1} \otimes U
\end{equation*}

{\UPDATED
Throughout this text whenever we state that the time taken to implement a unitary $U_A$ is $T_A$ and the cost of an algorithm is $\order{n T_A}$, we imply that the algorithm makes $n$ uses of the unitary $U_A$. Thus, if the circuit depth of $U_A$ is $T_A$, the circuit depth of our algorithm is $\order{n T_A}$.
}

\subsection{Quantum Input Models}
\label{ssec:quantum_input_models}

The complexities of quantum algorithms often depend on how the input data is accessed. For instance, in quantum algorithms for linear algebra (involving matrix operations), it is often assumed that there exists a black-box that returns the positions of the non-zero entries of the underlying matrix when queried. The algorithmic running time is expressed in terms of the number of queries made to this black-box. Such an input model, known as the \textit{Sparse Access Model}, helps design efficient quantum algorithms whenever the underlying matrices are sparse. Various other input models exist, and quantum algorithms are typically designed and optimized for specific input models.

Kerenidis and Prakash \cite[Section 5.1]{kerenidis2016quantum} introduced a different input model, known as the \textit{quantum data structure model}, which is more conducive for designing quantum machine learning algorithms. In this model, the input data (e.g: entries of matrices) arrive online and are stored in a classical data structure (often referred to as the KP-tree in the literature), which can be queried in superposition by using a QRAM. This facilitates efficiently preparing quantum states corresponding to the rows of the underlying matrix, that can then be used for performing several matrix operations. Subsequently, several quantum-inspired classical algorithms have also been developed following the breakthrough result of Tang~\cite{tang_reco_19}. Such classical algorithms have the same underlying assumptions as the quantum algorithms designed in the data structure input model and are only polynomially slower provided the underlying matrix is low rank.

In this work, we will consider the framework of \textit{block-encoding}, wherein it is assumed that the input matrix $A$ (up to some sub-normalization) is stored in the left block of some unitary. The advantage of the block-encoding framework, which was introduced in a series of works \cite[Definition 1]{LC16}, \cite[Section 1]{CGJ19}, \cite[Section 1.3]{GSLW2019}, is that it can be applied to a wide variety of input models. For instance, it can be shown that both the sparse access input model as well as the quantum data structure input model are specific instances of block-encoded matrices \cite[Sections 2.2 and 2.4]{CGJ19}, \cite[Section 5.2]{GSLW2019}. Here we formally define the framework of block-encoding and also express the sparse access model as well as the quantum data structure model as block-encodings. We refer the reader to \cite{CGJ19, GSLW2019} for proofs.  

\begin{definition}[Block Encoding, restated from \cite{GSLW2019}, Definition 24]
    \label{def:block_encoding}
Suppose that $A$ is an $s$-qubit operator, $\alpha, \varepsilon \in \mathbb{R}^+$ and $a \in \nn $, then we say that the $(s + a)$-qubit unitary $U_A$ is an $(\alpha, a, \varepsilon)$-block-encoding of $A$, if

\begin{equation}
 \norm{ A - \alpha (\bra{0}^{\otimes a} \otimes I)U_A(\ket{0}^{\otimes a} \otimes I) } \leq \varepsilon .
\end{equation}

\end{definition}
Let $\ket{\psi}$ be an $s$-qubit quantum state. Then applying $U_A$ to $\ket{\psi}\ket{0}^{\otimes a}$ outputs a quantum state that is $\frac{\varepsilon}{\alpha}$-close to 
$$
\dfrac{A}{\alpha}\ket{\psi}\ket{0}^{\otimes a}+\ket{\Phi^{\perp}},
$$
where $\left(I_s\otimes \ketbra{0}^{\otimes a}\right)\ket{\Phi^{\perp}}=0$. Equivalently, suppose $\tilde{A} := \alpha\paren{\bra{0}^{\otimes a} \otimes I_s} U_A \paren{\ket{0}^{\otimes a} \otimes I_s}$ denotes the actual matrix that is block-encoded into $U_A$, then $\norm{A-\tilde{A}}\leq\varepsilon$. 

 In the subsequent sections, we provide an outline of the quantum data structure model and the sparse access model which are particular instances of the block encoding framework.

\subsubsection{Quantum Data Structure Input Model}
\label{sssec:q_data_structure_model}

Kerenidis and Prakash introduced a quantum accessible classical data structure which has proven to be quite useful for designing several quantum algorithms for linear algebra \cite{kerenidis2016quantum}. The classical data structure stores entries of matrices or vectors and can be queried in superposition using a QRAM (quantum random access memory). We directly state the following theorem from therein.

\begin{theorem}[Implementing quantum operators using an efficient data structure, \cite{prakash2014quantum, kerenidis2016quantum}]
    \label{thm:qrom_data_structure}
    Let $ A \in \rr^{N \times d} $, and $ w $ be the number of non-zero entries of $A $.
    Then there exists a data structure of size $ \order{ w \log^2\left( dN \right)} $ that given the matrix elements $ ( i, j, a_{ij} ) $, stores them at a cost of $ \order {\logp{dN}} $ operations per element.
    Once all the non-zero entries of $A$ have been stored in the data structure, there exist quantum algorithms that are $ \varepsilon $-approximations to the following maps:
    \begin{equation*}
        U : \ket{i}\ket{0} \mapsto \frac{1}{\norm{A_{i, \cdot}}} \sum_{j=1}^{d}a_{i, j} \ket{i,j} = \ket{\psi_i},
    \end{equation*}
    \begin{equation*}
        V : \ket{0}\ket{j} \mapsto \frac{1}{\norm{A}_F} \sum_{i=1}^{N} \norm{A_i,.} \ket{i,j}=\ket{\phi_j}
    \end{equation*}
    where $\norm{A_{i, \cdot}}$ is the norm of the $i^{\mathrm{th}}$ row of $A$ and the second register of $\ket{\psi_i}$ is the quantum state corresponding to the $i^{\mathrm{th}}$ row of $A$. These operations can be applied at a cost of $ \order{ \mathrm{polylog}(Nd / \varepsilon) } $.
\end{theorem}

It was identified in Ref.~\cite{CGJ19} that if a matrix $A$ is stored in this quantum accessible data structure, there exists an efficiently implementable block-encoding of $A$. We restate their result here.
\begin{lemma}[Implementing block encodings from quantum data structures,~\cite{CGJ19}, Theorem 4]
    \label{lem:qrom_be_constr}
     Let the entries of the matrix $ A \in \mathbb{R}^{N\times d}$ be stored in a quantum accessible data structure, then there exist unitaries $  U_R, U_L $ that can be implemented at a cost of $ \order{ \mathrm{polylog}(dN / \varepsilon) } $ such that $ U_R^{\dagger} U_L $ is a $ ( \norm{A}_F, \lceil \logp{d + N} \rceil, \varepsilon ) $-block-encoding of $ A $.
\end{lemma}
\begin{proof}
The unitaries $U_R$ and $U_L$ can be implemented via $U$ and $V$ in the previous lemma. Let $U_R=U$ and $U_L=V.\texttt{SWAP}$. Then for $s=\lceil\log (d + N) \rceil$ we have
$$
U_R:~\ket{i}\ket{0^s}\rightarrow \ket{\psi_i}, 
$$
and 
$$
U_L:~\ket{j}\ket{0^s}\rightarrow \ket{\phi_j},
$$
So we have that the top left block of $U^{\dag}_R U_L$ is 
$$
 \sum_{i=1}^{N}\sum_{j=1}^{d}\braket{\psi_i|\phi_j}\ket{i,0}\bra{j,0}
$$
Now
\begin{align*}
\braket{\psi_i|\phi_j}&=\sum_{k=1}^d\sum_{\ell=1}^{N} \dfrac{a_{ik}}{\norm{A_{i, \cdot}}}\cdot\dfrac{\norm{A_\ell}}{\norm{A}_F}\underbrace{\braket{i,k|l,j}}_{:=\delta_{i,l}.\delta_{k,j}}\\
                      &=\dfrac{a_{ij}}{\norm{A}_F}.                      
\end{align*}
Moreover since only $\varepsilon$-approximations of $U$ and $V$ can be implemented we have that $U^{\dag}_R U_L$ is a $(\norm{A}_F,\lceil \log(n+d)\rceil,\varepsilon)$ block encoding of $A$ implementable with the same cost as $U$ and $V$.
\end{proof}

In Ref.~\cite{KP2020_iterative_quantum_gradient} argued that in certain scenarios, storing the entries of $ A^{(p)}, (A^{1-p})^{\dagger} $ might be useful as compared to storing $A$, for some $ p \in [0, 1]$. In such cases, the quantum data structure is a $(\mu_p, \lceil\log(N+d)\rceil,\varepsilon)$ block encoding of $A$, where $\mu_p(A)=\sqrt{s_{2p}(A).s_{2(1-p)}(A^T)}$ such that $s_p(A) := \max_{j} \norm{A_{j, \cdot}}_q^q$. Throughout the work, whenever our results are expressed in the quantum data structure input model, we shall state our complexity in terms of $\mu_A$. When the entries of $A$ are directly stored in the data structure, $\mu_A=\norm{A}_F$. Although, we will not state it explicitly each time, our results also hold when fractional powers of $A$ are stored in the database and simply substituting $\mu_A=\mu_p(A)$, yields the required complexity.

\subsubsection{Sparse Access Input Model}
\label{sssec:sparse_access_model}

The sparse access input model considers that the input matrix $ A \in \rr^{N \times d} $ has row sparsity $ s_r $ and column sparsity $ s_c $. Furthermore, it assumes that the entries of $A$ can be queried via an oracle as

\begin{equation*}
    O_A : \ket{i} \ket{j} \ket{0}^{\otimes b} \mapsto \ket{i} \ket{j} \ket{a_{ij}} \quad \forall i \in [N], j \in [d],
\end{equation*}

and the indices of the non-zero elements of each row and column can be queried via the following oracles:

\begin{align*}
    &O_r : \ket{i} \ket{j} \mapsto \ket{i} \ket{r_{ij}} \quad \forall i \in [N], k \in [s_r],\\
    &O_c : \ket{i} \ket{j} \mapsto \ket{c_{ij}} \ket{j} \quad \forall i \in [d], k \in [s_c]
\end{align*}
where $ r_{ij} $ is the $ j^{\mathrm{th}} $ non-zero entry of the $ i^{\mathrm{th}} $ row of $ A $ and $c_{ij}$ is the $i^{\mathrm{th}}$ non-zero entry of the $j^{\mathrm{th}}$ column of $A$. Gily\'{e}n et al.~\cite{GSLW2019} showed that a block encoding of a sparse $ A $ can be efficiently prepared by using these three oracles. We restate their lemma below.

\begin{lemma}[Constructing a block-encoding from sparse-access to matrices, \cite{GSLW2019full}, Lemma 48]
    \label{lem:sparse_access_be_constr}
    Let $ A \in \rr^{N \times d} $ be an $ s_r, s_c $ row, column sparse matrix given as a sparse access input. Then for all $ \varepsilon \in (0, 1) $, we can implement a $ (\sqrt{s_c s_r}, \mathrm{polylog}(Nd / \varepsilon), \varepsilon) $-block-encoding of $ A $ with $ \order{ 1 } $ queries to $O_r, O_c, O_A$ and $ \mathrm{polylog}( Nd / \varepsilon ) $ elementary quantum gates.
\end{lemma}

Throughout the paper, we shall assume input matrices are accessible via approximate block-encodings. This also allows us to write down the complexities of our quantum algorithms in this general framework. Additionally, we state the complexities in both the sparse access input model as well as the quantum accessible data structure input model as particular cases. 

\subsection{Quantum Singular Value Transformation}
\label{ssec:qsvt}

In a seminal work, Gily\'{e}n et al. presented a framework to apply an arbitrary polynomial function to the singular values of a matrix, known as Quantum Singular Value Transformation (QSVT) \cite{GSLW2019}. QSVT is quite general: many quantum algorithms can be recast to this framework, and for several problems, better quantum algorithms can be obtained \cite{GSLW2019, Grand_Uni_2021}. In particular, QSVT has been extremely useful in obtaining optimal quantum algorithms for linear algebra. For instance, using QSVT, given the block-encoding of a matrix $A$, one could obtain $A^{-c}$ with $c\in [0,\infty)$ with optimal complexity and by using fewer additional qubits than prior art. This section briefly describes this framework, which is a generalization of Quantum Signal Processing (QSP) \cite[Section 2]{LC16}, \cite[Theorem 2]{LC16b}, \cite{LYC16}. 
The reader may refer to \cite{Grand_Uni_2021} for a more pedagogical overview of these techniques.
 
Let us begin by discussing the framework of Quantum Signal Processing. QSP is a quantum algorithm to apply a $ d $-degree bounded polynomial transformation with parity $ d \mod 2 $ to an arbitrary quantum subsystem, using a quantum circuit $ U_{\Phi} $ consisting of only controlled single qubit rotations. This is achieved by interleaving a \textit{signal rotation operator} $W$ (which is an $x$-rotation by some fixed angle $ \theta $) and a \textit{signal processing operator} $ S_{\phi} $ (which is a $z$-rotation by a variable angle $ \phi \in [0, 2 \pi] $). In this formulation, the signal rotation operator is defined as

\begin{equation}
    \label{eqn:signal_rotation_operator}
    W(x) := \begin{pmatrix}
        x & i \sqrt{1 - x^2} \\
        i \sqrt{1 - x^2} & x
    \end{pmatrix},
\end{equation}

which is an $x$-rotation by angle $ \theta = -2 \arccos(x) $, and the signal processing operator is defined as

\begin{equation}
    \label{eqn:signal_processing_operator}
    S_{\phi} := e^{i \phi Z},
\end{equation}

which is a $z$-rotation by an angle $ - 2 \phi $. Interestingly, sandwiching them together for some $ \Phi := ( \phi_0, \phi_1, \ldots \phi_d ) \in \rr^{d + 1}$, as shown in \autoref{eqn:qsp_matrix}, gives us a matrix whose elements are polynomial transformations of $ x $,

\begin{align}
    \label{eqn:qsp_matrix}
    U_{\Phi} &:= e^{i \phi_0 Z} \prod_{j = 1}^{j= d} \paren{  W(x) e^{ i \phi_j Z } } \\
             &= \begin{pmatrix}
        P(x) & i Q(x) \sqrt{1 - x^2} \\
        i Q^*(x) \sqrt{1 - x^2} & P^*(x)
    \end{pmatrix},
\end{align}

such that

\begin{enumerate}
    \item $ \deg P \leq d;\ \deg Q \leq d - 1 $,
    \item $ P(x) $ has a parity $d \mod 2 $,
    \item $ | P(x) |^2 + (1 - x^2) | Q(x) |^2 = 1 \quad \forall x \in [-1, 1] $.
\end{enumerate}

Following the application of the quantum circuit $ U_{\Phi} $ for an appropriate $ \Phi $, one can project into the top left block of $ U_{\Phi} $ to recover the polynomial $ \bra{0} U_{\Phi} \ket{0} = P(x) $. Projecting to other basis allows the ability to perform more interesting polynomial transformations, which can be linear combinations of $ P(x), Q(x) $, and their complex conjugates. For example, projecting to $ \{ \ket{+}, \ket{-} \} $ basis gives us

\begin{equation}
    \label{eqn:qsp_hadamard_basis}
    \bra{+} U_{\Phi} \ket{+} = \real (P(x)) + i \real (Q(x)) \sqrt{1 - x^2} .
\end{equation}

Quantum Signal Processing can be formally stated as follows. 

\begin{theorem}[Quantum Signal Processing, Corollary 8 from \cite{GSLW2019}]
\label{thm:qsp-main-theorem}
    Let $P \in \mathbb{C}[x]$ be a polynomial of degree $d \ge 2$, such that 
    \begin{itemize}
        \item $P$ has parity-$(d\mod 2)$,
        \item $\forall x \in [-1, 1] : \abs{P(x)} \leq 1$,
        \item $\forall x \in (-\infty, -1] \cup [1, \infty) : \abs{P(x)} \geq 1$,
        \item if $d$ is even, then $\forall x \in \rr : P(ix) P^*(ix) \geq 1$. 
    \end{itemize}
    Then there exists a $\Phi \in \rr^d$ such that 
    \begin{equation}
        \prod_{j=1}^{d} \paren{e^{i \phi_j \sigma_z} W(x)} = \begin{pmatrix}
            P(x) & \cdot \\ 
            \cdot & \cdot 
        \end{pmatrix}.
    \end{equation}
\end{theorem}

{\UPDATED
Thus, QSP allows us to implement any polynomial $P(x)$ that satisfies the aforementioned requirements. Throughout this article, we refer to any such polynomial $P(x)$ as a \textit{QSP polynomial}. Quantum Singular Value Transformation is a natural generalization of this procedure. It allows us to apply a QSP polynomial transformation to each singular value of an arbitrary block of a unitary matrix. In addition to this generalization, QSVT relies on the observation that several functions can be well-approximated by QSP polynomials. Thus, through QSVT one can transform each singular value of a block-encoded matrix by any function that can be approximated by a QSP polynomial. Since several linear algebra problems boil down to applying specific transformations to the singular values of a matrix, QSVT is particularly useful for developing fast algorithms for quantum linear algebra. Next, we introduce QSVT formally via the following theorem.
}

\begin{theorem}[Quantum Singular Value Transformation \cite{GSLW2019full}, Section~3.2]
\UPDATED
\label{thm:qsvt-def}
Suppose $A\in \mathbb{R}^{N\times d}$ is a matrix with singular value decomposition $A=\sum_{j=1}^{d_{\min}}\sigma_j\ket{v_j}\bra{w_j}$, where $d_{\min}=\min\{N,d\}$ and $\ket{v_j}$ $(\ket{w_j})$ is the left (right) singular vector with singular value $\sigma_j$. Furthermore, let $U_A$ be a unitary such that $A=\widetilde{\Pi}U_A\Pi$, where $\Pi$ and $\widetilde{\Pi}$ are orthogonal projectors. Then, for any QSP polynomial $P(x)$ of degree $n$, there exists a vector $\Phi=\left(\phi_1,\phi_2,\cdots\phi_n\right)\in\mathbb{R}^n$ and a unitary
\begin{equation}
    \label{eqn:qsvt_sequence}
    U_{\Phi} = \begin{cases}
        e^{i \phi_1 ( 2 \widetilde{\Pi} - I )} U_A \left[ \prod_{k = 1}^{(n - 1) / 2} e^{i \phi_{2k} ( 2 \widetilde{\Pi} - I )} U_A^{\dagger} e^{i \phi_{2k+1} ( 2 \widetilde{\Pi} - I )} U_A \right], & n \text{ is odd } \\
        \left[ \prod_{k = 1}^{n/ 2} e^{i \phi_{2k-1} ( 2 \widetilde{\Pi} - I )} U_A^{\dagger} e^{i \phi_{2k} ( 2 \widetilde{\Pi} - I )} U_A \right], & n \text{ is even},
    \end{cases}
\end{equation}
such that
\begin{equation}
    \label{eqn:qsvt_sequence_projected}
    P^{SV} (A) = \begin{cases}
        \widetilde{\Pi} U_{\Phi} \Pi, & n \text{ is odd} \\
        \Pi U_{\Phi} \Pi, & n \text{ is even},
    \end{cases}
\end{equation}
where $ P^{SV} (A) $ is the polynomial transformation of the matrix $ A $ defined as
\begin{equation}
    \label{eqn:poly_sv_transformation_of_matrix}
    P^{SV} (A) := \begin{cases}
        \sum_j P(\sigma_j) \ket{v_j} \bra{w_j}, & P \text{ is odd} \\
        \sum_j P(\sigma_j) \ket{w_j} \bra{w_j}, & P \text{ is even}.
    \end{cases}
\end{equation}
\end{theorem}

{\UPDATED
Theorem \ref{thm:qsvt-def} tells us that for any QSP polynomial $P$ of degree $n$, we can implement $P^{SV}(A)$ using one ancilla qubit, $\Theta(n)$ applications of $U_A$, $U^{\dag}_A$ and controlled reflections $I-2\Pi$ and $I-2\widetilde{\Pi}$. Furthermore, if in some well-defined interval, some function $f(x)$ is well approximated by an $n$-degree QSP polynomial $P(x)$, then Theorem \ref{thm:qsvt-def} also allows us to implement a transformation that approximates $f(A)$, where
}
\begin{equation}
\UPDATED
    \label{eqn:sv_transformation_of_matrix}
     f(A) := \begin{cases}
        \sum_j f(\sigma_j) \ket{v_j} \bra{w_j}, & P\text{ is odd} \\
        \sum_j f(\sigma_j) \ket{w_j} \bra{w_j}, & P \text{ is even}.
    \end{cases}
\end{equation}
{\UPDATED
The following theorem from Ref.~\cite{GSLW2019full} deals with the robustness of the QSVT procedure, i.e. how errors propagate in QSVT. In particular, for two matrices $A$ and $\tilde{A}$, it shows how close their polynomial transformations ($P^{SV}(A)$ and $P^{SV}(\widetilde{A})$, respectively) are, as a function of the distance between $A$ and $\tilde{A}$.
}

\begin{lemma}[Robustness of Quantum Singular Value Transformation, \cite{GSLW2019full}, Lemma~23]
    \label{lem:robustness_of_QSVT}
    Let $P \in \mathbb{C}[x]$ be a QSP polynomial of degree $n$. Let $A, \tilde{A} \in \mathbb{C}^{N \times d}$ be matrices of spectral norm at most 1, such that 
    \begin{equation*}
        \norm{A - \tilde{A}} + \norm{\frac{A + \tilde{A}}{2}}^2 \leq 1.
    \end{equation*}
    Then, 
    \begin{equation*}
        \norm{P^{SV}(A) - P^{SV}(\tilde{A})} \leq n \sqrt{\frac{2}{1 - \norm{\frac{A + \tilde{A}}{2}}^2}} \norm{A - \tilde{A}}. 
    \end{equation*}
\end{lemma}
{\UPDATED
We will apply this theorem to develop a robust version of QSVT. More precisely, in order to implement QSVT, we require access to a unitary $U_A$, which is a block-encoding of some matrix $A$. This block-encoding, in most practical scenarios, is not perfect: we only have access to a $\varepsilon$-approximate block-encoding of $A$. If we want an $\delta$-accurate implementation of $P^{SV}(A)$, how precise should the block-encoding of $A$ be? Such a robustness analysis has been absent from prior work involving QSVT and will allow us to develop robust versions of a number of quantum algorithms in subsequent sections.  The following theorem determines the precision $\varepsilon$ required in the block-encoding of $A$ in terms of $n$, the degree of the QSP polynomial that we wish to implement and $\delta$, the accuracy of $P^{SV}(A)$.
}
\begin{theorem}[Robust QSVT]
    \label{thm:robust_qsvt}
    \UPDATED
    Let $P \in \mathbb{C}[x]$ be a QSP polynomial of degree $n \ge 2$.
    Let $\delta \in [0, 1]$ be the precision parameter.
    Let $U$ be an $(\alpha, a, \varepsilon)$-block-encoding of matrix $A \in \cc^{N \times d}$
    satisfying $\norm{A} \le \alpha/2$,
    implemented in cost $T$
    for some $\varepsilon \le {\alpha\delta}/{2n}$.
    Then we can construct a $(1, a + 1, \delta)$-block-encoding of $P(A/\alpha)$
    in cost $\order{nT}$.
\end{theorem}
\begin{proof}
\UPDATED
    Let $\tilde A$ be the encoded block of $U$, then $\norm{A - \tilde A} \le \varepsilon$.
    Applying QSVT on $U$ with the polynomial $P$, we get a block-encoding for $P(\tilde A / \alpha)$,
    with $\order{n}$ uses of $U, U^\dagger$, and as many multiply-controlled NOT gates. 
    Observe that $\norm{\frac{A}{\alpha} - \frac{\tilde{A}}{\alpha}} \le \frac{\varepsilon}{\alpha} \le \frac{\delta}{2n} \le \frac14$, and,
    
    \begin{align*}
        \norm{\frac{\frac{A}{\alpha} + \frac{\tilde A}{\alpha}}{2}}^2
        = \norm{\frac{A}{\alpha} + \frac{\tilde A - A}{2\alpha}}^2
        \le \paren{\frac{\norm{A}}{\alpha} + \frac{\norm{\tilde A - A}}{2\alpha}}^2
        \le \paren{\frac12 + \frac18}^2
        \le \frac12
    \end{align*}
    Therefore the error in the final block-encoding is given by invoking \autoref{lem:robustness_of_QSVT}
    with matrices $A/\alpha, \tilde{A}/\alpha$:
    \begin{align*}
        \norm{P\paren{\frac A \alpha} - P\paren{\frac {\tilde{A}} \alpha}}
        \le n \sqrt{\frac2{1 - \frac12}}~\frac\varepsilon\alpha
        = \frac{2n\varepsilon}{\alpha}
        \le \delta.
    \end{align*}
\end{proof}
{\UPDATED
In Section \ref{sec:algorithmic_primitives}, we will make use of Theorem \ref{thm:robust_qsvt}, to develop robust quantum algorithms for singular value discrimination, variable-time matrix inversion, positive and negative powers of matrices. Subsequently, in Sec.~\ref{sec:main_proof}, we shall combine algorithmic primitives to design robust quantum regularized least squares algorithms.
}
\subsection{Variable Time Amplitude Amplification}
\label{ssec:vtaa}

Ambainis \cite{ambainis2010variable} defined the notion of a \textit{variable-stopping-time quantum algorithm} and formulated the technique of \textit{Variable Time Amplitude Amplification} (VTAA), a tool that can be used to amplify the success probability of a variable-stopping-time quantum algorithm to a constant by taking advantage of the fact that computation on some parts of an algorithm can complete earlier than on other parts. The key idea here is to look at a quantum algorithm $\mathcal{A}$ acting on a state $\ket{\psi}$ as a combination of $m$ quantum sub-algorithms $\mathcal{A} = \mathcal{A}_m \cdot \mathcal{A}_{m-1} \cdot \ldots \mathcal{A}_1$, each acting on $\ket{\psi}$ conditioned on some ancilla flag being set. Formally, a variable stopping time algorithm is defined as follows

\begin{definition}[Variable-stopping-time Algorithm, \cite{ambainis2010variable}]
    A quantum algorithm $\mathcal{A}$ acting on $\mathcal{H}$ that can be written as $m$ quantum sub-algorithms, $\mathcal{A} = \mathcal{A}_m \cdot \mathcal{A}_{m-1} \cdot \ldots \mathcal{A}_1$ is called a variable stopping time algorithm if $\mathcal{H} = \mathcal{H}_C \otimes \mathcal{H}_{\mathcal{A}}$, where $\mathcal{H}_C = \otimes_{i=1}^{m} \mathcal{H}_{C_i}$ with $\mathcal{H}_{C_i} = \mathrm{span}(\ket{0}, \ket{1})$, and each unitary $\mathcal{A}_j$ acts on $\mathcal{H}_{C_j} \otimes \mathcal{H}_{\mathcal{A}}$ controlled on the first $j-1$ qubits $\ket{0}^{\otimes j  -1} \in \otimes_{i=1}^{j-1} \mathcal{H}_{C_i}$ being in the all zero state.
\end{definition}

Here $\mathcal{H}_{C_i}$ is a single qubit clock register. In VTAA, $\mathcal{H}_\mathcal{A}$ has a flag space consisting of a single qubit to indicate success, $\mathcal{H}_\mathcal{A} = \mathcal{H}_F \otimes \mathcal{H}_W$. Here $\mathcal{H}_F = \mathrm{Span}(\ket{g}, \ket{b})$ flags the good and bad parts of the run. Furthermore, for $1\leq i\leq m$, define the stopping times $t_i$ such that $t_1 < t_2 < \cdots t_m=T_{\max}$, such that the algorithm $\mathcal{A}_j\mathcal{A}_{j-1}\cdots \mathcal{A}_1$ having (gate/query) complexity $t_i$ halts with probability
\begin{equation*}
\label{eq:stopping-prob}
p_j=\norm{\Pi_{C_j}\mathcal{A}_j\mathcal{A}_{j-1}\cdots \mathcal{A}_1\ket{0}_{\mathcal{H}}}^2,
\end{equation*}
where $\ket{0}_{\mathcal{H}}\in\mathcal{H}$ is the all zero quantum state and $\Pi_{C_j}$ is the projector onto $\ket{1}$ in $\mathcal{H}_{C_j}$. From this one can define the average stopping time of the algorithm $\mathcal{A}$ defined as
\begin{equation*}
\label{eq:avg-prob}
\norm{T}_2=\sqrt{\sum_{j=1}^{m}p_j t^2_j}.
\end{equation*}
For a variable stopping time algorithm if the average stopping time $\norm{T}_2$ is less than the maximum stopping time $T_{\max}$, VTAA can amplify the success probability $(p_{\mathrm{succ}})$ much faster than standard amplitude amplification. In this framework, the success probability of $\mathcal{A}$ is given by
\begin{equation*}
\label{successprob-vtaa}
p_{\mathrm{succ}}=\norm{\Pi_F \mathcal{A}_m\mathcal{A}_{m-1}\cdots \mathcal{A}_1\ket{0}_{\mathcal{H}}}^2
\end{equation*}
While standard amplitude amplification requires time scaling as $\order{T_{\max}/\sqrt{p_{\mathrm{succ}}}}$, the complexity of VTAA is more involved. Following \cite{CGJ19}, we define the complexity of VTAA as follows.
\begin{lemma}[Efficient variable time amplitude amplification \cite{CGJ19full}, 
Theorem~23]
    \label{lem:variable_time_aa}
    Let $U$ be a state preparation unitary such that $U \ket{0}^{\otimes k} = \sqrt{p_{\mathrm{prep}}}\ket{0}\ket{\psi_0} + \sqrt{1 - p_{\mathrm{prep}}}\ket{1}\ket{\psi_1}$ that has a query complexity $T_U$. And let $\mathcal{A}=\mathcal{A}_m\mathcal{A}_{m-1}\cdots \mathcal{A}_1$ be a variable stopping time quantum algorithm that we want to apply to the state $\ket{\psi_0}$, with the following known bounds: $p_{\mathrm{prep}} \geq p'_{\mathrm{prep}}$ and $p_{\mathrm{succ}} \geq p'_{\mathrm{succ}}$. Define $T'_{\max} := 2T_{\max}/t_1$ and $$
    Q := \paren{T_{\max} + \frac{T_U + k}{\sqrt{p_{\mathrm{prep}}}}} \sqrt{\logp{T'_{\max}}} + \frac{\paren{\norm{T}_2 + \frac{T_U + k}{\sqrt{p_{\mathrm{prep}}}}}\logp{T'_{\max}}}{\sqrt{p_{\mathrm{succ}}}}.
    $$
    Then with success probability $\geq 1 - \delta$, we can create a variable-stopping time algorithm $\mathcal{A}'$ that prepares the state $a \ket{0} \mathcal{A}' \ket{\psi_0} + \sqrt{1 - a^2} \ket{1} \ket{\psi_{\textrm{garbage}}}$, such that $a = \Theta(1)$ is a constant and $\mathcal{A}'$ has the complexity $\order{Q}$. 
\end{lemma}

One cannot simply replace standard amplitude amplification with VTAA to boost the success probability of a quantum algorithm. A crucial task would be to recast the underlying algorithm in the VTAA framework. We will be applying VTAA to the quantum algorithm for matrix inversion by QSVT. So, first of all, in order to apply VTAA to the algorithm must be first recast into a variable-time stopping algorithm so that VTAA can be applied.

Originally, Ambainis~\cite{ambainis2010variable} used VTAA to improve the running time of the HHL algorithm from $\order{\kappa^2\log N}$ to $\order{\kappa\log^3\kappa\log N}$. Childs et al.~\cite{CKS17} designed a quantum linear systems algorithm with a polylogarithmic dependence on the accuracy. Additionally, they recast their algorithm into a framework where VTAA could be applied to obtain a linear dependence on $\kappa$. Later Chakraborty et al.~\cite{CGJ19} modified Ambainis' VTAA algorithm to perform variable time amplitude estimation.

In this work, to design quantum algorithms for $\ell_2$-regularized linear regression, we use a quantum algorithm for matrix inversion by QSVT. We recast this algorithm in the framework of VTAA to achieve nearly linear dependence in $\kappa$ (the effective condition number of the matrix to be inverted). Using QSVT instead of controlled Hamiltonian simulation improves the complexity of the overall matrix inversion algorithm (QSVT and VTAA) by a log factor. It also reduces the number of additional qubits substantially. Furthermore, we replace a gapped quantum phase estimation procedure with a more efficient quantum singular value discrimination algorithm using QSVT. This further reduces the number of additional qubits by $O(\log^2(\kappa/\delta))$ than in Refs.~\cite{CKS17,CGJ19}, where $\kappa$ is the condition number of the underlying matrix and $\delta$ is the desired accuracy. The details of the variable stopping time quantum algorithm for matrix inversion by QSVT are laid out in \autoref{ssec:mi_using_qsvt}.

\section{Algorithmic Primitives}
\label{sec:algorithmic_primitives}

This section introduces the building blocks of our quantum algorithms for quantum linear regression with general $\ell_2$-regularization. As mentioned previously, we work in the block-encoding framework. 
We develop robust quantum algorithms for arithmetic operations, inversion, and positive and negative powers of matrices using quantum singular value transformation, assuming we have access to approximate block-encodings of these matrices. While some of these results were previously derived assuming perfect block-encodings \cite{GSLW2019, CGJ19}, we calculate the precision required in the input block-encodings to output a block-encoding or quantum state arbitrarily close to the target. 



Given a $(\alpha,a,\varepsilon)$-block-encoding of a matrix $A$, we can efficiently amplify the sub-normalization factor from $\alpha$ to a constant and obtain an amplified block-encoding of $A$. For our quantum algorithms in Sec.~\ref{sec:main_proof}, we show working with pre-amplified block-encodings often yields better complexities. We state the following lemma which was proven in Ref.~\cite{low2017hamiltonian}:

\begin{lemma}[Uniform Block Amplification of Contractions, \cite{low2017hamiltonian}]
  \label{lem:uniform_block_ampl_contractions}
  Let $A \in \rr^{N \times d}$ such that $\norm{A} \le 1$
  If $\alpha \ge 1$ and
  $U$ is a $(\alpha, a, \varepsilon)$-block-encoding of A
  that can be implemented at a cost of $T_U$,
  then there is a $(\sqrt 2, a + 1, \varepsilon + \gamma)$-block-encoding of A
  that can be implemented at a cost of $\order{\alpha T_U \log\paren{1/\gamma}}$.
\end{lemma}

\begin{restatable}[Uniform Block Amplification]{corollary}{thmbodyUniformBlockAmp}
  \label{lem:uniform_block_ampl}
  Let $A \in \rr^{N \times d}$ and $\delta \in (0, 1]$.
  Suppose $U$ is a $(\alpha, a, \varepsilon)$-block-encoding of A,
  such that $\varepsilon \leq \frac{\delta}{2}$,
  that can be implemented at a cost of $T_U$.
  Then a $(\sqrt 2 \norm{A}, a + 1, \delta)$-block-encoding of A can be implemented
  at a cost of $\order{\frac{\alpha T_U}{\norm{A}} \log\paren{\norm{A}/\delta}}$.
\end{restatable}

We now obtain the complexity of applying a block-encoded matrix to a quantum state, which is a generalization of a lemma proven in Ref.~\cite{CGJ19}.

\begin{restatable}[Applying a Block-encoded Matrix on a Quantum State]{lemma}{thmbodyApplyBlockEnc}
  \label{lem:apply_block_enc}
  Let $A$ be an $s$-qubit operator such that its non-zero singular values lie in $[\norm{A}/\kappa,\norm{A}]$. Also let $\delta \in (0, 1)$, and $U_A$ be an $(\alpha, a, \varepsilon)$-block-encoding of $A$, implementable in time $T_A$, such that 
    $$
    \varepsilon \leq \frac{\delta\norm{A}}{2\kappa}.
    $$
  Furthermore, suppose $\ket b$ be an $s$-qubit quantum state, prepared in time $T_b$.
  Then we can prepare a state that is $\delta$-close to
  $\frac{A\ket b}{\norm{A\ket b}}$
  with success probability $\Omega \paren{1}$  at a cost of
  \begin{equation*}
      \order{\frac{\alpha\kappa}{\norm{A}}(T_A + T_b)}
  \end{equation*}
\end{restatable}

\begin{restatable}[Applying a pre-amplified Block-encoded Matrix on a Quantum State]{corollary}{thmbodyApplyBlockEncPreamp}
  \label{lem:apply_block_enc_preamp}
  Let $A$ be an $s$-qubit operator such that its non-zero singular values lie in $[\norm{A}/\kappa,\norm{A}]$. Also let $\delta \in (0, 1)$, and $U_A$ be an $(\alpha, a, \varepsilon)$-block-encoding of $A$, implementable in time $T_A$, such that 
    $$
    \varepsilon \leq \frac{\delta\norm{A}}{4\kappa}.
    $$
  Furthermore, suppose $\ket b$ be an $s$-qubit quantum state that can be prepared in time $T_b$.
  Then we can prepare a state that is $\delta$-close to
  $\frac{A\ket b}{\norm{A\ket b}}$
  with success probability $\Omega \paren{1}$ at a cost of
  \begin{equation*}
     \order{\frac{\alpha\kappa}{\norm{A}}\logp{\frac{\kappa}{\delta}} T_A + \kappa T_b}
  \end{equation*}
\end{restatable}

Now, it may happen that $U_b$ prepares a quantum state that is only $\varepsilon$-close to the desired state $\ket{b}$. In such cases, we have the following lemma
\begin{restatable}[Robustness of state preparation]{lemma}{thmbodyApplyBlockEncApprox}
  \label{lem:apply_block_enc_approx}
  Let $A$ be an $s$-qubit operator such that its non-zero singular values lie in $\left[\norm{A}/\kappa,\norm{A}\right]$. Suppose $\ket{b'}$ is a quantum state that is $\varepsilon/2\kappa$-close to $\ket b$
  and $\ket\psi$ is a quantum state that is $\varepsilon/2$-close to
  $A\ket{b'}/\norm{A\ket{b'}}$. Then we have that $\ket\psi$ is $\varepsilon$-close to $A\ket{b}/\norm{A\ket{b}}$.
\end{restatable}
The proofs for \autoref{lem:uniform_block_ampl}, \autoref{lem:apply_block_enc}, \autoref{lem:apply_block_enc_preamp}, and \autoref{lem:apply_block_enc_approx} can be found in \autoref{app:algo_primitives}.

\subsection{Arithmetic with Block-Encoded Matrices}
\label{subsec:arith-w-block}

The block-encoding framework embeds a matrix on the top left block of a larger unitary $U$. It has been demonstrated that this framework allows us to obtain sums, products, linear combinations of block-encoded matrices. This is particularly useful for solving linear algebra problems in general. Here, we state some of the arithmetic operations on block-encoded matrices that we shall be using in order to design the quantum algorithms of \autoref{sec:main_proof} and tailor existing results to our requirements. 

First we prove a slightly more general form of linear combination of unitaries in the block-encoding framework, presented in \cite{GSLW2019}. To do this we assume that we are given optimal state preparation pairs, defined as follows.

\begin{definition}[Optimal State Preparation Unitary]
    \label{def:constr_opt_spu}
    Let $m \in \mathbb{Z}^+$,
    and $s = \lceil\log{m}\rceil$.
    Let $\eta \in \mathbb{R}^m_+$.
    Then we call a $s$-qubit unitary $P$ a $\eta$ state-preparation unitary if
    \[P\ket{0} = \frac{1}{\sum_j \eta_j}\sum_j \sqrt{\eta_j} \ket{j}\]
\end{definition}

\begin{restatable}[Linear Combination of Block Encoded Matrices]{lemma}{thmbodyConstrLincomb}
    \label{lem:constr_lincomb}
    For each $j \in \{0, \ldots, m-1\}$,
      let $A_j$ be an $s$-qubit operator,
      and $y_j \in \mathbb{R}^+$.
      Let $U_j$ be a $(\alpha_j, a_j, \varepsilon_j)$-block-encoding of $A_j$,
        implemented in time $T_j$.
    Define the matrix $A = \sum_j y_j A_j$,
      and the vector $\eta \in \rr^m$ s.t. $\eta_j = y_j\alpha_j$.
    Let $U_\eta$ be a $\eta$ state-preparation unitary,
      implemented in time $T_\eta$.
    Then we can implement a 
    $$\paren{\sum_j y_j \alpha_j ,\max_j(a_j) + s, \sum_j y_j \varepsilon_j}$$
    block-encoding
    of $A$
    at a cost of $\order{\sum_j T_j + T_\eta}$.
\end{restatable}

The proof is similar to the one in Ref.~\cite{GSLW2019}, with some improvements to the bounds. The detailed proof can be found in \autoref{app:algo_primitives}.
We now specialize the above lemma for the case where we need a linear combination of just two unitaries. This is the case used in this work, and we obtain a better error scaling for this by giving an explicit state preparation unitary.

\begin{corollary}[Linear Combination of Two Block Encoded Matrices]
    \label{lem:constr_lincomb_two}
    For $j \in \{0, 1\}$,
    let $A_j$ be an $s$-qubit operator
    and $y_j \in \mathbb{R}^+$.
    Let $U_j$ be a $(\alpha_j, a_j, \varepsilon_j)$-block-encoding of $A_j$,
      implemented in time $T_j$.
    Then we can implement a $(y_0 \alpha_0 + y_1 \alpha_1, 1 + \max(a_0, a_1), y_0 \varepsilon_0 + y_1 \varepsilon_1)$ encoding of $y_0 A_0 + y_1 A_1$
    in time $\order{T_0 + T_1}$.
\end{corollary}

\begin{proof}

Let $\alpha = y_0 \alpha_0 + y_1 \alpha_1$
and $P = \frac{1}{\sqrt{\alpha}}
\begin{pmatrix} \sqrt{y_0 \alpha_0} & -\sqrt{y_1 \alpha_1} \\
\sqrt{y_1 \alpha_1} & \sqrt{y_0 \alpha_0} \end{pmatrix}$.
By \autoref{def:constr_opt_spu}, 
$P$ is a $\{y_0\alpha_0, y_1\alpha_1\}$ state-preparation-unitary.
Invoking \autoref{lem:constr_lincomb} with $P$,
we get the required unitary.

\end{proof}

Given block-encodings of two matrices $A$ and $B$, it is easy to obtain a block-encoding of $AB$. 

\begin{lemma}[Product of Block Encodings, \cite{GSLW2019full}, Lemma 53]
     \label{lem:prod_of_be}
    If $U_A$ is an $(\alpha, a, \delta)$-block-encoding
      of an $s$-qubit operator $A$
      implemented in time $T_A$,
    and $U_B$ is a $(\beta, b, \varepsilon)$-block-encoding
      of an $s$-qubit operator $B$
      implemented in time $T_B$,
    then $(I^{\otimes b} \otimes U_A) (I^{\otimes a} \otimes U_B)$
    is an $(\alpha\beta, a+b, \alpha\varepsilon + \beta\delta)$-block-encoding of $AB$
      implemented at a cost of $\order{T_A + T_B}$.
\end{lemma}
Directly applying \autoref{lem:prod_of_be} results in a block-encoding of $\frac{AB}{\alpha\beta}$. If $\alpha$ and $\beta$ are large, then the sub-normalization factor $\alpha\beta$ might incur an undesirable overhead to the cost of the algorithm that uses it. In many cases, the complexity of obtaining products of block-encodings can be improved if we first amplify the block-encodings (using Lemma \ref{lem:uniform_block_ampl}) and then apply \autoref{lem:prod_of_be}. We prove the following lemma:

\begin{lemma}[Product of Amplified Block-Encodings]
     \label{lem:prod_of_be_preamp}
    Let $\delta \in (0, 1]$. If $U_A$ is an $(\alpha_A, a_A, \varepsilon_A)$-block-encoding
      of an $s$-qubit operator $A$
      implemented in time $T_A$,
    and $U_B$ is a $(\alpha_B, a_B, \varepsilon_B)$-block-encoding
      of an $s$-qubit operator $B$
      implemented in time $T_B$,
      such that $\varepsilon_A \leq \frac{\delta}{4\sqrt{2}\norm{B}}$ and $\varepsilon_B \leq \frac{\delta}{4\sqrt{2}\norm{A}}$.
    Then we can implement a
    $(2\norm{A}\norm{B}, a_A+a_B+2, \delta)$-block-encoding of $AB$
      implemented at a cost of 
      $$\order{\paren{\frac{\alpha_A}{\norm{A}} T_A + \frac{\alpha_B}{\norm{B}} T_B}\log\paren{\frac{\norm{A}\norm{B}}{\delta}}}.$$
\end{lemma}

\begin{proof}
  Using \autoref{lem:uniform_block_ampl}
    for some $\delta_A \ge 2\varepsilon_A$ we get a
    $(\sqrt 2 \norm{A}, a_A + 1, \delta_A)$-block-encoding of $A$
    at a cost of 
    $$\order{\frac{\alpha_A T_A}{\norm{A}} \log\paren{\norm{A}/\delta_A}}.$$
  Similarly for some $\delta_B \ge 2\varepsilon_B$ we get a
    $(\sqrt 2 \norm{B}, a_B + 1, \delta_B)$-block-encoding of $B$
    at a cost of 
    $$\order{\frac{\alpha_B T_B}{\norm{B}} \log\paren{\norm{B}/\delta_B}}.$$
      Now using \autoref{lem:prod_of_be} we get 
    a $(2, a_A+a_B+2, \sqrt{2}\paren{\norm{A}\delta_B + \norm{B}\delta_A})$-block-encoding of $AB$.
  We can choose
    $\delta_A = \frac{\delta}{2\sqrt 2\norm{B}}$
    and 
    $\delta_B = \frac{\delta}{2\sqrt 2\norm{A}}$
  which bounds the final block-encoding error by $\delta$.
\end{proof}
{\UPDATED
Observe that we have assumed that $A$ and $B$ are $s$-qubit operators. For any two matrices of dimension ${N\times d}$ and $d \times K$, such that $N,d, K \leq 2^s$, we can always pad them with rows and columns of zero entries and convert them to $s$-qubit operators. Thus, in the scenario where $A$ and $B$ are not $s$-qubit operators, one can consider block encodings of padded versions of these matrices. Note that this does not affect the operations on the sub-matrix blocks encoding $A$ and $B$. Thus, the above results can be used to perform block-encoded matrix products for arbitrary (compatible) matrices.
} 

Next we show how to find the block encoding of tensor product of matrices from their block encodings. This procedure will be useful in creating the dilated matrices required for regularization. The proof can be found in \autoref{app:algo_primitives}. 

\begin{restatable}[Tensor Product of Block Encoded Matrices]{lemma}{thmbodyBlockEncTensor}
\label{lemma:block_enc_tensor}
  Let $U_1$ and $U_2$ be
  $(\alpha, a, \varepsilon_1)$ and $(\beta, b, \varepsilon_2)$-block-encodings
  of $A_1$ and $A_2$, $s$ and $t$-qubit operators,
  implemented in time $T_1$ and $T_2$ respectively.
  Define $S := \Pi_{i=1}^s \texttt{SWAP}_{a+b+i}^{a+i}$. 
  Then, $S (U_1 \otimes U_2) S^\dagger$ is an $(\alpha\beta, a+b, \alpha\varepsilon_2 + \beta\varepsilon_1 +\varepsilon_1\varepsilon_2)$ block-encoding of $A_1 \otimes A_2$,
    implemented at a cost of $\order{T_1 + T_2}$.
\end{restatable}

We will now use Lemma \ref{lemma:block_enc_tensor} to augment one matrix into another, given their approximate block-encodings.

\begin{lemma}[Block-encoding of augmented matrix]
  \label{lem:augmented_regression_matrix}
If $U_A$ is an $(\alpha_A, a_A,\varepsilon_A)$-block encoding of an $s$-qubit operator $A$ that can be implemented in time $T_A$ and $U_B$ is an $(\alpha_B, a_B,\varepsilon_B)$-block encoding of an $s$-qubit operator $B$ that can be implemented in time $T_B$, then we an implement an
$(\alpha_A + \alpha_B,
    \max(a_A, a_B) + 2,
    \varepsilon_A + \varepsilon_B)$-block-encoding of
  \begin{equation*}
    A_B = \begin{pmatrix}A&0\\ B&0\end{pmatrix}
  \end{equation*}

at a cost of $\order{T_A + T_B}$.
\end{lemma}

\begin{proof}
    Let $M_A = \begin{pmatrix}
        1 & 0 \\
        0 & 0
    \end{pmatrix}$.
    Then the \texttt{SWAP} gate is a $(1, 1, 0)$ block encoding of $M_A$.
    By \autoref{lemma:block_enc_tensor},
    we can implement $U_A'$, an $(\alpha_A, a_A + 1, \varepsilon_A)$-block-encoding
    of their tensor product $M_A \otimes A = \begin{pmatrix}
        A & 0 \\
        0 & 0
    \end{pmatrix}$
    at a cost of $\order{T_A}$.
    Similarly, Let $M_B = \begin{pmatrix}
        0 & 0 \\
        1 & 0
    \end{pmatrix}$.
    Then $ (I\otimes X)\cdot \texttt{SWAP}$ is a $(1, 1, 0)$-block-encoding of $M_B$.
    Similarly \autoref{lemma:block_enc_tensor},
    we can implement $U_B'$, an $(\alpha_B, a_B + 1, \varepsilon_B)$-block-encoding
    of $M_B \otimes B = \begin{pmatrix}
        0 & 0 \\
        B & 0
    \end{pmatrix}$
    at a cost of $\order{T_B}$.
    We add them by using \autoref{lem:constr_lincomb_two} on $U_A'$ and $U_B'$,
    to implement $U_{A_B}$,
    an $(\alpha_A + \alpha_B, 2 + \max(a_A,a_B), \varepsilon_A + \varepsilon_B)$-block-encoding of
    $A_B = \begin{pmatrix}A&0\\ B&0\end{pmatrix}$.
    This can be implemented at a cost of $\order{T_A + T_B}$.
\end{proof}

\subsection{Robust Quantum Singular Value Discrimination}
\label{subsec:qevd}

The problem of deciding whether the eigenvalues of a Hamiltonian lie above or below a certain threshold, known as \textit{eigenvalue discrimination}, finds widespread applications. For instance, the problem of determining whether the ground energy of a generic local Hamiltonian is $<\lambda_a$ or $>\lambda_b$ is known to be QMA-Complete \cite{kempe2006complexity}. Nevertheless, quantum eigenvalue discrimination has been useful in preparing ground states of Hamiltonians. Generally, a variant of quantum phase estimation, which effectively performs a projection onto the eigenbasis of the underlying Hamiltonian, is used to perform eigenvalue discrimination \cite{ge2019faster}. Recently, it has been shown that QSVT can be used to approximate a projection onto the eigenspace of an operator by implementing a polynomial approximation of the \textit{sign function} \cite{lin2020near}. This was then used to design improved quantum algorithms for ground state preparation.

{\UPDATED
In our work, we design a more general primitive, known as \textit{Quantum Singular Value Discrimination} (QSVD). Instead of eigenvalues, the algorithm distinguishes whether a singular value $\sigma$ is $\le \sigma_a$ or $\ge \sigma_b$. This is particularly useful when the block-encoded matrix is not necessarily Hermitian and hence, may not have well-defined eigenvalues. We use this procedure to develop a more space-efficient variable stopping time matrix inversion algorithm in \autoref{ssec:mi_using_qsvt}. Owing to the widespread use of singular values in a plethora of fields, we believe that our QSVD procedure is of independent interest.
}

Let us define the \textit{sign function} $\mathrm{sign} : \rr \to \rr$ as follows:

\begin{equation}
    \mathrm{sign} (x) = \begin{cases}
        -1 &\quad x < 0 \\
        0  &\quad x = 0 \\
        1  &\quad x > 0. \\
    \end{cases}
\end{equation}

Given a threshold singular value $c$, 
Low and Chuang \cite{low2017hamiltonian} showed that there exists a polynomial approximation to $\mathrm{sign}(c - x)$ (based on its approximation of the \textit{erf function}). We use the result of Ref.~\cite{Grand_Uni_2021}, where such a polynomial of even parity was considered.
{\UPDATED
This is crucial, as for even polynomials, QSVT maps right (left) singular vectors to right (left) singular vectors, which enables us to use the polynomial in \cite{Grand_Uni_2021} for singular value discrimination.
}

\begin{lemma}[Polynomial approximation to the sign function \cite{low2017hamiltonian,low2017quantum,Grand_Uni_2021}]
\label{lem:approx-sign}
For any $\varepsilon, \Delta, c \in (0,1)$, there exists an efficiently computable even polynomial $P_{\varepsilon,\Delta,c}(x)$ of degree $l=\order{\frac{1}{\Delta}\log(1/\varepsilon)}$ such that 
\begin{itemize}
\item[1.~] $\forall x\in [0,1] \colon \abs{P_{\varepsilon,\Delta,c}(x)}\leq 1$\\
\item[2.~] $\forall x \in [0, 1] \setminus \paren{c - \frac{\Delta}{2}, c + \frac{\Delta}{2}} \colon  \abs{P_{\varepsilon, \Delta, c}(x)-\mathrm{sign}(c - x)} \leq \varepsilon$
\end{itemize}
\end{lemma}

Therefore, given a matrix $A$ with singular values between $[0,1]$,
we can use QSVT to implement $P_{\varepsilon,\Delta,c}(A)$ which correctly distinguishes
between singular values of $A$ whose value is less than $c-\Delta/2$
and those whose value is greater than $c+\Delta/2$. 
For our purposes, we shall consider that we are given $U_A$,
which is an $(\alpha,a,\varepsilon)$ block-encoding of a matrix $A$.
Our goal would be to distinguish whether a certain singular value $\sigma$
satisfies $0\le \sigma \le\varphi$ or $2\varphi\le \sigma \le 1$.
Since $U_A$ (approximately) implements $A/\alpha$,
the task can be rephrased as distinguishing whether a singular value of $A/\alpha$ is in $[0,\varphi/\alpha]$ or in $[2\varphi/\alpha,1]$.
For this, we develop a robust version of quantum singular value discrimination ($QSVD(\phi,\delta)$), which indicates the precision $\varepsilon$ required to commit an error that is at most $\delta$.

\begin{theorem}[Quantum Singular Value Discrimination using QSVT]
    \label{thm:QEVD}
    \UPDATED
    Suppose $A \in \cc^{N \times N}$ is an $s$-qubit operator (where $N = 2^s$)
    with singular value decomposition $A = \sum_{j \in [N]} \sigma_j \ketbra{u_j}{v_j}$
    such that all $\sigma_j$ lie in the range $[0,1]$.
    Let $\varphi \in \paren{0,\frac{1}{2}}$ and $\delta \in (0, 1]$ be some parameters.
    Suppose that for some $\alpha \ge 2$ and $\varepsilon$ satisfying
    \begin{equation*}
        \varepsilon = \smalloh{\frac{\delta\varphi}{\log(1/\delta)}}
    \end{equation*}
    we have access to $U_A$, an $(\alpha, a, \varepsilon)$-block-encoding of $A$ implemented in cost $T_A$.
    Then there exists a quantum algorithm $QSVD(\varphi,\delta)$ which implements a $(1, a + 1, \delta)$-block-encoding of some $(s + 1)$-qubit operator $D \in \cc^{2N \times 2N}$
    satisfying the following constraints for all $j \in [N]$:
    \begin{itemize}
        \item $\sigma_j \le \varphi \implies D\ket0\ket{v_j} = \ket0\ket{v_j}$
        \item $\sigma_j \ge 2\varphi \implies D\ket0\ket{v_j} = \ket1\ket{v_j}$
    \end{itemize}
    This algorithm has a cost of
    $$\order{\frac{\alpha}{\varphi}\logp{\frac{1}{\delta}}T_A}.$$
\end{theorem}

\begin{proof}
\UPDATED
We invoke \autoref{lem:approx-sign} with parameters
        $\varepsilon' := \frac\delta2$,
        $c := \frac{3\varphi}{2\alpha}$
        and $\Delta := \frac{\varphi}{2\alpha}$,
    to construct an even polynomial $P := P_{\varepsilon', \Delta, c}$
    of degree $n := \order{\frac\alpha\varphi \log(\frac1{\varepsilon'})}$,
    which is an $\varepsilon'$-approximation of $f(x) := \text{sign}\paren{\frac{3\varphi}{2\alpha} - x}$
    for $x \in \left[0, \frac\varphi\alpha \right] \cup \left[\frac{2\varphi}\alpha, 1 \right]$.
Invoking \autoref{thm:robust_qsvt} with $P$ and $U_A$,
    we get $U_B$ -- a $(1, a + 1, \gamma)$-block-encoding of $B := P(A/\alpha)$,
    implemented in cost $\order{n T_A}$,
    where $\varepsilon$ must satisfy $\varepsilon \le \alpha\gamma/2n$.

\newcommand{\qsvdswap}{\ensuremath{\texttt{SWAP}_{[s, s + a + 1]}}}

Now consider the following unitary $W$ that acts on $s + a + 2$ qubits:
\[ W := \qsvdswap^\dagger (H \otimes I_{s + a + 1}) \paren{\controlled{U_B}} (H \otimes I_{s + a + 1}) \qsvdswap \]

$W$ is the required block-encoding of $D$,
    and $\texttt{SWAP}_{[l, r]}$ sequentially swaps adjacent qubits with indices in range $[l, r]$
    effectively moving qubit indexed $l$ to the right of qubit $r$.
    (where qubits are zero-indexed, with higher indices for ancillas).
Let $\tilde B$ be the top-left block of $U_B$ (therefore $\norm{B - \tilde B} \le \gamma$).
Then we can extract $\tilde D$, the top-left block of $W$ as follows:
\begin{align*}
\tilde D
&= \paren{\bra{0}^{\ot a + 1} \ot I_{s + 1}}
    \qsvdswap^\dagger
    \paren{\ketbra{+} \ot I_{s + a + 1} + \ketbra{-} \ot U_B}
    \qsvdswap
    \paren{\ket{0}^{\otimes a + 1} \otimes I_{s + 1}}
\\
&= \ketbra{+} \otimes I_s + \ketbra{-} \otimes \tilde B
\end{align*}

Let us define index sets $L, R \subseteq [N]$ where
    $L := \{j \in [N] \mid \sigma_j \le \varphi \}$ and
    $R := \{j \in [N] \mid \sigma_j \ge 2\varphi \}$;
    and the corresponding subspace projections
    $\Pi_L := \sum_{j \in L} \ketbra{v_j}$,
    $\Pi_R := \sum_{j \in R} \ketbra{v_j}$, and
    $\Pi_\perp := I_{s} - \Pi_L - \Pi_R$.
Using these we pick our required operator $D$ as follows:
    \[ D := I \otimes \Pi_L + X \otimes \Pi_R + \tilde D (I \otimes \Pi_\perp) \]
That is, $D$ behaves as expected on the required subspace, and acts identical to $\tilde D$ on the remaining space.
The error in the block-encoding can be computed as
\begin{align*}
\norm{D - \tilde D} 
&= \norm{I \ot \Pi_L + X \ot \Pi_R + \tilde D (I \ot \Pi_\perp) - \tilde D}\\
&= \norm{I \ot \Pi_L + X \ot \Pi_R - \tilde D (I \ot (\Pi_L + \Pi_R))} \\
&= \norm{(I \ot I_s - \tilde D) (I \ot \Pi_L) + (X \ot I_s - \tilde D) (I \ot \Pi_R)} \\
&= \norm{\paren{\ketbra{-} \ot (I_s - \tilde B)} (I \ot \Pi_L) - \paren{\ketbra{-} \ot (I_s + \tilde B)} (I \ot \Pi_R)} \\
&= \norm{(I_s - \tilde B)\Pi_L - (I_s + \tilde B)\Pi_R} \\
&= \norm{(I_s - B)\Pi_L - (I_s + B)\Pi_R + (B - \tilde B)(\Pi_L - \Pi_R)} \\
&\le \norm{(I_s - P(A/\alpha))\Pi_L - (I_s + P(A/\alpha))\Pi_R} + \norm{B - \tilde B}\norm{\Pi_L - \Pi_R} \\
&\le \varepsilon' + \gamma \\
\end{align*}
We can choose $\gamma = \delta/2$, therefore 
\[
\varepsilon 
\le \frac{\alpha\delta}{4n} 
= \smalloh{\frac{\delta\varphi}{\log(\frac1\delta)}}
\]

\end{proof}

In \autoref{ssec:mi_using_qsvt}, we develop a variable stopping time quantum algorithm for matrix inversion using QSVT. In order to recast the usual matrix inversion to the VTAA framework, we need to be able to apply this algorithm to specific ranges of the singular values of the matrix. This is achieved by applying a controlled QSVD algorithm, to determine whether the input singular vector corresponds to an singular value less than (or greater than) a certain threshold. Based on the outcome of controlled QSVD, the standard inversion algorithm is applied. These two steps correspond to sub-algorithms $\mathcal{A}_j$ of the VTAA framework.

In prior works such as Refs.~\cite{ambainis2010variable, CKS17, CGJ19}, gapped phase estimation (GPE) was used to implement this. GPE requires an additional register of $\order{\log(\kappa)\log(1/\delta)}$ qubits to store the estimated phases. For the whole VTAA procedure, $\log\kappa$ such registers are needed. As a result, substituting GPE with QSVD, we save $\order{\log^2(\kappa)\log(1/\delta)}$ qubits.

\subsection{Variable-Time Quantum Algorithm for Matrix Inversion using QSVT}
\label{ssec:mi_using_qsvt}

Matrix inversion by QSVT applies a polynomial approximation of $f(x)=1/x$, satisfying the constraints laid out in \autoref{ssec:qsvt}. Here, we make use of the result of \cite{Grand_Uni_2021} to implement $A^{+}$. We adapt their result to the scenario where we have an approximate block-encoding of $A$ as input. Finally, we convert this to a variable stopping time quantum algorithm and apply VTAA to obtain a linear dependence on the condition number of $A$.   

\begin{lemma}[Matrix Inversion polynomial (Appendix C of \cite{Grand_Uni_2021})]
  \label{lem:poly_approx_inverse}
  Given $\kappa \geq 1, \varepsilon \in \mathbb{R}^+$,
  there exists an odd QSP polynomial $P_{\varepsilon, \kappa}^\text{MI}$
    of degree $\order{\kappa \log(\kappa/\varepsilon)}$,
    which is an $\frac{\varepsilon}{2\kappa}$ approximation
    of the function $f(x) = \frac{1}{2\kappa x}$
    in the range $\mathcal{D}:=[-1,-\frac{1}{\kappa}]\cup [\frac{1}{\kappa},1]$.
  Also in this range $P_{\varepsilon, \kappa}^\text{MI}$ is bounded from above by $1$,
    i.e. $ \forall x \in \mathcal{D} : \left|P_{\varepsilon, \kappa}^\text{MI}(x)\right| \le 1$.
\end{lemma}

\begin{theorem}[Inverting Normalized Matrices using QSVT]
    \label{thm:matrix_inversion_qsvt_normalized}
    Let $A$ be a normalized matrix
        with non-zero singular values in the range $[1/\kappa_A,1]$
        for some $\kappa_A \ge 1$.
    Let $\delta \in (0, 1].$ 
    For some
    {\UPDATED $\varepsilon = \smalloh{ \frac{\delta}{\kappa_A^2 \logp{{\kappa_A}/{\delta}} } }$ and $\alpha \ge 2$},
        let $U_A$ be an $(\alpha, a, \varepsilon)$-block-encoding of $A$,
        implemented in time $T_A$.
    Then we can implement a $(2 \kappa_A, a + 1, \delta)$-block-encoding of $A^+$
    at a cost of \[ \order{\kappa_A \alpha \logp{\frac{\kappa_A}{\delta}}T_A}. \]
\end{theorem}

\begin{proof}
\UPDATED
We use the matrix inversion polynomial defined in \autoref{lem:poly_approx_inverse},
    $P := P^{MI}_{\phi, \kappa}$ for this task,
    with $\kappa = \kappa_A\alpha$ and an appropriate $\phi$.
    This has a degree of $n := \order{\kappa_A \alpha \log\paren{\kappa_A\alpha/\phi}}$.
We invoke \autoref{thm:robust_qsvt} to apply QSVT using the polynomial $P$ above,
    block-encoding $U_A$, and an appropriate error parameter $\gamma$
    such that $\varepsilon \le \alpha\gamma/2n$,
    to get the unitary $U$, a $(1, a + 1, \gamma)$-block-encoding of $P(A/\alpha)$.
As $P$ is a $(\phi/2\kappa)$-approximation of $f(x) := 1/2\kappa x$, we have
\begin{align*}
    \norm{f(A/\alpha) - P(A/\alpha)} \le \frac{\phi}{2\kappa},
\end{align*}
which implies $U$ is a $(1, a + 1, \gamma + \phi/2\kappa)$-block-encoding of $f(A/\alpha)$.
And because $f(A/\alpha) = \frac{\alpha A^+}{2\kappa} = A^+ / 2\kappa_A$,
we can re-interpret $U$ as a $(2\kappa_A, a + 1, 2\kappa_A\gamma + \phi/\alpha)$-block-encoding of $A^+$.
Choosing $2\kappa_A \gamma = \phi/\alpha = \delta/2$, the final block-encoding has an error of $\delta$. This gives us $\phi = \alpha\delta/2$ and $\gamma = \delta/4\kappa_A$, and
\[ 
\varepsilon \le \frac{\alpha\gamma}{2n} 
= \frac{\alpha\delta}{8\kappa_A n}
= \order{\frac{\delta}{\kappa_A^2 \log(\kappa_A/\delta)}}
 \]

\end{proof}

Next, we design a map $W(\gamma, \delta)$ that uses QSVT to invert the singular values of a matrix if they belong to a particular domain. This helps us recast the usual matrix inversion algorithm as a variable-stopping-time algorithm and will be a key subroutine for boosting the success probability of this algorithm using VTAA. This procedure was also used in Refs.~\cite{CKS17,CGJ19} for the quantum linear systems algorithms.

\begin{theorem}[Efficient inversion of block-encoded matrix]
\label{thm:efficient_inversion_of_bem}
  Let $A$ be a normalized matrix with non-zero singular values in the range $\left[1/\kappa,1\right]$,
    for some $\kappa \ge 1$.
  Let $\delta \in (0, 1];~ 0<\gamma \leq 1$.
  Let $U_A$ be an $(\alpha, a, \varepsilon)$-block-encoding of $A$
    implemented in time $T_A$,
    such that {
    \UPDATED
    $\alpha \geq 2$ and
    $\varepsilon = \smalloh{\frac{\delta\gamma^2}{\logp{\frac1{\delta\gamma}}}}$
    }.
  Then for any quantum state $\ket{b}$ that is spanned by the left singular vectors of $A$ corresponding to the singular values in the range $\left[\gamma, 1\right]$, there exists a unitary $W(\gamma, \delta)$ that implements 
    \begin{equation}
        \label{eqn:efficient_inversion_qsvt_w_map}
        W(\gamma, \delta) : \ket{0}_F \ket{0}_Q \ket{b}_I \mapsto \dfrac{1}{a_{\max}} \ket{1}_F \ket{0}_Q f(A) \ket{b}_I + \ket{0}_F \ket{\perp}_{QI}
    \end{equation}
    where $a_{\max} = \order{\kappa_A}$ is a constant independent of $\gamma$,
    $\ket{\perp}_{QI}$ is an unnormalized quantum state orthogonal to $\ket{0}_Q$ and $\norm{f(A)\ket{b} - A^{+}\ket{b}} \leq \delta$.
    Here $F$ is a 1-qubit flag register,
        $Q$ is an $\alpha$-qubit ancilla register, and
        $I$ is the $\log N$-qubit input register.
    This unitary has a cost 
    \begin{equation}
        \label{eqn:efficient_MI_complexity}
        \order{\frac{\alpha}{\gamma}\logp{\frac{1}{\gamma\delta}} T_A}
    \end{equation}
\end{theorem}

\begin{proof}
    Since we only need to invert the singular values in a particular range, we can use the procedure in \autoref{thm:matrix_inversion_qsvt_normalized} with $\kappa_A$ modified to the restricted range. That gives us the description of a quantum circuit $\widetilde{W}(\gamma, \delta)$ that can implement the following map
    \begin{equation*}
        \widetilde{W}(\gamma, \delta): \ket{b}_I\ket{0}_Q \mapsto \frac{\gamma}{2} f(A) \ket{b}_I \ket{0}_Q + \ket{\perp}_{QI},
    \end{equation*}
    where $\ket{\perp}$ is an unnomalized state with no component along $\ket{0}_Q$. This has the same cost as \autoref{eqn:efficient_MI_complexity}. Here $\norm{f(A)\ket{\psi} - A^+ \ket{\psi}} \leq \delta$ whenever $\ket{\psi}$ is a unit vector in the span of the singular vectors of $A$ corresponding to the singular values in $[\gamma, 1]$. This follows from the sub-multiplicativity property of the matrix-vector product. 

    Next, we must transform the amplitude of the good part of the state to $\Theta({\kappa})$, independent of $\gamma$. To achieve this, we will have to flag it with an ancillary qubit to use a controlled rotation to modify the amplitude. Thus we add a single qubit $\ket{0}_F$ register and flip this register controlled on register $Q$ being in the state $\ket{0}$ (the good part). This gives us the transformation 
    \begin{equation*}
        \widetilde{W}'(\gamma, \delta) : \ket{0}_F\ket{b}_I\ket{0}_Q \mapsto \frac{\gamma}{2}\ket{1}_Ff(A)\ket{b}_I\ket{0}_Q + \ket{0}_F\ket{\perp}_{QI}
    \end{equation*}
        Then we use a controlled rotation to replace the amplitude $\gamma/2$ with some constant $a_{\max}^{-1}$ which is independent of $\gamma$, which is achieved by introducing the relevant phase to the flag space  
    \begin{equation*}
        \ket{1}_F \mapsto \frac{2}{\gamma a_{\max}} \ket{1}_F + \sqrt{1 - \frac{4}{\gamma^2a^2_{\max}}} \ket{0}_F. 
    \end{equation*}
    This gives us the desired $W(\gamma, \delta)$ as in \autoref{eqn:efficient_inversion_qsvt_w_map}. 
\end{proof}

Given such a unitary $W(\gamma,\delta)$, Ref.~\cite{CGJ19} laid out a procedure for a variable time quantum algorithm $\mathcal{A}$ that takes as input the block encoding of an $N \times d$ matrix $A$, and a state preparation procedure $U_b : \ket{0}^{\otimes n} \mapsto \ket{b}$, and outputs a quantum state that is a bounded distance away from $A^+\ket{b}/\norm{A^+\ket{b}}$. In order to determine the branches of the algorithm on which to apply VTAA at a particular iteration, \cite{CKS17, CGJ19, ambainis2010variable} use the technique of gapped phase estimation, which given a unitary $U$, a threshold $\phi$ and one of its eigenstate $\ket{\lambda}$, decides if the corresponding eigenvalue is a bounded distance below the threshold, or a bounded distance above it. In this work, we replace gapped phase estimation with the QSVD algorithm (\autoref{thm:QEVD}) which can be applied directly to any block-encoded (not necessarily Hermitian) matrix $A$, and allows for saving on $\order{\log^2(\kappa/\delta)}$ qubits.
~\\~\\
\textbf{The Variable time Algorithm:} This algorithm will be a sequence of $m$ sub-algorithms $\mathcal{A} =  \mathcal{A}_m \cdot \mathcal{A}_{m-1} \cdot \ldots \mathcal{A}_1$, where $m = \lceil \log\kappa\rceil + 1$. The overall algorithm acts on the following registers:

\begin{itemize}
    \item $m$ single qubit clock registers $C_i : i \in [m]$. 
    \item An input register $I$, initialized to $\ket{0}^{\otimes s}$.
    \item Ancillary register space $Q$ for the block encoding of $A$, initialized to $\ket{0}^{\otimes a}$. 
    \item A single qubit flag register $\ket{0}_F$ used to flag success of the algorithm.
\end{itemize}

Once we have prepared the above state space, we use the state preparation procedure to prepare the state $\ket{b}$. Now we can define how each $\mathcal{A}_j$ acts on the state space. Let $\varepsilon' = \frac{\delta}{a_{\max}m}$. The action of $\mathcal{A}_j$ can be broken down into two parts:

\begin{enumerate}
    \item If $C_{j-1} \ldots C_1$ is in state $\ket{0}^{\otimes (j-1)}$, apply QSVD$(2^{-j},\varepsilon')$, (\autoref{thm:QEVD}) to the state $\ket{b}$. The output is to be written to the clock register $C_j$.  
    \item If the state of $C_j$ is now $\ket{1}$, apply $W(2^{-j}, \varepsilon')$ to $I \otimes F \otimes Q$.
\end{enumerate}

Additionally, we would need algorithms $\mathcal{A}'=\mathcal{A}'_m\cdots\mathcal{A}'_1$ which are similar to $\mathcal{A}$, except that in Step 2, it implements $W'$ which sets the flag register to $1$. That is,
$$
W'\ket{b}_I\ket{0}_F\ket{0}_Q = \ket{b}_I\ket{1}_F\ket{0}_Q.
$$
Now we are in a position to define the variable time quantum linear systems algorithm using QSVT.
  
\begin{theorem}[Variable Time Quantum Linear Systems Algorithm Using QSVT]
  \label{thm:matrix_inversion_vtaa}
    Let $\varepsilon,\delta>0$.
    Let $A$ is a normalized $N\times d$ matrix such that its non-zero singular values lie in $[1/\kappa,1]$.
    Suppose that for 
    \begin{equation*}
      \UPDATED
    \varepsilon = \smalloh{\frac{\delta}{\kappa^3\log^2\paren{\frac{\kappa}{\delta}}}},
    \end{equation*}
    we have access to $U_A$ which is an $(\alpha, a, \varepsilon)$-block-encoding of $A$, implemented with cost $T_A$. 
    Let $\ket b$ be a state vector which is spanned by the left singular vectors of $A$.
    Suppose there exists a procedure to prepare the state $\ket{b}$ in cost $T_b$.
    Then there exists a variable time quantum algorithm
    that outputs a state that is $\delta$-close $\frac{A^+\ket{b} }{\norm{A^+\ket{b}}}$
    at a cost of
    \begin{equation}
        \order{\kappa\log\kappa\paren{\alpha T_A \logp{\frac{\kappa}{\delta}} + T_b}}
    \end{equation}  
    using $\order{\logp{\kappa}}$ additional qubits. 
\end{theorem}

\begin{proof}
The correctness of the algorithm is similar to that of Refs.~\cite{CKS17,CGJ19}, except here, we use QSVD instead of gapped phase estimation.
According to \autoref{lem:variable_time_aa}, we need $T_{\max}$ (the maximum time any of the sub-algorithms $\mathcal{A}_j$ take), $\norm{T}_2^2$ (the $\ell_2$-averaged stopping time of the sub-algorithms), and $\sqrt{p_{\mathrm{succ}}}$ (the square root of the success probability.) 
Now each sub-algorithm consists of two steps, implementing QEVD with precision $2^{-j}$ and error $\varepsilon'$, followed by $W(2^{-j}, \varepsilon')$. 
From \autoref{thm:QEVD}, the first step costs 
    \begin{equation*}
        \order{\alpha  T_A  2^{j} \logp{\frac{1}{\varepsilon'}}},
    \end{equation*}
    
    and the cost of implementing $W(2^{-j}, \varepsilon')$ is as described in \autoref{eqn:efficient_MI_complexity}. Thus the overall cost of $\mathcal{A}_j$, which is the sum of these two costs, turns out to be 

    \begin{equation}
        \label{eqn:complexity_a_j}
        \order{\alpha T_A 2^{j}\logp{\frac{2^{j}}{\varepsilon'}}}
    \end{equation}

    Note that the time $t_j$ required to implement $\mathcal{A}_j\ldots\mathcal{A}_1$ is also the same as \autoref{eqn:complexity_a_j}. Also, 
    \begin{align*}
        T_{\max} &= \max_j~t_j \\
                &= \max_j~\order{\alpha T_A 2^{j}\logp{\frac{2^{j}}{\varepsilon'}}} \\
                &= \order{\alpha T_A \kappa\logp{\frac{\kappa}{\varepsilon'}}} \\
                &= \order{\alpha T_A \kappa\logp{\frac{\kappa\logp{\kappa}}{\delta}}}.
    \end{align*}

The $\norm{T}_2^2$ is dependent on the probability that $\mathcal{A}$ stops at the $j^{\mathrm{th}}$ step. This is given by $p_j = \norm{\Pi_{C_j}\mathcal{A}_j \ldots \mathcal{A}_1\ket{\psi}_I\ket{0}_{CFPQ}}^2$, where $\Pi_{C_j}$ is the projector on $\ket{1}_{C_j}$, the $j^{\mathrm{th}}$ clock register. From this, $\norm{T}_2^2$ can be calculated as

    \begin{align*}
        \norm{T}_2^2 &= \sum_j p_j t_j^2 \\ 
                       &= \sum_j \norm{\Pi_{C_j}\mathcal{A}_j \ldots \mathcal{A}_1\ket{\psi}_I\ket{0}_{CFPQ}}^2 t_j^2 \\ 
                       &= \sum_k \abs{c_k}^2 \sum_j \paren{\norm{\Pi_{C_j}\mathcal{A}_j\ldots \mathcal{A}_1 \ket{v_k}_I\ket{0}_{CFPQ}}^2 t_j^2} \\ 
                       &= \order{\alpha^2 T_A^2 \sum_k{\log^2\paren{\frac{1}{\sigma_k\varepsilon'}} \frac{\abs{c_k}^2}{\sigma_k^2}}  }
    \end{align*}

    Therefore 
    \begin{equation}
        \norm{T}_2 = \order{\alpha T_A \logp{\frac{\kappa \log{\kappa}}{\delta}}\sqrt{\sum_k\frac{\abs{c_k}^2}{\sigma_k^2}}}.
    \end{equation}

    Next we calculate the success probability. 
    \begin{align*}
        \sqrt{p_{\mathrm{succ}}} &= \norm{\Pi_F \frac{A^{-1}}{\alpha_{\max}} \ket{b}_I\ket{\phi}_{CFPQ}} + \order{m \varepsilon'} \\ 
                                 &= \frac{1}{\alpha_{\max}} \sqrt{\sum_j \frac{\abs{c_j}^2}{\sigma_j^2}} + \order{\frac{\delta}{\alpha_{\max}}} \\ 
                                 &= \Omega\paren{\frac{1}{\kappa}\sqrt{\sum_j \frac{\abs{c_j}^2}{\sigma_j^2}}}
    \end{align*}

    Given these, we can use \autoref{lem:variable_time_aa} to write the final complexity of matrix inversion with VTAA:
    \begin{align*}
        T_{\max} + T_{b} + \frac{(\norm{T}_2 + T_b)\logp{T'_{\max}}}{\sqrt{p_{\textrm{succ}}}} 
        = \order{\kappa \log\kappa\paren{\alpha T_A \logp{\frac{\kappa}{\delta}} + T_b}}
    \end{align*}
    \UPDATED
    The upper bound on the precision required for the input block-encoding, $\varepsilon$,
    can be calculated from the bounds on the precisions for $W(\kappa, \varepsilon')$ (\autoref{thm:efficient_inversion_of_bem}) and $\mathrm{QSVD}(\kappa, \varepsilon')$ (\autoref{thm:QEVD}) as follows:
    \begin{align*}
    \varepsilon
    = \smalloh{\min\paren{\frac{\varepsilon'}{\kappa^2\logp{\frac{\kappa}{\varepsilon'}}}, \frac{\varepsilon'}{\kappa\logp{\frac{1}{\varepsilon'}}}}} 
    = \smalloh{{\frac{\varepsilon'}{\kappa^2\logp{\frac{\kappa}{\varepsilon'}}}}} 
    = \smalloh{\frac{\delta}{\kappa^3\log^2\paren{\frac{\kappa}{\delta}}}}
    \end{align*}
\end{proof}
The overall complexity is better by a $\log$ factor and requires $\order{\log^2(\kappa/\delta)}$ fewer additional qubits as compared to the variable time algorithms in Refs.~\cite{CKS17,CGJ19}.

\subsection{Negative Powers of Matrices using QSVT}
\label{subsec:negative-power-qsvt} 

We consider the problem:
    given an approximate block-encoding of a matrix $A$,
    we need to prepare a block-encoding of $A^{-c}$, where $c\in (0,1)$.
This procedure will be used to develop algorithms for $\ell_2$-regularized versions of GLS.
We will directly use the results of~\cite{GSLW2019}.

\begin{lemma}[Polynomial approximations of negative power functions, \cite{GSLW2019full}, Corollary 67]
    \label{lem:poly_approx_negative_powers}
    Let $\varepsilon, \delta \in (0, \frac{1}{2}], c > 0$
    and let $f(x) := \frac{\delta^c}{2} x^{-c}$,
    then there exist even/odd polynomials
    $P_{c, \varepsilon, \delta}, P'_{c, \varepsilon, \delta} \in \rr[x]$
    such that $\norm{P_{c, \varepsilon, \delta} - f}_{[\delta, 1]} \leq \varepsilon$,
    $\norm{P_{c, \varepsilon, \delta}}_{[-1, 1]} \le 1$
    and $\norm{P'_{c, \varepsilon, \delta} - f}_{[\delta, 1]} \le \varepsilon$,
    $\norm{P'_{c, \varepsilon, \delta}}_{[-1, 1]} \le 1$.
    Moreover the degree of the polynomials are
    $\order{\frac{\max(1, c)}{\delta} \logp{\frac{1}{\varepsilon}}}$.
\end{lemma}

\begin{theorem}[Negative fractional powers of a normalized matrix using QSVT]
    \label{thm:constr_neg_powers_normalized}
    Let $c \in (0, 1)$ be some constant and $\delta \in (0, 1]$
    Let $A$ be a normalized matrix with non-zero singular values in the range $[1/\kappa,1]$.
    Let $U_A$ be a $(\alpha, a, \varepsilon)$-block-encoding of a matrix $A$, implemented in time $T_A$ such that $\alpha\geq 2$ and 
    \[
        \UPDATED
        \varepsilon 
        = \smalloh{\frac \delta {\kappa^{c + 1} \log(\kappa/\delta)}}
    \]
    Then we can construct a $(2\kappa^c, a+1, \delta)$-block-encoding of $A^{-c}$ at a cost of
    \begin{equation*}
      \order{\alpha \kappa \logp{\frac{\kappa}\delta} ~T_A}.
    \end{equation*}
\end{theorem}

\begin{proof}
\UPDATED
  From \autoref{lem:poly_approx_negative_powers},
  using $\Delta := \frac1{\kappa \alpha}$
  and an appropriate $\varphi \in (0, \frac{1}{2}]$,
  we get an even QSP polynomial $P := P_{c, \varphi, \Delta}$
  which is $\varphi$-close to $f(x) := \frac{1}{2\kappa^c\alpha^cx^c}$,
  and has degree $n$ such that $n = \order{\alpha\kappa\logp{\frac{1}{\varphi}}}$.
  Therefore \[ \norm{f(A/\alpha) - P(A/\alpha)} \leq \varphi. \]

  Using \autoref{thm:robust_qsvt} we can construct $U_P$, a $(1, a+1, \gamma)$-block-encoding of $P(A/\alpha)$,
  given that $\varepsilon \leq \frac{\alpha\gamma}{2n}$. 
  Then from triangle inequality it follows that it is a $(1, a+1,~ \varphi + \gamma)$-block-encoding of $f(A/\alpha)$. 
  And because $f(A/\alpha) = \frac{A^{-c}}{2\kappa^c}$,
  $U_P$ can be re-interpreted as a $(2\kappa^c, a+1, 2\kappa^c(\varphi+\gamma))$-block-encoding of $A^{-c}$. 
  We therefore choose $\varphi = \gamma = \frac{\delta}{4\kappa^c}$, and choose $\varepsilon$ as
  \[
  \varepsilon 
  = \smalloh{\alpha \frac{\delta}{4\kappa^c} \frac{1}{\alpha\kappa \log(4\kappa^c/\delta)}}
  = \smalloh{\frac \delta {\kappa^{c + 1} \log(\kappa/\delta)}}
  \]
  
\end{proof}

Having discussed the necessary algorithmic primitives, we are now in a position to design quantum algorithms for linear regression with general $\ell_2$-regularization. We will first deal with ordinary least squares followed by weighted and generalized least squares.

\section{Quantum Least Squares with General $\ell_2$-Regularization}
\label{sec:main_proof}

In this section, we derive the main results of our paper, namely quantum algorithms for quantum ordinary least squares (OLS), quantum weighted least squares (WLS) and quantum generalized least squares (GLS) with $\ell_2$-regularization. 

\subsection{Quantum Ordinary Least Squares}
\label{subsec:ols-regularized}

Given $N$ data points $\{a_i,b_i\}_{i=1}^N$ such that $a_i\in\rr^d$ and $b_i\in \rr$, the objective of linear regression is to find $x\in\rr^d$ that minimizes the loss function
\begin{equation}
    \mathcal{L}_{O}=\sum_{j=1}^N (x^Ta_i-b_i)^2.
\end{equation}
Consider the $N\times d$ matrix $A$ (known as the data matrix) such that the $i^{\mathrm{th}}$ row of $A$ is the vector $a_i$ transposed and the column vector $b=(b_1\cdots b_N)^T$. Then, the solution to the OLS problem is given by $x=(A^TA)^{-1}A^T b=A^+b$. 

For the $\ell_2$-regularized version of the OLS problem, a penalty term is added to its objective function. This has the effect of shrinking the singular values of $A$ which helps overcome problems such as rank deficiency and \textit{overfitting} for the OLS problem. The loss function to be minimized is of the form
\begin{equation}
    \norm{Ax-b}^2_2+\norm{Lx}^2_2,
\end{equation}
where $L$ is the $N\times d$ penalty matrix and $\lambda>0$ is the optimal regularizing parameter. The solution $x\in\rr^d$ satisfies
\begin{equation}
    x=(A^TA+\lambda L^TL)^{-1}A^Tb.
\end{equation}
Therefore, for quantum ordinary least squares with general $\ell_2$-regularization, we assume that we have access to approximate block-encodings of the data matrix $A$, $L$ and a procedure to prepare the quantum state $\ket{b}=\sum_{j=1}^N b_j\ket{j} / \norm{b}$. Our algorithm outputs a quantum state that is close to
\begin{equation}
\label{eq:sol-ols-l2}
\ket{x}=\dfrac{(A^TA+\lambda L^TL)^{-1}A^T\ket{b}}{\norm{(A^TA+\lambda L^TL)^{-1}A^T\ket{b}}}.
\end{equation}

In order to implement a quantum algorithm that implements this, a straightforward approach would be the following: We first construct block-encodings of $A^TA$ and $L^TL$, given block encodings of $A$ and $L$, respectively (Using \autoref{lem:prod_of_be}).
We could then implement a block-encoding of $A^TA+\lambda L^TL$ using these block encodings (By \autoref{lem:constr_lincomb}).
On the other hand, we could also prepare a quantum state proportional to $A^T\ket{b}$ by using the block-encoding for $A$ and the unitary preparing $\ket{b}$.
Finally, using the block encoding of $A^TA+\lambda L^TL$, we could implement a block-encoding of $(A^TA+\lambda L^TL)^{-1}$ (using \autoref{thm:matrix_inversion_qsvt_normalized}) and apply it to the state $A^T\ket{b}$.
Although this procedure would output a quantum state close to $\ket{x}$, it is not efficient.
It is easy to see that the inverse of $A^TA+\lambda L^TL$, would be implemented with a complexity that has a quadratic dependence on the condition numbers of $A$ and $L$.
This would be undesirable as it would perform worse than the unregularized quantum least squares algorithm, where one is able to implement $A^+$ directly.
However, it is possible to design a quantum algorithm that performs significantly better than this.

The first observation is that it is possible to recast this problem as finding the pseudoinverse of some augmented matrix. Given the data matrix $A\in\rr^{N\times d}$, the regularizing matrix $L\in\rr^{N\times d}$, let us define the following augmented matrix

\begin{equation}
\label{eq:augmented-data-matrix}
    A_L := \begin{pmatrix}
        A & 0 \\
        \sqrt{\lambda} L & 0
    \end{pmatrix}.
\end{equation}

It is easy to see that the top left block of $A^{+}_L=(A^TA+\lambda L^TL)^{-1}A^T$, which is the required linear transformation to be applied to $b$. Consequently, our strategy would be to implement a block-encoding of $A_L$, given block-encodings of $A$ and $L$. Following this, we use matrix inversion by QSVT to implement $A^{+}_L\ket{b}\ket{0}$. The first register is left in the quantum state given in \autoref{eq:sol-ols-l2}.

{\UPDATED
From this, it is clear that the complexity of our quantum algorithm would depend on the effective condition number of the augmented matrix $A_L$. In this regard, we shall assume that the penalty matrix $L$ is a \textit{good regularizer}. That is, $L$ is chosen such that it does not have zero singular values (positive definite). This is a fair assumption as if $L$ has only non-zero singular values, the minimum singular value of $A_L$ is guaranteed to be lower bounded by the minimum singular value of $L$. This ensures that the effective condition number of $A_L$ depends on $\kappa_L$, even when the data matrix $A$ has zero singular values and $A^TA$ is not invertible. Consequently, this also guarantees that regularized least squares provide an advantage over their unregularized counterparts.
} 

Next, we obtain bounds on the effective condition number of the augmented matrix $A_L$ for a good regularizer $L$ via the following lemma:

\begin{lemma}[Condition number and Spectral Norm of $A_L$]
  \label{thm:spec-augmented_regression_matrix}
  Let the data matrix $A$ and the {\UPDATED positive definite} penalty matrix $L$
      have spectral norms $\norm{A}$ and $\norm{L}$, respectively.
  Furthermore, suppose their effective condition numbers be upper bounded by $\kappa_A$ and $\kappa_L$.
  Then the ratio between the maximum and minimum (non-zero) singular value of $A_L$ is upper bounded by 
  \begin{equation*}
  \kappa
  = \kappa_L\paren{1 + \frac{\norm{A}}{\sqrt\lambda \norm{L}}}
  \end{equation*}
  We can also bound the spectral norm as
  \begin{equation*}
      \norm{A_L} = \Theta \paren{\norm{A} + \sqrt\lambda\norm{L}}
  \end{equation*}
\end{lemma}
\begin{proof}

To bound the spectral norm and condition number of $A_L$, consider the eigenvalues of the following matrix:
    \begin{equation*}
      A_L^T A_L = \begin{pmatrix}
          A^T A + \lambda L^T L & 0 \\
          0 & 0
        \end{pmatrix}
    \end{equation*}

    This implies that the non-zero eigenvalues of $A_L^T A_L$
      are the same as those of $A^T A + \lambda L^T L$.
    Therefore, using triangle inequality, the spectral norm of $A_L$ can be upper-bounded as follows:
    \begin{align*}
      \norm{A_L}
       = \sqrt \norm{A_L^T A_L} 
       = \sqrt \norm{A^T A + \lambda L^T L} 
       \le \sqrt {\norm{A^T A} + \lambda \norm{L^T L}} 
       = \sqrt {\norm{A}^2 + \lambda \norm{L}^2} 
       \leq \norm{A} + \sqrt\lambda \norm{L}
    \end{align*}
    Similarly $\norm{A_L} \ge \norm{A}$ and $\norm{A_L} \ge \sqrt\lambda \norm{L}$,
    which effectively gives the tight bound for $\norm{A_L}$.
    

As $L^TL$ is positive definite,
we have that its minimum singular value is $\sigma_{\mathrm{min}(L)}=\norm{L}/\kappa_L$.
And we also know that $A^TA$ is positive semidefinite, 
so by Weyl's inequality, the minimum singular value of $A_L$ is lower bounded by
\begin{align*}
    \sigma_{\min}\paren{A_L} 
    \ge \sqrt {\sigma_{\min}\paren{A}^2 + \lambda \sigma_{\min}\paren{L}^2} 
    \ge \sqrt {\lambda \frac{\norm{L}^2}{\kappa_L^2}} 
    = \sqrt \lambda \frac{\norm{L}}{\kappa_L}
\end{align*}

Thus, 
\begin{equation*}
\dfrac{\sigma_{\max}\paren{A_L}}{\sigma_{\min}\paren{A_L}}\leq\kappa
=\kappa_L\paren{1 + \frac{\norm{A}}{\sqrt\lambda \norm{L}}}
\end{equation*}
\end{proof}

{\UPDATED In the theorems and lemmas for regularized quantum linear regression and its variants that we develop in this section, we consider that $L$ is a \textit{good regularizer} in order to provide a simple expression for $\kappa$. However, this is without loss of generality. When $L$ is not a good regularizer, the expressions for the respective complexities will remain unaltered, except that $\kappa$ would now correspond to the condition number of the augmented matrix.}

Now it might be possible that $\ket{b}$ does not belong to the row space of $(A^TA+\lambda L^TL)^{-1}A^T$ which is equivalent to saying $\ket{b}\ket{0}$ may not lie in $\mathrm{row}(A^{+}_L)$.  However, it is reasonable to expect that the initial hypothesis of the underlying model being close to linear is correct. That is, we expect $\ket{b}$ to have a good overlap with $\mathrm{row}\left(A^+_L\right)=\mathrm{col}\left(A_L\right)$. The quantity that quantifies how far the model is from being linear is the so called \textit{normalized residual sum of squares}. For $\ell_2$-regularized ordinary least squares, this is given by
\begin{equation}
\mathcal{S}_O = \dfrac{\norm{(I-\Pi_{\mathrm{col}(A_L)})\ket{b}\ket{0}}^2}{\norm{\ket{b}}^2}=1-\norm{\Pi_{\mathrm{col}(A_L)}\ket{b}\ket{0}}^2.
\end{equation} 

If the underlying data can indeed be fit by a linear function, $\mathcal{S}_O$ will be low. Subsequently, we assume that $\mathcal{S}_O=1-\norm{\Pi_{\mathrm{col}(A_L)}\ket{b}\ket{0}}^2\leq \gamma<1/2$. This in turn implies that $\norm{\Pi_{\mathrm{col}(A_L)}\ket{b}\ket{0}}^2=\Omega(1)$, implying that the data can be reasonably fit by a linear model.\footnote{Our results also hold if we assume that $\mathcal{S}_O\leq \gamma$ for some $\gamma\in (0,1)$. That is, $\norm{\Pi_{\mathrm{col}(A_L)}}\geq 1-\gamma$. In such a scenario our complexity to prepare $A^{+}_L\ket{b,0}/ \norm{A^+_L\ket{b,0}}$ is rescaled by $1/\sqrt{1-\gamma}$.}

Now we are in a position to present our quantum algorithm for the quantum least squares problem with general $\ell_2$-regularization. We also present an improved quantum algorithm for the closely related quantum ridge regression, which is a special case of the former. 


\begin{theorem}[Quantum Ordinary Least Squares with General $\ell_2$-Regularization]
    \label{thm:quantum_least_squares_gen_tik_reg}
    Let $A, L \in \rr^{N \times d}$ be the data and penalty matrices with effective condition numbers $\kappa_A$ and $ \kappa_L$ respectively, and $\lambda \in \rr^+$ be the regression parameter.
    Let $U_A$ be a $(\alpha_A, a_A, \varepsilon_A)$-block-encoding of $A$
      implemented in time $T_A$
    and $U_L$ be a $(\alpha_L, a_L, \varepsilon_L)$-block-encoding of $L$
      implemented in time $T_L$.
    Furthermore, suppose $U_b$ be a unitary that prepares $\ket b$ in time $T_b$ and
    {\UPDATED
    define  \[ \kappa = \OLSKappa \]
    }
    Then for any $\delta \in (0, 1)$ such that
    \begin{equation}
        \label{eqn:qols_error_condition}
        \UPDATED
        \varepsilon_A, \sqrt\lambda\varepsilon_L = \smalloh{\frac{\delta}{\kappa^3\log^2\paren{\frac{\kappa}{\delta}}}}
    \end{equation}
      we can prepare a state that is $\delta$-close to
      \begin{equation*} \OLSstate \end{equation*}
      with probability $\Theta(1)$,
      at a cost of
      \begin{equation}
          \label{eqn:ols_complexity_vtqa}
          \UPDATED
          \OLScomplexity{\alpha_A}{\alpha_L} 
      \end{equation}
    using only $\order{\log \kappa}$ additional qubits.
\end{theorem}
\begin{proof}
  We invoke \autoref{lem:augmented_regression_matrix}, to obtain a unitary $U$,
  which is a
  $(\alpha_A + \sqrt\lambda \alpha_L, \max(a_A, a_L) + 2, \varepsilon_A + \sqrt\lambda\varepsilon_L)$-block-encoding of the matrix $A_L$, implemented at a cost of $\order{T_A + T_L}$. Note that in \autoref{lem:augmented_regression_matrix}, $A$ and $L$ are considered to be $s$-qubit operators. For $N\times d$ matrices, such that $N,d\leq 2^s$, we can pad them with zero entries. Padding $A$ and $L$ with zeros may result in the augmented matrix $A_L$ having some zero rows between $A$ and $L$. However, this is also not an issue as we are only interested in the top left block of $A^{+}_L$ which remains unaffected.

Note that $U$ can be reinterpreted as a $\left(\frac{\alpha_A + \sqrt\lambda \alpha_L}{\norm{A_L}}, \max(a_A, a_L) + 2, \frac{\varepsilon_A + \sqrt\lambda \varepsilon_L}{\norm{A_L}}\right)$-block-encoding of the normalized matrix $A_L/\norm{A_L}$. Furthermore, we can prepare the quantum state $\ket{b}\ket{0}$ in time $T_b$. Now by using \autoref{thm:matrix_inversion_vtaa} with $U$ and an appropriately chosen $\delta$ specified above, we obtain a quantum state that is $\delta$-close to 
\begin{equation*}
\normalized{(A^TA + \lambda L^TL)^{-1} A^T \ket b}
\end{equation*}
in the first register. 
\end{proof}




In the above complexity, when $L$ is a good regularizer, $\kappa$ is independent of $\kappa_A$.
$\kappa$ can be made arbitrarily smaller than $\kappa_A$ by an appropriate choice of $L$.
Thus the regularized version has significantly better time complexity than the unregularized case.
One such example of a good regularizer is in case of \textit{Quantum Ridge Regression}, where we use the identity matrix to regularize.
The corollary below elucidates this.

\begin{corollary}[Quantum Ridge Regression]
    \label{thm:quantum_ridge_regression}
    Let $A$ be a matrix of dimension $N \times d$
      with effective condition number $\kappa_A$
    and $\lambda \in \rr^+$ be the regression parameter.
    Let $U_A$ be a $(\alpha, a, \varepsilon)$-block-encoding of $A$
      implemented in time $T_A$.
    Let $U_b$ be a unitary that prepares $\ket b$ in time $T_b$.
    If $\kappa = 1 + \norm{A} / \sqrt\lambda$
    then for any $\delta$ such that
    \begin{equation*}
    \UPDATED
    \varepsilon = \smalloh{\frac{\delta}{\kappa^3\log^2\paren{\frac{\kappa}{\delta}}}}
    \end{equation*}
    we can prepare a state $\delta$-close to
    \begin{equation*}
        \normalized{\paren{A^TA + \lambda I}^{-1} A^T \ket b}
    \end{equation*}
    at a cost of
      \begin{equation} \RidgeComplexity \end{equation}
   with probability $\Theta(1)$ using only $\order{\log \kappa}$ additional qubits.
\end{corollary}

\begin{proof}
    The identity matrix $I$ is a trivial $(1, 0, 0)$-block-encoding of itself, and $\kappa_I = 1$. We invoke \autoref{thm:quantum_least_squares_gen_tik_reg} with $L = I$ to obtain the solution.
\end{proof}

Being in the block-encoding framework allows us to express the complexity of our quantum algorithm in specific input models such as the \textit{quantum data structure input model} and the \textit{sparse access model}. We express these complexities via the following corollaries.

\begin{corollary}[Quantum Ordinary Least Squares with $\ell_2$-Regularization in the Quantum Data Structure Input Model]
    \label{cor:quantum_reg_in_qrom}
    Let $ A, L \in \mathbb{R}^{N \times d} $
    with effective condition numbers $\kappa_A, \kappa_L$ respectively.
    Let $\lambda \in \mathbb{R}^+$ and $b \in \rr^N$.
    Let $\kappa$ be the effective condition number of the augmented matrix $A_L$.
    Suppose that $A$, $L$ and $b$ are stored in a quantum accessible data structure. 
    Then for any $\delta > 0$ 
    there exists a quantum algorithm to prepare a quantum state $\delta$-close to
    \begin{equation*}
        \OLSstate{}
    \end{equation*}
    with probability $\Theta(1)$,
    at a cost of
    \begin{equation}
        \order{\kappa
            \paren{\frac{\mu_A + \sqrt\lambda \mu_L}{\norm{A} + \sqrt\lambda \norm{L}}}
            \polylog{Nd, \kappa, \frac{1}{\delta}, \lambda}
        }.
    \end{equation}
\end{corollary}
\begin{proof}
Since $b$ is stored in the data structure,
    for some $\varepsilon_b > 0$,
    we can prepare the state $\ket{b'}$ that is $\varepsilon_b$-close
    to $\ket{b}=\sum_{i}b_i\ket{i}/\norm{b}$ 
    using $T_b=\order{\mathrm{polylog}(N/\varepsilon_b)}$ queries to the data structure (see \autoref{sssec:q_data_structure_model}.)
    \UPDATED
Similarly, for some parameters $\varepsilon_A, \varepsilon_L > 0$,
    we can construct a $\left(\mu_A,\lceil\log(d+N)\rceil,\varepsilon_A\right)$-block-encoding
    of $A$ using $T_A=\order{\mathrm{polylog}(Nd/\varepsilon_A)}$ queries to the data structure
    and a $\left(\mu_L,\lceil\log(d+N)\rceil,\varepsilon_B\right)$-block-encoding
    of $L$ using $T_L=\order{\mathrm{polylog}(Nd/\varepsilon_B)}$ queries. 

We invoke \autoref{thm:quantum_least_squares_gen_tik_reg} with a precision $\delta/2$ by choosing $\varepsilon_A$ and $\varepsilon_L$ such that equation \autoref{eqn:qols_error_condition} is satisfied. 
This gives us a state that is $\delta/2$-close to
\begin{equation*}
    \normalized{\paren{A^TA + \lambda L^TL}^{-1} A^T \ket{b'}}
\end{equation*}

To compute the final precision as $\delta$, we use \autoref{lem:apply_block_enc_approx} by choosing $\varepsilon_b = \frac{\delta}{2\kappa}$.
The complexity can be calculated by plugging in the relevant values in \autoref{eqn:ols_complexity_vtqa}
\end{proof}

In the previous corollary $\mu_A=\norm{A}_F$ and $\mu_L=\norm{L}_F$ when the matrix $A$ and $L$ are stored in the data structure. Similarly, $\mu_A=\mu_p(A)$ and $\mu_L=\mu_p(L)$ when the matrices $A^{(p)}, A^{(1-p)}$ and $L^{(p)}, L^{(1-p)}$ are stored in the data structure. 

Now we discuss the complexity of quantum ordinary least squares with $\ell_2$-regularization in the \textit{sparse access input model}. We call a matrix $M$ as $(s_r,s_c)$ row-column sparse if it has a row sparsity $s_r$ and column sparsity $s_c$. 

\begin{corollary}[Quantum Ordinary least squares with $\ell_2$-regularization in the sparse access model]
    \label{cor:quantum_reg_in_sparse}
    Let $ A \in \rr^{N \times d} $ be $ (s^A_r, s^A_c)$ row-column sparse,
    and similarly, let $ L \in \rr^{N \times d} $ be $ (s^L_r, s^L_c)$ row-column sparse,
    with effective condition numbers $\kappa_A$ and $\kappa_L$ respectively.
    Let $ \lambda \in \mathbb{R}^+ $ and $\delta > 0$.
    Suppose there exists a unitary that prepares $\ket{b}$ at a cost, $T_b$.
    Then there is a quantum algorithm to prepare 
    a quantum state that is $\delta$-close to
  
    \begin{equation*}
    \normalized{(A^TA + \lambda L^TL)^{-1} A^T \ket b}
    \end{equation*} 
    with probability $\Theta(1)$,
    at a cost of
    \begin{equation}
        \order{\kappa
            \paren{\frac{\sqrt{s^A_r s^A_c} + \sqrt{\lambda s^L_r s^L_c}}{\norm{A}+ \sqrt\lambda \norm{L}}}
            \polylog{Nd, \kappa, \frac{1}{\delta}, \lambda}
            + \kappa\log\kappa T_b
        }.
    \end{equation}
\end{corollary}
\begin{proof}
    The proof is similar to \autoref{cor:quantum_reg_in_qrom} but with $\alpha_A=\sqrt{s^A_r s^A_c}$ and $\alpha_L=\sqrt{s^L_r s^L_c}$. 
\end{proof}

\subsection{Quantum Weighted And Generalized Least Squares}
\label{subsec:ols-gls-l2}

This technique of working with a augmented matrix will also hold for the other variants of ordinary least squares. In this section, we begin by briefly describing these variants before moving on to designing quantum algorithms for the corresponding problems. 
\\~\\
\textbf{Weighted Least Squares:~} For the WLS problem, each observation $\{a_i,b_i\}$ is assigned some weight $w_i\in\rr^{+}$ and the objective function to be minimized is of the form 
\begin{equation}
\mathcal{L}_W:= \sum_j w_j (x^T a_j - b_j)^2.
\end{equation}
If $W\in\rr^{N\times N}$ is the diagonal matrix with $w_i$ being the $i^{\mathrm{th}}$ diagonal entry, then the optimal $x$ satisfies
\begin{equation}
x=(A^TWA)^{-1}A^TWb.
\end{equation}
The $\ell_2$-regularized version of WLS satisfies
\begin{equation}
x=(A^TWA+\lambda L^TL)^{-1}A^TWb
\end{equation}
Our quantum algorithm outputs a state that is close to 
\begin{equation}
\ket{x}=\normalized{(A^TWA+\lambda L^TL)^{-1}A^TW\ket{b}}
\end{equation}
given approximate block-encodings of $A$, $W$ and $L$. 
Much like \autoref{eq:augmented-data-matrix}, finding the optimal solution reduces to finding the   
pseudo inverse of an augmented matrix $A_L$ given by
\[
A_L := \begin{pmatrix}
        \sqrt{W}A & 0 \\
        \sqrt{\lambda} L & 0
    \end{pmatrix}.
\]
The top left block of $A^{+}_L=(A^TWA+\lambda L^TL)^{-1}A^T\sqrt{W}$, which is the required linear transformation to be applied to the vector $y=\sqrt{W} b$.
The ratio between the minimum and maximum singular values of $A_L$, $\kappa$, can be obtained analogously to \autoref{thm:spec-augmented_regression_matrix}. 
For the $\ell_2$-regularized WLS problem, \textit{normalized residual sum of squares} is given by
\begin{equation}
\mathcal{S}_W = \dfrac{\norm{(I-\Pi_{\mathrm{col}(A_L)})\ket{y}\ket{0}}^2}{\norm{\ket{y}}^2}=1-\norm{\Pi_{\mathrm{col}(A_L)}\ket{y}\ket{0}}^2.
\end{equation} 
Subsequently, we assume that $\mathcal{S}_W=1-\norm{\Pi_{\mathrm{col}(A_L)}\ket{y}\ket{0}}^2\leq \gamma<1/2$. This in turn implies that $\norm{\Pi_{\mathrm{col}(A_L)}\ket{y}\ket{0}}^2=\Omega(1)$, implying that the data can be reasonably fit by a linear model.
\\~\\
\textbf{Generalized Least Squares.}
Similarly, we can extend this to GLS problem, where there the input data may be correlated.
These correlations are given by the non-singular covariance matrix $\Omega\in\rr^{N\times N}$.
The WLS problem is a special case of the GLS problem, corresponding to when $\Omega$ is a diagonal matrix.
The objective function to be minimized is
\begin{equation}
    \mathcal{L}_\Omega:= \sum_{i,j} (\Omega^{-1})_{ij} (x^T a_i - b_i)(x^Ta_j-b_j).
\end{equation}
The optimal $x\in\rr^d$ satisfies
\begin{equation}
x=(A^T\Omega^{-1}A)^{-1}A^T\Omega^{-1}b
\end{equation}
Similarly, the $\ell_2$-regularized GLS solver outputs $x$ such that
\begin{equation}
x=(A^T\Omega^{-1}A+\lambda L^TL)^{-1}A^T\Omega^{-1}b.
\end{equation}
So, given approximate block-encodings of $A$, $\Omega$ and $L$ a quantum GLS solver outputs a quantum state close to
\begin{equation}
\ket{x}=\normalized{(A^T\Omega^{-1}A+\lambda L^TL)^{-1}A^T\Omega^{-1}\ket{b}}
\end{equation} 
The augmented matrix $A_L$ is defined as
\[
A_L := \begin{pmatrix}
        \Omega^{-1/2} A & 0 \\
        \sqrt{\lambda} L & 0
    \end{pmatrix}.
\]
Then top left block of $A^{+}_{L}$ to the vector $y=\Omega^{-1/2}b$ yields the optimal $x$. Thus the quantum GLS problem with $\ell_2$-regularization first prepares $\Omega^{-1/2}\ket{b}\ket{0}$ and then uses the matrix inversion algorithm by QSVT to implement $A^{+}_L \Omega^{-1/2}\ket{b}\ket{0}$. Analogous to OLS and WLS, we assume that the normalized residual sum of squares $\mathcal{S}_{\Omega}\leq \gamma<1/2$.

\subsubsection{Quantum Weighted Least Squares}
In this section, we derive the complexity of the $\ell_2$-regularized WLS problem.
{\UPDATED
We assume that we have a diagonal weight matrix $W \in \rr^{N \times N}$ such that its smallest and largest diagonal entries are $\wmin$ and $\wmax$, respectively. 
  This implies that $\norm{W} = \wmax$ and $\kappa_W = \wmax/\wmin$. 
  We take advantage of the fact that the matrix $W$ is diagonal and then apply controlled rotations to directly implement a block encoding of $\sqrt{W}A$. Additionally, given a state preparation procedure for $\ket{b}$, we can easily prepare a state proportional to $\sqrt{W}\ket{b}$. 
  We then use \autoref{thm:quantum_least_squares_gen_tik_reg} to solve QWLS. 

We first formalize this idea in \autoref{thm:qwls-general}, assuming direct access to (i) a block encoding of $B=\sqrt{W}A$, and (ii) a procedure for preparing the state $\ket{b_w}=\frac{\sqrt{W}\ket{b}}{\norm{\sqrt{W}\ket{b}}}$. Subsequently, for the specific input models, we show that we can indeed efficiently obtain a block-encoding of $B$ and prepare the state $\ket{b_w}$.  
}

\begin{theorem}[Quantum Weighted Least Squares with General $\ell_2$-Regularization]
\label{thm:qwls-general}
    Let $A, L \in \rr^{N \times d}$, be the data and penalty matrix, with effective condition numbers $\kappa_A$ and $\kappa_L$, respectively. 
    Let $\lambda \in \rr^+$ be the regularizing parameter. 
    Let $W \in \rr^{N \times N}$ be a diagonal weight matrix with the largest and smallest diagonal entries being $\wmax, \wmin$, respectively.
    {\UPDATED
    Let $U_{B}$ be a $(\alpha_B, a_B, \varepsilon_B)$ block encoding of $B := \sqrt{W}A$ implemented in time $T_B$ and let $U_L$ be a $(\alpha_L, a_L, \varepsilon_L)$ block encoding of $L$ implemented in time $T_L$, such that $\varepsilon_B = \smalloh{\frac{\delta}{\kappa^3\log^2\paren{\frac{\kappa}{\delta}}}}$ and $\varepsilon_L = \smalloh{\frac{\delta}{\sqrt{\lambda}\kappa^3\log^2\paren{\frac{\kappa}{\delta}}}}$.
    Let $U_{b_w}$ be a unitary that prepares $\frac{\sqrt{W}\ket{b}}{\norm{\sqrt{W}\ket{b}}}$ in time $T_{b_w}$.
  Define 
  \begin{equation*}
    \kappa := \kappa_L \paren{1 + \frac{\sqrt{\wmax}\norm{A}}{\sqrt{\lambda}\norm{L}}}
  \end{equation*}
  }
  Then for any $\delta > 0$ we can prepare a quantum state that is $\delta$-close to
    \begin{equation*}
      \normalized{(A^TWA + \lambda L^TL)^{-1} A^TW \ket b}
    \end{equation*}
  with probability $\Theta(1)$,
  at a cost of  
    \begin{equation}
    \UPDATED
        \order{
            \kappa\log\kappa\paren{
                \frac{\alpha_B + \sqrt{\lambda}\alpha_L}{\sqrt{\wmax}\norm{A} + \sqrt{\lambda} \norm{L}} \logp{\frac{\kappa}{\delta}} (T_B + T_L) + T_{b_w}
            }
        },
    \end{equation}
  using only $\order{\log\kappa}$ additional qubits. 
\end{theorem}

\begin{proof}
\UPDATED
    We then invoke \autoref{thm:quantum_least_squares_gen_tik_reg} with $B$ and $L$ as the data and regularization matrices, respectively. This requires that $\varepsilon_{B}, \varepsilon_L$ such that \[\varepsilon_{B} + \sqrt{\lambda}\varepsilon_L = \smalloh{\frac{\delta}{\kappa^3\log^2\paren{\frac{\kappa}{\delta}}}}.\] Thus, we get the upper bounds on the precision $\varepsilon_B,~\varepsilon_L$ required. This gives us a quantum state $\delta$-close to \[\normalized{(A^TWA + \lambda L^TL)^{-1} A^TW \ket b}.\]
\end{proof}

{\UPDATED
Next, we construct the block encodings for $\sqrt{W}A$ and the state $\frac{\sqrt{W}\ket{b}}{\norm{\sqrt{W}\ket{b}}}$ efficiently in the quantum data structure input model. This construction would also apply to the sparse access input model with slight modifications. 
}

\begin{lemma}[Efficiently preparing $\sqrt{W} A$ in the Quantum Data Structure Model]
    \label{lem:wls_qram_sqrtWA}
    Let $W \in \rr^{N \times N}$
    such that $W = \mathrm{diag}(w_1, w_2 \ldots w_N)$
        and $w_{\max} := \max_i w_i$, and $A \in \rr^{N \times d}$
        be stored in a quantum-accessible data structure.
    Then for any $\delta > 0$ there exists a
    $$(\sqrt{w_{\max}}\norm{A}_{F}, \lceil\logp{N + d}\rceil, \delta)$$
    block-encoding of $\sqrt{W}A$ that can be implemented at the cost $\order{\mathrm{polylog}(Nd/\delta)}$.
\end{lemma}

\begin{proof}
    $\forall j \in [N]$, define
    \begin{equation*}
        \ket{\psi_j} := \sqrt{\frac{w_j}{w_{\max}}} \ket{j} \frac{1}{\norm{A_{j, \cdot}}} \sum_{k \in [d]} A_{j, k} \ket{k}.
    \end{equation*}
    
    Similarly, $\forall k \in [d]$, define
    \begin{equation*}
        \ket{\phi_k} := \frac{1}{\norm{A}_F} \paren{\sum_{j \in [N]} \norm{A_{j,\cdot}} \ket{j}} \ket{k}.
    \end{equation*}
    
    Observe that $\forall j \in [N], k \in [d]$,
    \begin{equation*}
        \braket{\psi_j|\phi_k} = \sqrt{\frac{w_j}{w_{\max}}} \frac{A_{j, k}}{\norm{A}_F} = \frac{\bra{j} \sqrt{W}A\ket{k}}{\sqrt{w_{\max}} \norm{A}_F}.
    \end{equation*}
    Given quantum data structure accesses to $W$ and $A$, one can construct quantum circuits $W_R$ and $W_L$ similar to $U_L$ and $U_R$ from \autoref{lem:qrom_be_constr} that prepare $\ket{\phi_k}$ and $\ket{\psi_j}$ above. $\ket{\phi_k}$ can be prepared just as in \autoref{lem:qrom_be_constr}, while $\ket{\psi_j}$ can be prepared using controlled rotations on the state $\ket{\frac{w_j}{w_{\max}}}$ (which can be constructed from the QRAM access to $W$) after adding an ancilla qubit and the QRAM access to $A$. Thus, $W_R^{\dagger}W_L$ is the required block encoding, which according to \autoref{thm:qrom_data_structure} can be implemented using $\mathrm{polylog}(Nd/\delta)$ queries. 
\end{proof}

\begin{lemma}[Efficiently preparing $\sqrt{W}\ket{b}$ in the Quantum Data Structure Model]
    \label{lem:spp_sqrtWb}
  Let $b \in \rr^N$ and $W \in \rr^{N \times N}$. Suppose that $b$ and $W$ are stored in a quantum-accessible data structure such that we have a state preparation procedure that acts as 
  \begin{align*}
      &U_W : \ket{j} \ket{0} \mapsto \ket{j} \ket{w_j}, \\ 
      &U_b : \ket{0} \mapsto \sum_j \frac{b_j}{\norm{b}} \ket{j}.
  \end{align*}
  Then for any $\delta > 0$ we can prepare the quantum state that is $\delta$-close to $\normalized{\sqrt{W}\ket{b}}$
  with constant success probability and at a cost of
  $\order{\sqrt{\frac{w_{\max}}{w_{\min}}}\polylog{\frac{N}{\delta}}}$. 
  
\end{lemma}

\begin{proof}
    Use $U_b$ to prepare the state 
    \begin{equation*}
        \ket{b} = \frac{1}{\norm{b}}\sum_j b_j \ket{j}
    \end{equation*}
    in time polylog$(N)$. Then, apply the following transformation 
    \begin{align*}
        \ket{j}\ket{0}\ket{0}
        &\mapsto \ket{j}\ket{w_j}\ket{0} \\
        &\mapsto \ket{j} \ket{w_j} \paren{\sqrt{\frac{w_j}{w_{\max}}} \ket{0} + \sqrt{1 - \frac{w_j}{w_{\max}}} \ket{1}} \\
        &\mapsto \ket{j} \ket{0} \paren{\sqrt{\frac{w_j}{w_{\max}}} \ket{0} + \sqrt{1 - \frac{w_j}{w_{\max}}} \ket{1}}
    \end{align*}
    which can again be applied using some controlled rotations, a square root circuit and $U_W$. This gives us the state (ignoring some blank registers)
    \begin{equation}
        \sum_j \paren{\sqrt{\frac{w_j}{w_{\max}}} \ket{0} + \sqrt{1 - \frac{w_j}{w_{\max}}} \ket{1}} \frac{b_j}{\norm{b}} \ket{j}.
    \end{equation}
    The probability for the ancilla to be in $\ket{0}$ state is $$\Omega\paren{\frac{w_{\min}}{w_{\max}}}.$$
    Thus performing $\order{\sqrt{\frac{w_{\max}}{w_{\min}}}}$ rounds of amplitude amplification on $\ket{0}$ gives us a constant probability of observing $\ket{0}$, and therefore obtaining the desired state $\normalized{\sqrt{W}\ket{b}}$. 
\end{proof}

Using the above two theorems, and the quantum OLS solver (\autoref{thm:quantum_least_squares_gen_tik_reg}), we can construct an algorithm for regularized quantum WLS. 

\newcommand{\WLSComplexity}[2]{
    \order{
        \kappa \paren{
        \frac
            {\sqrt\wmax #1 + \sqrt\lambda #2}
            {\sqrt\wmax \norm{A}   + \sqrt\lambda \norm{L}} 
        + \sqrt{\frac{\wmax}{\wmin}} }
        \polylog{Nd, \kappa, \frac1\delta}
    }
}

\begin{theorem}[Quantum Weighted Least Squares with General $\ell_2$-Regularization in the Quantum Data Structure Model]
      \label{thm:wls_l2_qram}
  Let $A, L \in \rr^{N \times d}$ with effective condition numbers $\kappa_A, \kappa_L$ respectively be stored in an efficient quantum accessible data structure.
  Let $W \in \rr^{N \times N}$ be a diagonal matrix with largest and smallest singular values $\wmax, \wmin$ respectively, which is also stored in an efficient quantum accessible data structure.
  Furthermore, suppose the entries of the vector $b \in \rr^N$ are also stored in a quantum-accessible data structure and define,
  \begin{equation*}
  \UPDATED
    \kappa := \kappa_L \paren{1 + \frac{\sqrt{\wmax}\norm{A}}{\sqrt{\lambda}\norm{L}}}
  \end{equation*}
  Then for any $\delta > 0$ we can prepare a quantum state that is $\delta$-close to
    \begin{equation*}
      \normalized{(A^TWA + \lambda L^TL)^{-1} A^TW \ket b}
    \end{equation*}
  with probability $\Theta(1)$,
  at a cost of  
  \begin{equation}
  \UPDATED
        \WLSComplexity{\norm{A}_F}{\norm{L}_F}
  \end{equation}
\end{theorem}

\begin{proof}
\UPDATED
    Choose some precision parameter $\varepsilon > 0$ for accessing the data structure.
    Given access to $W$ and $A$,
        we can use \autoref{lem:wls_qram_sqrtWA}
        to prepare a $(\sqrt{\wmax}\norm{A}_F, \lceil\logp{N + d}\rceil, \varepsilon)$-block-encoding
        of $\sqrt{W}A$,
        using $T_A := \order{\polylog{Nd/\varepsilon}}$ queries to the data structure.
    Similarly, 
        \autoref{lem:qrom_be_constr} allows us to build a
        $(\norm{L}_F, \lceil\logp{N+d}\rceil, \varepsilon)$-block-encoding
        of $L$
        using $T_L := \order{\mathrm{polylog}(Nd/\varepsilon)}$ queries to the data structure. 
    
    Next, using \autoref{lem:spp_sqrtWb},
        for any $\varepsilon_b > 0$, 
        we can prepare a state $\varepsilon_b$-close to
        $\ket{b'} := \normalized{\sqrt{W}\ket{b}}$.
        This procedure requires
        $T_b := \order{\sqrt{\frac{w_{\max}}{w_{\min}}}\polylog{N/\varepsilon_b}}$ queries to the data structure.
Now we can invoke the OLS solver in
    \autoref{thm:quantum_least_squares_gen_tik_reg}
    with a precision of $\delta_b$,
    by considering $\sqrt{W}A$ as the data matrix
    and $\normalized{\sqrt{W}\ket{b}}$ as the input state.
In order for the input block-encoding precision to satisfy the bound in \autoref{eqn:qols_error_condition}, we choose $\varepsilon$ such that
    \begin{equation*}
        \varepsilon =
        o\paren{\frac{\delta_b}{\kappa^3 \log^2 \paren{\frac{\kappa}{\delta_b}}}}.
    \end{equation*}
    Finally, for the output state to be $\delta$-close to the required state,
    we choose $\delta_b = \delta/2$ and $\varepsilon_b = \delta/2\kappa$
    to use the robustness result from \autoref{lem:apply_block_enc_approx}.
    This gives us
    \begin{align*}
        \logp{\frac{1}{\varepsilon}}
        &= \order{\logp{
            \frac{\kappa^3\log^2\paren{\frac{\kappa}{\delta_b}}}{\delta_b}
        }} \\
        &= \order{
        \logp{\frac{\kappa}{\delta}}
        }
    \end{align*}
    Now we can substitute the cost of the individual components in \autoref{eqn:ols_complexity_vtqa} to obtain the final cost as
    \begin{align*}
        &\order{
            \kappa\log\kappa \paren{
            \frac{\sqrt\wmax \norm{A}_F + \sqrt\lambda \norm{L}_F}{\sqrt\wmax \norm{A} + \sqrt\lambda \norm{L}} \logp{\frac{\kappa}{\delta}} \mathrm{polylog}\paren{\frac{Nd}{\varepsilon}}
            + \sqrt{\frac{\wmax}{\wmin}} \mathrm{polylog}\paren{\frac{N\kappa}{\delta}}}
        } \\ 
        &= \order{
            \kappa \paren{
            \frac
                {\sqrt\wmax \norm{A}_F + \sqrt\lambda \norm{L}_F}
                {\sqrt\wmax \norm{A}   + \sqrt\lambda \norm{L}} 
            + \sqrt{\frac{\wmax}{\wmin}} }
            \polylog{\frac{Nd\kappa}{\delta}}
        }
        \end{align*}
\end{proof}

Now, for the sparse access model, we can obtain a block encoding similar to \autoref{lem:wls_qram_sqrtWA} and a quantum state similar to \autoref{lem:spp_sqrtWb}, with the same query complexities. 
Thus we have an algorithm similar to \autoref{thm:wls_l2_qram} in the sparse access model as well. 
We directly state the complexity of this algorithm.

\begin{theorem}[Quantum Weighted Least Squares with General $\ell_2$-Regularization in the Sparse Access Model]
      \label{thm:wls_l2_sparse}
  Let $ A \in \rr^{N \times d} $ be $ (s^A_r, s^A_c)$ row-column sparse,
    and similarly, let $ L \in \rr^{N \times d} $ be $ (s^L_r, s^L_c)$ row-column sparse,
    with effective condition numbers $\kappa_A$ and $\kappa_L$ respectively.
    Let $ \lambda \in \mathbb{R}^+ $. Let $W \in \rr^{N \times N}$ be a diagonal matrix with the largest and the smallest diagonal entries being $\wmax, \wmin$, respectively. Suppose that the diagonal entries of $W$ are stored in a QROM such that, for any $\delta>0$, we can compute $\ket{j}{0}\mapsto \ket{j}\ket{w_j}$ in cost $O\left(\polylog{Nd/\delta}\right)$ as well as $\wmax$.
  Furthermore, suppose there exists a unitary that prepares $\ket{b}$ at a cost $T_b$ and define,
  \begin{equation*}
  \UPDATED
    \kappa := \kappa_L \paren{1 + \frac{\sqrt{\wmax}\norm{A}}{\sqrt{\lambda}\norm{L}}}
  \end{equation*}
  Then for any $\delta > 0$ we can prepare a quantum state that is $\delta$-close to
    \begin{equation*}
      \normalized{(A^TWA + \lambda L^TL)^{-1} A^TW \ket b}
    \end{equation*}
    with probability $\Theta(1)$,
  at a cost of  
  \begin{equation}
  \UPDATED
      \order{
        \kappa \paren{
        \frac
            {\sqrt\wmax \sqrt{s_r^A s_c^A} + \sqrt\lambda \sqrt{s_r^L s_c^L}}
            {\sqrt\wmax \norm{A}   + \sqrt\lambda \norm{L}} 
        + \sqrt{\frac{\wmax}{\wmin}}~T_b }
        \polylog{Nd, \kappa, \frac1\delta}
    }
  \end{equation}
\end{theorem}

\subsubsection{Quantum Generalized Least Squares}
\label{subsec:gls-l2}

In this section,
  we assume that we have block-encoded access to the correlation matrix $\Omega \in \rr^{N \times N}$, with condition number $\kappa_\Omega$. We begin by preparing a block encoding of $\Omega^{-1/2}$, given an approximate block-encoding of $\Omega$.

\begin{lemma}[Preparing $\Omega^{-1/2}$]
  \label{lem:gls_inv_sqrt_omega}
  Let $\Omega \in \rr^{N \times N}$ be a matrix 
    with condition number $\kappa_\Omega$
  Let $U_\Omega$ be an $(\alpha_\Omega, a_\Omega, \varepsilon_\Omega)$-block-encoding
   of $\Omega$, implemented in time $T_\Omega$.
  For any $\delta$ such that
  \begin{equation*}
      \UPDATED
      \eps_\Om = \smalloh{
      \frac{\sqrt{\norm{\Omega}}\delta}{\kappa^{1.5}\logp{\frac{\kappa}{\sqrt{\norm{\Omega}}\delta}}}
      },
  \end{equation*}
  we can prepare a 
    $(2\sqrt{\kappa_\Omega/\norm{\Omega}}, a_\Omega + 1, \delta)$-block-encoding
    of $\Omega^{-1/2}$
  at a cost of 
  \begin{equation*}
  \order{\frac{\alpha_\Omega\kappa_\Omega}{\norm{\Omega}}\log\paren{\frac{\kappa_\Omega}{\delta\sqrt{\norm{\Omega}}}} T_\Omega}
  \end{equation*}
  Moreover, the condition number of $\Omega^{-1/2}$ is bounded by $\sqrt{\kappa_\Omega}$.
\end{lemma}

\begin{proof}
  $U_{\Omega}$ can be re-interpreted as a $(\frac{\alpha}{\norm{\Omega}},~a,~\frac{\varepsilon}{\norm{\Omega}})$-block-encoding of $\frac{\Omega}{\norm{\Omega}}$. We can then prepare the required unitary by invoking \autoref{thm:constr_neg_powers_normalized} on $U_\Omega$ with $c = 1/2$ and some $\gamma$ such that we get a $(2\sqrt{\kappa_{\Omega}},~a+1,~\gamma)$ block encoding of $\sqrt{\norm{\Omega}} \Omega^{-1/2}$, which is a $(2\sqrt{\frac{\kappa}{\norm{\Omega}}},~a+1,~\frac{\gamma}{\sqrt{\norm{\Omega}}})$ block encoding of $\Omega^{-1/2}$. Fixing $\gamma = \sqrt{\norm{\Omega}}\delta$ gives us the required result.
\end{proof}

We will now use this lemma in conjunction with \autoref{thm:quantum_least_squares_gen_tik_reg} to develop quantum algorithms for GLS with general $\ell_2$-regularization.

\begin{theorem}[Quantum Generalized Least Squares with General $\ell_2$-regularization]
    \label{thm:gls_l2}
  Let $A, L \in \rr^{N \times d}$ be the data and penalty matrices
      with effective condition numbers $\kappa_A, \kappa_L$ respectively.
  Let $\Omega \in \rr^{N \times N}$ be the covariance matrix
      with condition number $\kappa_\Omega$.
  Let $\delta > 0$ be the precision parameter.
  Define $\kappa$ as \[ \UPDATED \kappa := \GLSKappa. \]
  For some $\varepsilon_A$ such that
      \begin{equation*}
          \UPDATED
          \varepsilon_A = \smalloh{
          \frac
              {\delta \sqrt{\NO}}
              {\kappa^3 \sqrt{\KO} \log^2 \frac\kappa\delta}
          }
      \end{equation*}
      we have access to $U_A$,
      an $(\alpha_A, a_A, \varepsilon_A)$-block-encoding of $A$
      implemented in time $T_A$.
  For some $\varepsilon_L$ such that
      \begin{equation*}
          \UPDATED
          \varepsilon_L = \smalloh{
          \frac
              {\delta}
              {\sqrt\lambda \kappa^3 \log^2 \frac\kappa\delta}
          }
      \end{equation*}
      we have access to $U_L$,
      an $(\alpha_L, a_L, \varepsilon_L)$-block-encoding of $L$
      implemented in time $T_L$.
  For some $\varepsilon_\Omega$ such that
      \begin{equation*}
          \UPDATED
          \eps_\Om = \smalloh{
            \frac
                {\delta}
                {\norm{A} \kappa^3 \KO^{1.5} \log^3{\frac\kappa\delta} \logp{\frac{\KO}{\norm{A}\NO}}}
          }
      \end{equation*}
      we have access to $U_\Omega$,
      an $(\alpha_\Omega, a_\Omega, \varepsilon_\Omega)$-block-encoding of $\Omega$
      implemented in time $T_\Omega$.
  Let $U_b$ be a unitary that prepares the state $\ket b$ in time $T_b$.
  
  Then we can prepare the quantum state that is $\delta$-close to
  \begin{equation*}
    \normalized{\paren{A^T\Omega^{-1}A + \lambda L^TL}^{-1}A^T\Omega^{-1}\ket b}
  \end{equation*}
   with probability $\Theta(1)$,
  at a cost of \begin{equation} \UPDATED \order{\GLSComplexity} \end{equation}
  using only $\order{\log \kappa}$ additional qubits.
\end{theorem}

\begin{proof}
\UPDATED
    Observe that by choosing $A' := \Om^{-1/2} A, L' := L, \ket{b'} := \Om^{-1/2}\ket{b}$ (upto normalization) in the quantum ordinary least squares, we get a state proportional to $(A'^T A' + \lambda L'^T L')^{-1} A'^T \ket{b'} = (A^T\Om^{-1}A + \lambda L^TL) A^T \Om^{-1} \ket b$, which is the desired state.
    
    For convenience, let us define the matrix $B := \Om^{-1/2}$
    (and therefore $\kappa_B = \sqrt{\KO}$ and $\norm{B} = \sqrt{\KO/\NO}$).
    We now need to prepare a block-encoding of $BA$ and the quantum state $\normalized{B \ket b}$,
    which we then use to invoke \autoref{thm:quantum_least_squares_gen_tik_reg}.

    We begin by using \autoref{lem:gls_inv_sqrt_omega}
        with some precision $\eps_B$
        to construct a $(\alpha_B, a_B, \eps_B)$-block-encoding
        of $B = \Om^{-1/2}$,
        where $\alpha_B = 2\sqrt{\frac\KO\NO} = 2\norm{B}$,
        and $a_B = a_\Om + 1$.
    This bounds $\eps_\Om$ as 
    \[
    \varepsilon_{\Om} = \smalloh{
        \frac
            {\sqrt{\NO} \varepsilon_B}
            {\KO^{1.5} \logp{\frac{\KO}{\sqrt{\NO\varepsilon_B}}}
            }
    },
    \]
    and has a cost of
    \[
    T_B := \order{
        \frac{\alpha_\Om \KO}{\NO}
        \logp{\frac{\KO}{\eps_B \sqrt{\NO}}} T_\Om
    }
    \]

  Then using \autoref{lem:prod_of_be_preamp} with precision $\gamma$
  satisfying $\gamma \ge 4\sqrt 2 \max\paren{\norm{B}\eps_A, \norm{A}\eps_B}$,
    we get a\\ $(2\norm{A}\norm{B}, a_A + a_B + 3, \gamma)$-block-encoding
    of $A' := BA = \Omega^{-1/2}A$
    at a cost
    \begin{equation*}
        T_{A'} := \order{
        \paren{\frac{\alpha_A}{\norm{A}} T_A + \frac{\alpha_B}{\norm{B}} T_B}
        \log\paren{\frac{\norm{A}\norm{B}}{\gamma}}}.
    \end{equation*}

  To prepare $\normalized{B\ket b}$, 
  we use \autoref{lem:apply_block_enc}
    with precision $\varepsilon_b \ge 2\varepsilon_B\kappa_B/\norm{B}$.
    This prepares a state that is $\eps_b$-close to $\ket{b'} := \normalized{B\ket b}$
    with constant success probability
    at a cost of
    \begin{equation*}
        T_{b'} := \order{
            \frac{\alpha_B \kappa_B}{\norm{B}}
            (T_B + T_b)
        }
        = \order{\kappa_B (T_B + T_b)}
    \end{equation*}
    
  We could invoke OLS directly using the above two, but that ends up with a product of sub-normalization factors
  ($\alpha$ terms) in the complexity. We want to avoid this, because in most common cases $\alpha$-s for block-encodings are quite large.
  So we also pre-amplify $U_L$ using \autoref{lem:uniform_block_ampl}:
  for any $\delta_L \ge 2\varepsilon_L$
    we get a $(\sqrt{2}\norm{L}, a_L + 1, \delta_L)$-encoding of $L$
    at a cost of
    \begin{equation*}
    T_{L'} := \order{\frac{\alpha_L}{\norm{L}} T_L \logp{\frac{\norm{L}}{\delta_L}}}.
    \end{equation*}

    Now that we have these, we can use \autoref{thm:quantum_least_squares_gen_tik_reg}
        to get a quantum state $\delta'$-close to
        $\ket{\psi} := \normalized{A_L^+\ket{b'}}$,
        where $A^+_L = (A^T\Omega^{-1}A + \lambda L^TL)^{-1}A^T\Omega^{-1/2}$.
    This would require that
        $\gamma, \sqrt{\lambda} \delta_L \in \smalloh{\frac{\delta'}{\kappa^3 \log^2 \paren{\frac{\kappa}{\delta'}}}}$ and would cost 
    \[ \order{\kappa\log\kappa
        \paren{
            \paren{
                \frac
                    {2\norm{A}\norm{B} + \sqrt{2\lambda} \norm{L}}
                    {\norm{BA} + \sqrt{\lambda} \norm{L}}
            } 
            \logp{\frac{\kappa}{\delta'}} (T_{A'} + T_{L'})
            + T_{b'}
        }
    }. \]

    To simplify the ratio of norms term,
    we can first lower-bound $\norm{BA} \ge \norm{A}/\norm{B^{-1}} = \norm{A}/\sqrt\NO$.
    And as $\norm{B} = \sqrt{\KO/\NO}$,
    the whole term can be simplified to $\order{\sqrt{\KO}}$.
    This simplifies the cost expression to
    $\order{\kappa\log\kappa \paren{\sqrt{\KO} \logp{\kappa/{\delta'}}(T_{A'} + T_{L'}) + T_{b'}}}$.
    
  We can compute the error between $\ket\psi$ and the expected state by using \autoref{lem:apply_block_enc_approx}.
  For the final error to be $\delta$, we have to choose $\varepsilon_b = \delta/2\kappa$ and $\delta' = \delta / 2$.
  Therefore
  \begin{align*}
      \eps_B 
      \le \frac{\eps_b \norm{B}}{4\kappa_B}
      = \Theta\paren{\frac{\delta}{\kappa \sqrt{\NO}}}
  \end{align*}
  \begin{align*}
    &\gamma, \sqrt\lambda \delta_L \in \smalloh{\frac{\delta}{\kappa^3 \log^2 (\kappa/\delta)}}
    \\ \implies&
    \logp{\frac1\gamma} = \smalloh{\logp{\frac\kappa\delta}},\;\;
    \logp{\frac1{\delta_L}} = \smalloh{\logp{\frac{\sqrt\lambda\kappa}\delta}} 
  \end{align*}
  \[
      \eps_A = \smalloh{\frac{\gamma}{\norm{B}}}, \;\;
      \eps_B = \smalloh{\frac{\gamma}{\norm{A}}}
  \]
  Combining both bounds of $\eps_B$ by using sums or products, we can effectively bound
  \[
  \eps_\Om = \smalloh{
    \frac
        {\delta}
        {\norm{A} \kappa^3 \KO^{1.5} \log^3{\frac\kappa\delta} \logp{\frac{\KO}{\norm{A}\NO}}}
  }
  \]

\newcommand{\coeffSymb}{\mathcal{C}}
\newcommand{\coeff}[2]{\coeffSymb_{#1}(#2)}
  Finally for the final costs,
  we calculate the respective coefficients of terms $T_A, T_\Om, T_L$ and $T_b$,
  (excluding the common factor of $\kappa\sqrt{\KO}\log\kappa$ for brevity).
  Let us label these ``coefficient extraction'' functions as $\coeffSymb$ with matching subscripts,
  and the total cost as $T$.

  \begin{align*}
  \coeff{A}{T}
      &= \order{\logp{\frac\kappa\delta} \coeff{A}{T_{A'}}} \\
      &= \order{\logp{\frac\kappa\delta} \frac{\alpha_A}{\norm{A}} \logp{\frac{\norm{A}\norm{B}}{\gamma}}} \\
      &= \order{\frac{\alpha_A}{\norm{A}}
          \log^2\paren{\frac{\kappa \KO \norm{A}}{\delta \NO}}}\\
  \coeff{L}{T}
      &= \order{\logp{\frac\kappa\delta} \coeff{L}{T_{L'}}} \\
      &= \order{\logp{\frac\kappa\delta} 
          \frac{\alpha_L}{\norm{L}} \logp{\frac{\norm{L}}{\delta_L}}} \\
      &= \order{\frac{\alpha_L}{\norm{L}} \log^2 \paren{\frac{\kappa\norm{L}}\delta}} \\
  \coeff{\Om}{T}
      &= \order{\logp{\frac\kappa\delta} \coeff{\Om}{T_{A'}}
      + \frac{\coeff{\Om}{T_{b'}}}{\sqrt{\KO}}
      } \\
      &= \order{\paren{
          \logp{\frac\kappa\delta} \logp{\frac{\norm{A}\norm{B}}{\gamma}}
          + 1
          } \coeff{\Om}{T_B}
       }\\
      &= \order{
          \log^2\paren{\frac{\kappa \KO \norm{A}}{\delta \NO}}
          \frac{\alpha_\Om \KO}{\NO} \logp{\frac{\KO}{\eps_B \sqrt{\NO}}}
          }\\
      &= \order{
          \frac{\alpha_\Om \KO}{\NO}
          \log^3\paren{\frac{\kappa \KO \norm{A}}{\delta \NO}}
          }\\
  \coeff{b}{T}
        &= \order{\frac{\coeff{\Om}{T_{b'}}}{\sqrt{\KO}}}
        = \order{1} 
  \end{align*}
  And hence the final complexity is given by the expression
  \begin{align*}
  T &= \order{\kappa \sqrt{\KO} \log\kappa \paren{
          \coeff{A}{T} \cdot T_A
          + \coeff{L}{T}\cdot T_L
          + \coeff{\Om}{T}\cdot T_\Om
          + \coeff{b}{T}\cdot T_b
          }} \\
    &= \order{\kappa \sqrt{\KO} \log\kappa \paren{
          \frac{\alpha_A}{\norm{A}}
          \log^2\paren{\frac{\kappa \KO \norm{A}}{\delta \NO}} T_A
          + \frac{\alpha_L}{\norm{L}} \log^2 \paren{\frac{\kappa\norm{L}}\delta} T_L
          + \frac{\alpha_\Om \KO}{\NO} \log^3\paren{\frac{\kappa \KO \norm{A}}{\delta \NO}} T_\Om
          +  T_b
          }} \\
    &= \order{\GLSComplexity}
  \end{align*}
\end{proof}

One immediate observation is that for the special case of the (unregularized) quantum GLS problem (when $L=0$ and $\lambda=0$), our algorithm has a slightly better complexity than \cite{CGJ19} and requires fewer additional qubits.  
Now, we will state the complexities of this algorithm in specific input models, namely the quantum data structure model and the sparse-access input model.

\newcommand{\GLSComplexitySpecial}[3]{
    \order{
        \kappa \sqrt{\KO}
        \paren{
            \frac{#1}{\norm{A}} 
            + \frac{#2}{\norm{L}} 
            + \frac{\kappa_\Omega #3}{\norm{\Omega}}
        }
        \polylog{Nd, \kappa, \frac1\delta, \frac\KO\NO, \norm{A}, \norm{L}, \lambda}
    }
}
\begin{corollary}[Quantum Generalized Least Squares with General $\ell_2$-Regularization in the Quantum Data Structure Model]
    \label{cor:gls_data_structure}
  Let $A, L \in \rr^{N \times d}$ be the data and penalty matrices
      with effective condition numbers $\kappa_A, \kappa_L$ respectively.
  and $\Omega \in \rr^{N \times N}$ be the covariance matrix
      with condition number $\kappa_\Omega$.
  Let the matrices $A, L, \Omega$ and the vector $b$
  be stored in a quantum-accessible data structure.
  Define $\kappa$ as \[ \kappa := \GLSKappa \]
  Then for any $\delta > 0$,
  we can prepare the quantum state that is $\delta$-close to
  \begin{equation*}
    \normalized{\paren{A^T\Omega^{-1}A + \lambda L^TL}^{-1}A^T\Omega^{-1}\ket b}
  \end{equation*}
  with probability $\Theta(1)$,
  at a cost of
  \begin{equation}
      \UPDATED
      \GLSComplexitySpecial{\mu_A}{\mu_L}{\mu_\Om}
  \end{equation}
\end{corollary}

\begin{proof}
    The proof is very similar to \autoref{cor:quantum_reg_in_qrom} with the extra 
    input of $\Omega$.
    We can use the data structure to prepare the block-encodings 
    for $A, L, \Omega$ and the state $\ket{b}$,
    with precisions $\varepsilon_A, \varepsilon_L, \varepsilon_\Omega, \varepsilon_b$ respectively.
    We invoke \autoref{thm:gls_l2} with a precision of $\delta_b$,
    and choose the above $\varepsilon$ terms to be equal to their corresponding upper-bounds.
    And finally we use \autoref{lem:apply_block_enc_approx} with $\varepsilon_b = \delta/2\kappa$ and $\delta_b = \delta/2$ to get the final error as $\delta$.
\end{proof}

Now, $\mu_A=\norm{A}_F$ (similarly for $\mu_L$ and $\mu_\Omega$). As $\norm{A}_F\leq \sqrt{r(A)}\norm{A}$, where $r(A)$ is the rank of $A$, we have that the complexity of \autoref{cor:gls_data_structure} can be re-expressed as
\begin{equation}
\UPDATED
 \order{
        \kappa
        \sqrt{\kappa_{\Omega}}
        \paren{
            \sqrt{r(A)} 
            + \sqrt{r(L)} 
            + \sqrt{r(\Omega)}\kappa_\Omega
        }
        \polylog{\frac{Nd\kappa}{\delta}}
    }.
\end{equation}

\begin{corollary}[Quantum Generalized Least Squares with General $\ell_2$-Regularization in the Sparse Access Model]
    Let $A \in \rr^{N \times d}$ be a $(s^A_r, s^A_c)$ row-column sparse data matrix.
    Let $L \in \rr^{N \times d}$ be a $(s^L_r, s^L_c)$ row-column sparse penaly matrix.
    Let $\Omega \in \rr^{N \times N}$ be a $(s^\Omega_r, s^\Omega_c)$ row-column sparse covariance matrix.
    Suppose we have a procedure to prepare $\ket b$ in cost $T_b$. 
  Define $\kappa$ as \[ \kappa := \GLSKappa \]
  Then for any $\delta > 0$,
  we can prepare the quantum state that is $\delta$-close to
  \begin{equation*}
    \normalized{\paren{A^T\Omega^{-1}A + \lambda L^TL}^{-1}A^T\Omega^{-1}\ket b}
  \end{equation*}
  with probability $\Theta(1)$,
  at a cost of
  \begin{equation}
      \UPDATED
      \order{
        \kappa \sqrt{\KO}
        \paren{
            \frac{\sqrt{s^A_r s^A_c}}{\norm{A}} 
            + \frac{\sqrt{s^L_r s^L_c}}{\norm{L}} 
            + \frac{\kappa_\Omega \sqrt{s^\Om_r s^\Om_c}}{\norm{\Omega}}
            + T_b
        }
        \polylog{Nd, \kappa, \frac1\delta, \frac\KO\NO, \norm{A}, \norm{L}, \lambda}
      }
  \end{equation}
\end{corollary}
\begin{proof}
    The algorithm is similar to \autoref{cor:gls_data_structure}, but with
    $\alpha_A = \sqrt{s^A_r s^A_c}$,
    $\alpha_L = \sqrt{s^L_r s^L_c}$,
    $\alpha_\Omega = \sqrt{s^\Omega_r s^\Omega_c}$.
\end{proof}

\section{Future Directions}
\label{sec:discussion}

Our algorithms for quantum linear regression with general $\ell_2$-regularization made use of QSVT to implement various several matrix operations. However, it is possible to use QSVT directly to obtain the solution to \textit{quantum ridge regression}. This requires computing a polynomial approximation for the transformation $\sigma \mapsto \sigma/(\sigma^2+\lambda)$, to be applied on the singular values of $A$, which lie between $[1/\kappa_A,1]$.
{\UPDATED
However, it is unclear how to extend this while considering general $\ell_2$-regularization. For instance, even when the data matrix and the penalty matrix share the same right singular vectors, this approach involves
}
obtaining polynomial approximations to directly implement transformations of the form $\sigma\mapsto \sigma/(\sigma^2+\lambda\widetilde{\sigma}^2)$, where $\widetilde{\sigma}$ is a singular value of the penalty matrix $L$. A monomial is no longer sufficient to approximate this quantum singular value transformation. It would be interesting to explore whether newly developed ideas of M-QSVT \cite{rossi2022multivariable} can be used to implement such transformations directly with improved complexity.

While developing quantum machine learning algorithms, it is essential to point out the caveats, even at the risk of being repetitive \cite{Aaronson2015readTheFinePrint}. Our quantum algorithms output a quantum state $\ket{x}$ whose amplitudes encode the solution of the classical (regularized) linear regression problem. While given access to the data matrix and the penalty matrix, we achieve an exponential advantage over classical algorithms, this advantage is not generic. If similar assumptions ($\ell_2$-sample and query access) are provided to a classical device, Gily\'{e}n et al.~developed a quantum algorithm \cite{Gilyen2020AnIQ} for ridge regression (building upon \cite{chia2020sampling}) which has a running time in $\order{\mathrm{poly}(\kappa, \mathrm{rank}(A), 1/\delta)}$. This implies that any quantum algorithm for this problem can be at most polynomially faster in $\kappa$ under these assumptions. One might posit that similar quantum-inspired classical algorithms for general $\ell_2$-regression can also be developed. The exponential quantum speedup, however, is retained when the underlying matrices are sparse.

Another future direction of research would be to recast our algorithms in the framework of adiabatic quantum computing (AQC) following the works of \cite{Lin2020optimalpolynomial, an2020quantum}. Quantum algorithms for linear systems in this framework have the advantage that a linear dependence on $\kappa$ can be obtained without using complicated subroutines like variable-time amplitude amplification. The strategy is to implement these problems in the AQC model and then use time-dependent Hamiltonian simulation \cite{low2018hamiltonian} to obtain their complexities in the circuit model. One caveat is that, so far, time-dependent Hamiltonian simulation algorithms have only been developed in the sparse-access model and therefore the advantage of the generality of the block-encoding framework is lost. 

In the future, it would also be interesting to explore other quantum algorithms for machine learning such as principal component regression and linear support vector machines \cite{rebentrost2014quantum} using QSVT. Finally, following the results of \cite{chen2021quantum}, it would be interesting to investigate techniques for quantum machine learning that do not require the quantum linear systems algorithm as a subroutine.

\section*{Acknowledgements}
SC thanks Andr\'{a}s Gily\'{e}n, Stacey Jeffery and J\'{e}r\'{e}mie Roland for useful discussions. SC acknowledges funding from the Science and Engineering Board, Department of Science and Technology (SERB-DST), Government of India via grant number SRG/2022/000354. SC is also supported by IIIT Hyderabad via the Faculty Seed Grant. 
AP thanks Michael Walter for useful discussions. AP acknowledges support by the BMBF through project Quantum Methods and Benchmarks for Resource Allocation (QuBRA).

\bibliographystyle{alphaurl}
\bibliography{references}

\newcommand{\etalchar}[1]{$^{#1}$}
\begin{thebibliography}{MRTC21}

\bibitem[Aar15]{Aaronson2015readTheFinePrint}
Scott Aaronson.
\newblock Read the fine print.
\newblock {\em Nature Physics}, 11(4):291--293, Apr 2015.
\newblock \href {https://doi.org/10.1038/nphys3272} {\path{doi:10.1038/nphys3272}}.

\bibitem[AL22]{an2020quantum}
Dong An and Lin Lin.
\newblock Quantum linear system solver based on time-optimal adiabatic quantum computing and quantum approximate optimization algorithm.
\newblock {\em ACM Transactions on Quantum Computing}, 3(2), mar 2022.
\newblock \href {https://doi.org/10.1145/3498331} {\path{doi:10.1145/3498331}}.

\bibitem[Amb12]{ambainis2010variable}
Andris Ambainis.
\newblock {Variable time amplitude amplification and quantum algorithms for linear algebra problems}.
\newblock In Christoph D{\"u}rr and Thomas Wilke, editors, {\em 29th International Symposium on Theoretical Aspects of Computer Science (STACS 2012)}, volume~14 of {\em Leibniz International Proceedings in Informatics (LIPIcs)}, pages 636--647, Dagstuhl, Germany, 2012. Schloss Dagstuhl--Leibniz-Zentrum fuer Informatik.
\newblock \href {https://doi.org/10.4230/LIPIcs.STACS.2012.636} {\path{doi:10.4230/LIPIcs.STACS.2012.636}}.

\bibitem[Bis95]{bishop1995training}
Chris~M. Bishop.
\newblock Training with noise is equivalent to tikhonov regularization.
\newblock {\em Neural Computation}, 7(1):108--116, 1995.
\newblock \href {https://doi.org/10.1162/neco.1995.7.1.108} {\path{doi:10.1162/neco.1995.7.1.108}}.

\bibitem[CdW21]{chen2021quantum}
Yanlin Chen and Ronald de~Wolf.
\newblock Quantum algorithms and lower bounds for linear regression with norm constraints.
\newblock {\em arXiv preprint}, 2021.
\newblock \href {https://doi.org/10.48550/ARXIV.2110.13086} {\path{doi:10.48550/ARXIV.2110.13086}}.

\bibitem[CGJ18]{CGJ19full}
Shantanav Chakraborty, András Gilyén, and Stacey Jeffery.
\newblock The power of block-encoded matrix powers: Improved regression techniques via faster hamiltonian simulation.
\newblock {\em arXiv preprint}, 2018.
\newblock \href {https://doi.org/10.48550/arXiv.1804.01973} {\path{doi:10.48550/arXiv.1804.01973}}.

\bibitem[CGJ19]{CGJ19}
Shantanav Chakraborty, Andr{\'a}s Gily{\'e}n, and Stacey Jeffery.
\newblock {The Power of Block-Encoded Matrix Powers: Improved Regression Techniques via Faster Hamiltonian Simulation}.
\newblock In Christel Baier, Ioannis Chatzigiannakis, Paola Flocchini, and Stefano Leonardi, editors, {\em 46th International Colloquium on Automata, Languages, and Programming (ICALP 2019)}, volume 132 of {\em Leibniz International Proceedings in Informatics (LIPIcs)}, pages 33:1--33:14, Dagstuhl, Germany, 2019. Schloss Dagstuhl--Leibniz-Zentrum fuer Informatik.
\newblock \href {https://doi.org/10.4230/LIPIcs.ICALP.2019.33} {\path{doi:10.4230/LIPIcs.ICALP.2019.33}}.

\bibitem[CGL{\etalchar{+}}20]{chia2020sampling}
Nai-Hui Chia, Andr\'{a}s Gily\'{e}n, Tongyang Li, Han-Hsuan Lin, Ewin Tang, and Chunhao Wang.
\newblock {\em Sampling-Based Sublinear Low-Rank Matrix Arithmetic Framework for Dequantizing Quantum Machine Learning}, page 387–400.
\newblock Association for Computing Machinery, New York, NY, USA, 2020.
\newblock \href {https://doi.org/10.1145/3357713.3384314} {\path{doi:10.1145/3357713.3384314}}.

\bibitem[CKS17]{CKS17}
Andrew~M. Childs, Robin Kothari, and Rolando~D. Somma.
\newblock Quantum algorithm for systems of linear equations with exponentially improved dependence on precision.
\newblock {\em SIAM Journal on Computing}, 46(6):1920–1950, Jan 2017.
\newblock \href {https://doi.org/10.1137/16M1087072} {\path{doi:10.1137/16M1087072}}.

\bibitem[CYGL22]{chen2022faster}
Menghan Chen, Chaohua Yu, Gongde Guo, and Song Lin.
\newblock Faster quantum ridge regression algorithm for prediction.
\newblock {\em International Journal of Machine Learning and Cybernetics}, Apr 2022.
\newblock \href {https://doi.org/10.1007/s13042-022-01526-6} {\path{doi:10.1007/s13042-022-01526-6}}.

\bibitem[EHN96]{engl1996regularization}
H.W. Engl, M.~Hanke, and A.~Neubauer.
\newblock {\em Regularization of Inverse Problems}.
\newblock Mathematics and Its Applications. Springer Netherlands, 1996.
\newblock URL: \url{https://link.springer.com/book/9780792341574}.

\bibitem[GHO99]{golub1999tikhonov}
Gene~H. Golub, Per~Christian Hansen, and Dianne~P. O'Leary.
\newblock Tikhonov regularization and total least squares.
\newblock {\em SIAM Journal on Matrix Analysis and Applications}, 21(1):185--194, 1999.
\newblock \href {https://doi.org/10.1137/S0895479897326432} {\path{doi:10.1137/S0895479897326432}}.

\bibitem[GSLW18]{GSLW2019full}
Andr{\'{a}}s Gily{\'{e}}n, Yuan Su, Guang~Hao Low, and Nathan Wiebe.
\newblock Quantum singular value transformation and beyond: exponential improvements for quantum matrix arithmetics.
\newblock {\em arXiv preprint}, jun 2018.
\newblock \href {https://doi.org/10.48550/arXiv.1806.01838} {\path{doi:10.48550/arXiv.1806.01838}}.

\bibitem[GSLW19]{GSLW2019}
Andr\'{a}s Gily\'{e}n, Yuan Su, Guang~Hao Low, and Nathan Wiebe.
\newblock Quantum singular value transformation and beyond: Exponential improvements for quantum matrix arithmetics.
\newblock In {\em Proceedings of the 51st Annual ACM SIGACT Symposium on Theory of Computing}, STOC 2019, page 193–204, New York, NY, USA, 2019. Association for Computing Machinery.
\newblock \href {https://doi.org/10.1145/3313276.3316366} {\path{doi:10.1145/3313276.3316366}}.

\bibitem[GST22]{Gilyen2020AnIQ}
Andr{\'{a}}s Gily{\'{e}}n, Zhao Song, and Ewin Tang.
\newblock An improved quantum-inspired algorithm for linear regression.
\newblock {\em {Quantum}}, 6:754, June 2022.
\newblock \href {https://doi.org/10.22331/q-2022-06-30-754} {\path{doi:10.22331/q-2022-06-30-754}}.

\bibitem[GTC19]{ge2019faster}
Yimin Ge, Jordi Tura, and J.~Ignacio Cirac.
\newblock Faster ground state preparation and high-precision ground energy estimation with fewer qubits.
\newblock {\em Journal of Mathematical Physics}, 60(2):022202, 2019.
\newblock \href {https://doi.org/10.1063/1.5027484} {\path{doi:10.1063/1.5027484}}.

\bibitem[Hem75]{hemmerle1975explicit}
William~J. Hemmerle.
\newblock An explicit solution for generalized ridge regression.
\newblock {\em Technometrics}, 17(3):309--314, 1975.
\newblock URL: \url{http://www.jstor.org/stable/1268066}, \href {https://doi.org/10.2307/1268066} {\path{doi:10.2307/1268066}}.

\bibitem[HH93]{hanke1993regularization}
Martin Hanke and Per~Christian Hansen.
\newblock Regularization methods for large-scale problems.
\newblock {\em Surv. Math. Ind}, 3(4):253--315, 1993.

\bibitem[HHL09]{HHL2009}
Aram~W. Harrow, Avinatan Hassidim, and Seth Lloyd.
\newblock Quantum algorithm for linear systems of equations.
\newblock {\em Physical Review Letters}, 103(15), Oct 2009.
\newblock \href {https://doi.org/10.1103/physrevlett.103.150502} {\path{doi:10.1103/physrevlett.103.150502}}.

\bibitem[HK00]{hoerl1970ridge}
Arthur~E. Hoerl and Robert~W. Kennard.
\newblock Ridge regression: Biased estimation for nonorthogonal problems.
\newblock {\em Technometrics}, 42(1):80--86, 2000.
\newblock URL: \url{http://www.jstor.org/stable/1271436}, \href {https://doi.org/10.2307/1271436} {\path{doi:10.2307/1271436}}.

\bibitem[KKR06]{kempe2006complexity}
Julia Kempe, Alexei Kitaev, and Oded Regev.
\newblock The complexity of the local hamiltonian problem.
\newblock {\em SIAM Journal on Computing}, 35(5):1070--1097, 2006.
\newblock \href {https://doi.org/10.1137/S0097539704445226} {\path{doi:10.1137/S0097539704445226}}.

\bibitem[KP17]{kerenidis2016quantum}
Iordanis Kerenidis and Anupam Prakash.
\newblock {Quantum Recommendation Systems}.
\newblock In {\em 8th Innovations in Theoretical Computer Science Conference (ITCS 2017)}, volume~67 of {\em Leibniz International Proceedings in Informatics (LIPIcs)}, pages 49:1--49:21, Dagstuhl, Germany, 2017. Schloss Dagstuhl--Leibniz-Zentrum fuer Informatik.
\newblock \href {https://doi.org/10.4230/LIPIcs.ITCS.2017.49} {\path{doi:10.4230/LIPIcs.ITCS.2017.49}}.

\bibitem[KP20]{KP2020_iterative_quantum_gradient}
Iordanis Kerenidis and Anupam Prakash.
\newblock Quantum gradient descent for linear systems and least squares.
\newblock {\em Phys. Rev. A}, 101:022316, Feb 2020.
\newblock \href {https://doi.org/10.1103/PhysRevA.101.022316} {\path{doi:10.1103/PhysRevA.101.022316}}.

\bibitem[LC17a]{low2017hamiltonian}
Guang~Hao Low and Isaac~L Chuang.
\newblock Hamiltonian simulation by uniform spectral amplification.
\newblock {\em arXiv:1707.05391}, 2017.
\newblock \href {https://doi.org/10.48550/ARXIV.1707.05391} {\path{doi:10.48550/ARXIV.1707.05391}}.

\bibitem[LC17b]{LC16b}
Guang~Hao Low and Isaac~L. Chuang.
\newblock Optimal hamiltonian simulation by quantum signal processing.
\newblock {\em Phys. Rev. Lett.}, 118:010501, Jan 2017.
\newblock \href {https://doi.org/10.1103/PhysRevLett.118.010501} {\path{doi:10.1103/PhysRevLett.118.010501}}.

\bibitem[LC19]{LC16}
Guang~Hao Low and Isaac~L. Chuang.
\newblock Hamiltonian simulation by qubitization.
\newblock {\em Quantum}, 3:163, Jul 2019.
\newblock \href {https://doi.org/10.22331/q-2019-07-12-163} {\path{doi:10.22331/q-2019-07-12-163}}.

\bibitem[LMR14]{lloyd2014quantum}
Seth Lloyd, Masoud Mohseni, and Patrick Rebentrost.
\newblock Quantum principal component analysis.
\newblock {\em Nature Physics}, 10(9):631--633, Sep 2014.
\newblock \href {https://doi.org/10.1038/nphys3029} {\path{doi:10.1038/nphys3029}}.

\bibitem[Low17]{low2017quantum}
Guang~Hao Low.
\newblock {\em Quantum signal processing by single-qubit dynamics}.
\newblock PhD thesis, Massachusetts Institute of Technology, 2017.
\newblock URL: \url{http://hdl.handle.net/1721.1/115025}.

\bibitem[LT20a]{lin2020near}
Lin Lin and Yu~Tong.
\newblock Near-optimal ground state preparation.
\newblock {\em {Quantum}}, 4:372, December 2020.
\newblock \href {https://doi.org/10.22331/q-2020-12-14-372} {\path{doi:10.22331/q-2020-12-14-372}}.

\bibitem[LT20b]{Lin2020optimalpolynomial}
Lin Lin and Yu~Tong.
\newblock Optimal polynomial based quantum eigenstate filtering with application to solving quantum linear systems.
\newblock {\em {Quantum}}, 4:361, November 2020.
\newblock \href {https://doi.org/10.22331/q-2020-11-11-361} {\path{doi:10.22331/q-2020-11-11-361}}.

\bibitem[LW18]{low2018hamiltonian}
Guang~Hao Low and Nathan Wiebe.
\newblock Hamiltonian simulation in the interaction picture.
\newblock {\em arXiv preprint}, 2018.
\newblock \href {https://doi.org/10.48550/arXiv.1805.00675} {\path{doi:10.48550/arXiv.1805.00675}}.

\bibitem[LYC16]{LYC16}
Guang~Hao Low, Theodore~J. Yoder, and Isaac~L. Chuang.
\newblock Methodology of resonant equiangular composite quantum gates.
\newblock {\em Phys. Rev. X}, 6:041067, Dec 2016.
\newblock \href {https://doi.org/10.1103/PhysRevX.6.041067} {\path{doi:10.1103/PhysRevX.6.041067}}.

\bibitem[Mar70]{marquaridt1970generalized}
Donald~W. Marquardt.
\newblock Generalized inverses, ridge regression, biased linear estimation, and nonlinear estimation.
\newblock {\em Technometrics}, 12(3):591--612, 1970.
\newblock URL: \url{http://www.jstor.org/stable/1267205}, \href {https://doi.org/10.2307/1267205} {\path{doi:10.2307/1267205}}.

\bibitem[MRTC21]{Grand_Uni_2021}
John~M. Martyn, Zane~M. Rossi, Andrew~K. Tan, and Isaac~L. Chuang.
\newblock Grand unification of quantum algorithms.
\newblock {\em PRX Quantum}, 2:040203, Dec 2021.
\newblock \href {https://doi.org/10.1103/PRXQuantum.2.040203} {\path{doi:10.1103/PRXQuantum.2.040203}}.

\bibitem[Mur12]{murphy2012machine}
Kevin~P Murphy.
\newblock {\em Machine learning: a probabilistic perspective}.
\newblock MIT press, 2012.
\newblock URL: \url{https://mitpress.mit.edu/books/machine-learning-1}.

\bibitem[Pra14]{prakash2014quantum}
Anupam Prakash.
\newblock {\em Quantum Algorithms for Linear Algebra and Machine Learning.}
\newblock PhD thesis, EECS Department, University of California, Berkeley, Dec 2014.
\newblock URL: \url{http://www2.eecs.berkeley.edu/Pubs/TechRpts/2014/EECS-2014-211.html}.

\bibitem[RC22]{rossi2022multivariable}
Zane~M. Rossi and Isaac~L. Chuang.
\newblock Multivariable quantum signal processing ({M}-{QSP}): prophecies of the two-headed oracle.
\newblock {\em {Quantum}}, 6:811, September 2022.
\newblock \href {https://doi.org/10.22331/q-2022-09-20-811} {\path{doi:10.22331/q-2022-09-20-811}}.

\bibitem[RML14]{rebentrost2014quantum}
Patrick Rebentrost, Masoud Mohseni, and Seth Lloyd.
\newblock Quantum support vector machine for big data classification.
\newblock {\em Phys. Rev. Lett.}, 113:130503, Sep 2014.
\newblock \href {https://doi.org/10.1103/PhysRevLett.113.130503} {\path{doi:10.1103/PhysRevLett.113.130503}}.

\bibitem[SSP16]{schuld2016predictionByRegression}
Maria Schuld, Ilya Sinayskiy, and Francesco Petruccione.
\newblock Prediction by linear regression on a quantum computer.
\newblock {\em Phys. Rev. A}, 94:022342, Aug 2016.
\newblock \href {https://doi.org/10.1103/PhysRevA.94.022342} {\path{doi:10.1103/PhysRevA.94.022342}}.

\bibitem[SX20]{Shao2020}
Changpeng Shao and Hua Xiang.
\newblock Quantum regularized least squares solver with parameter estimate.
\newblock {\em Quantum Information Processing}, 19(4):113, Feb 2020.
\newblock \href {https://doi.org/10.1007/s11128-020-2615-9} {\path{doi:10.1007/s11128-020-2615-9}}.

\bibitem[Tan19]{tang_reco_19}
Ewin Tang.
\newblock A quantum-inspired classical algorithm for recommendation systems.
\newblock In {\em Proceedings of the 51st Annual ACM SIGACT Symposium on Theory of Computing}, STOC 2019, page 217–228, New York, NY, USA, 2019. Association for Computing Machinery.
\newblock \href {https://doi.org/10.1145/3313276.3316310} {\path{doi:10.1145/3313276.3316310}}.

\bibitem[Vin78]{vinod1978survey}
Hrishikesh~D. Vinod.
\newblock A survey of ridge regression and related techniques for improvements over ordinary least squares.
\newblock {\em The Review of Economics and Statistics}, 60(1):121--131, 1978.
\newblock URL: \url{http://www.jstor.org/stable/1924340}, \href {https://doi.org/10.2307/1924340} {\path{doi:10.2307/1924340}}.

\bibitem[vW15]{vanwieringen2021lecture}
Wessel~N. van Wieringen.
\newblock Lecture notes on ridge regression, 2015.
\newblock \href {https://doi.org/10.48550/ARXIV.1509.09169} {\path{doi:10.48550/ARXIV.1509.09169}}.

\bibitem[Wan17]{wang_2017_linear_regression}
Guoming Wang.
\newblock Quantum algorithm for linear regression.
\newblock {\em Phys. Rev. A}, 96:012335, Jul 2017.
\newblock \href {https://doi.org/10.1103/PhysRevA.96.012335} {\path{doi:10.1103/PhysRevA.96.012335}}.

\bibitem[WBL12]{wiebe2012quantum}
Nathan Wiebe, Daniel Braun, and Seth Lloyd.
\newblock Quantum algorithm for data fitting.
\newblock {\em Phys. Rev. Lett.}, 109:050505, Aug 2012.
\newblock \href {https://doi.org/10.1103/PhysRevLett.109.050505} {\path{doi:10.1103/PhysRevLett.109.050505}}.

\bibitem[YGW21]{yu2019improved}
Chao-Hua Yu, Fei Gao, and Qiao-Yan Wen.
\newblock An improved quantum algorithm for ridge regression.
\newblock {\em IEEE Transactions on Knowledge and Data Engineering}, 33(3):858--866, 2021.
\newblock \href {https://doi.org/10.1109/TKDE.2019.2937491} {\path{doi:10.1109/TKDE.2019.2937491}}.

\end{thebibliography}

\appendix

\section{Algorithmic Primitives}
\label{app:algo_primitives}

This appendix contains detailed proofs for certain lemmas and corollaries in \autoref{sec:algorithmic_primitives}, for completeness. These proofs are not necessary to understand the techniques and results of the paper, but may help the reader develop a better intuition for the methods used.

\thmbodyUniformBlockAmp*
\begin{proof}
  We can re-interpret $U$ as a
    $(\alpha/\norm{A}, a, \varepsilon/\norm{A})$-block-encoding
    of $A/\norm{A}$.
  Invoking \autoref{lem:uniform_block_ampl_contractions}
    with $\gamma = \frac{\delta}{2\norm{A}}$,
    we get $U'$, a $(\sqrt 2, a + 1, \varepsilon/\norm{A} + \frac{\delta}{2\norm{A}})$-block-encoding
    of $A/\norm{A}$,
    implemented at a cost of $\order{\frac{\alpha}{\norm{A}} T_U \log\paren{\norm{A}/\delta}}$ which is a
    $(\sqrt 2 \norm{A}, a + 1, \delta)$-block-encoding of $A$.
\end{proof}

\thmbodyApplyBlockEnc*
\begin{proof}
    The proof is similar to Lemma 24 of \cite{CGJ19}.
    We have $\norm{A\ket{b}}\geq\frac{\norm{A}}{\kappa}$.
    By applying $U_A$ to $\ket{0}\ket{b}$ (implementable at a cost of $T_A+T_B$), 
    followed by $\frac{\alpha\kappa}{\norm{A}}$-rounds of amplitude amplification (conditioned on having $\ket{0}$ in the first register) , we obtain a quantum state that within $\delta$ of $\ket{0}\otimes\frac{A\ket{b}}{\norm{A\ket{b}}}$.
\end{proof}

\thmbodyApplyBlockEncPreamp*
\begin{proof}
    We first pre-amplify the unitary using \autoref{lem:uniform_block_ampl} with some $\gamma \ge 2\varepsilon$.
    We get a $(\sqrt{2}\norm{A}, a + 1, \gamma)$-block-encoding of $A$ implemented at a cost of
    \begin{equation*}
        T_{A'} := \order{\frac{\alpha T_A}{\norm{A}} \logp{\frac{\norm{A}}{\gamma}}}
    \end{equation*}
    
    Now we invoke \autoref{lem:apply_block_enc} with $\delta = \frac{2\kappa\gamma}{\norm{A}}$ and the above unitary to prepare the state,
    which has a time complexity of
    \begin{equation*}
        \order{\kappa\paren{T_{A'} + T_b}} 
        = \order{\frac{\alpha\kappa}{\norm{A}}\logp{\frac{\kappa}{\delta}} T_A + \kappa T_b}
    \end{equation*}
\end{proof}

\thmbodyApplyBlockEncApprox*
\begin{proof}
  We know that $$\norm{\ket b - \ket{b'}} \le \frac{\varepsilon}{2\kappa}$$
  and $$\norm{\ket\psi - \frac{A\ket{b'}}{\norm{A\ket{b'}}}} \le \frac{\varepsilon}{2}$$
  For small enough $\varepsilon \ll \kappa$, we can assume that
    $\norm{A\ket b} \approx \norm{A\ket{b'}}$.
  We can derive the final error as
  \begin{align*}
    \norm{\ket\psi - \frac{A\ket{b}}{\norm{A\ket{b}}}}
    &= \norm{\ket\psi - \frac{A\ket{b} - A\ket{b'} + A\ket{b'}}{\norm{A\ket{b}}}} \\
    &= \norm{\ket\psi - \frac{A\ket{b'}}{\norm{A\ket{b}}}
        + \frac{A\ket{b'} - A\ket{b'}}{\norm{A\ket{b}}}} \\
    &\le \norm{\ket{\psi} - \frac{A\ket{b'}}{\norm{A\ket{b'}}}}
        + \norm{\frac{A\ket{b'} - A\ket{b'}}{\norm{A\ket{b}}}} \\
    &\le \frac{\varepsilon}{2} + \frac{\norm{A} \norm{\ket b - \ket{b'}}}{\norm{A\ket{b}}} \\
    &\le \frac{\varepsilon}{2} + \frac{\varepsilon}{2} \\
    &= \varepsilon
  \end{align*}
\end{proof}

\subsection{Arithmetic with Block-Encoded Matrices}

\thmbodyConstrLincomb*
\begin{proof}
Let $a = max_j(a_j) + s$ and $\alpha = \sum_j y_j \alpha_j$.
For each $j \in \{0,\ldots,m-1\}$,
construct the extended unitary $U'_j$ by padding ancillas to $U_j$,
i.e. $U_j^\prime = I_{a - s - a_j} \otimes U_j$.
Note that $U'_j$ is a $(\alpha_j, a - s, \varepsilon_j)$-block-encoding of $A_j$.
Let $B_j = (\bra{0}^{a_j} \otimes I_s)U_j(\ket{0}^{a_j} \otimes I_s)$
  denote the top left block of $U_j$ and $U_j^\prime$,
  and observe that $\|A_j - \alpha_j B_j \| \le \varepsilon_j$.
We also construct $P$
  --- an $\eta$ state-preparation unitary 
  s.t. $P\ket{0} = \sum_j \sqrt{y_j\alpha_j} \ket j$ ---
  by invoking \autoref{def:constr_opt_spu}.

Consider the unitary
$W = (P^\dagger \otimes I_{a - 1} \otimes I_s)(\sum_j \ketbra{j}{j} \otimes U_j^\prime ) (P\otimes I_{a - 1} \otimes I_s)$.
This is a $(\alpha, a, \varepsilon)$-block-encoding
of $A = \sum_j y_j A_j$,
where $\varepsilon$ is computed as:

\begin{align*}
\norm{ A - \alpha (\bra{0}^a \otimes I_s)W(\ket{0}^a \otimes I_s) }
&= \norm{ \sum_{j=0}^{m-1} y_j A_j - \alpha(\bra{0}^a \otimes I_s)W(\ket{0}^a \otimes I_s) } \\
&= \norm{ \sum_j y_j A_j - \alpha(\bra{0}^a \otimes I_s)(\sum_j P^\dagger \ketbra{j}{j} P \otimes U_j^\prime)(\ket{0}^a \otimes I_s) } \\
&= \norm{ \sum_j y_j A_j - \alpha\sum_j \bra{0}P^\dagger \ketbra{j}{j}P\ket{0} \otimes B_j } \\
&= \norm{ \sum_j \bigg( y_j A_j - \alpha\bra{0}P^\dagger \ketbra{j}{j}P\ket{0}  B_j \bigg) } \\
&= \norm{ \sum_j \bigg( y_j A_j - \alpha\Big(\frac{y_j \alpha_j}{\alpha}\Big) B_j \bigg) } \\
&\le \sum_j y_j \norm{ A_j - \alpha_j B_j } \\
&\le \sum_j y_j \varepsilon_j = \varepsilon
\end{align*}
\end{proof}

\thmbodyBlockEncTensor*
\begin{proof}
    From the property of Kronecker products $(A \otimes B)(C \otimes D) = (AC) \otimes (BD)$. For $j \in \{1, 2\}$ let $\tilde{A_j} = \paren{\bra{0}^{\otimes a_j} \otimes I_s} U_j \paren{\ket{0}^{\otimes a_j} \otimes I_s}$. Then it follows that 

    \begin{equation}
        \label{eqn:block_encoding_tensor_block}
        \paren{ \bra{0}^{\otimes a} \otimes I_s \otimes \bra{0}^{\otimes b} \otimes I_t } (U_1 \otimes U_2) \paren{ \ket{0}^{\otimes a} \otimes I_s \otimes \ket{0}^{\otimes b} \otimes I_t } = \tilde{A}_1 \otimes \tilde{A}_2
    \end{equation}

    Therefore $\tilde{A}_1 \otimes \tilde{A}_2$ is block-encoded in $U_1 \otimes U_2$ as a non-principal block-encoding, and we can use \texttt{SWAP} gates to move it to the principal block as follows.

    \begin{align*}
        S \paren{ \ket{0}^{\otimes a} \otimes I_s \ket{0}^{\otimes b} \otimes I_t } &= \Pi_{i=1}^s \texttt{SWAP}_{a+b+i}^{a+i} \paren{ \ket{0}^{\otimes a} \otimes I_s \ket{0}^{\otimes b} \otimes I_t } \\
                                                                                    &= \Pi_{i=1}^{s-1} \texttt{SWAP}_{a+b+i}^{a+i} \texttt{SWAP}_{a + b + s}^{a+s} \paren{ \ket{0}^{\otimes a} \otimes I_s \ket{0}^{\otimes b} \otimes I_t } \\
                                                                                    &= \Pi_{i=1}^{s-1} \texttt{SWAP}_{a+b+i}^{a+i} \paren{ \ket{0}^{\otimes a} \otimes I_{s-1} \ket{0}^{\otimes b} \otimes I_{t+1} } \\
                                                                                    &= \ldots \\ 
                                                                                    &= \ket{0}^{\otimes a+b} \otimes I_{s+t}
    \end{align*}

    Similarly,

    \begin{equation*}
        \paren{\bra{0}^{\otimes a} \otimes I_s \otimes \bra{0}^{\otimes b} \otimes I_t} S^{\dagger} = \bra{0}^{\otimes a + b} \otimes I_{s+t}.
    \end{equation*}

    From \autoref{eqn:block_encoding_tensor_block} we have 

    \begin{align*}
        \tilde{A}_1 \otimes \tilde{A}_2 &= \paren{\bra{0}^{\otimes a} \otimes I_s \otimes \bra{0}^{\otimes b} \otimes I_t} S^{\dagger}S(U_1 \otimes U_2) S^{\dagger}S \paren{ \ket{0}^{\otimes a} \otimes I_s \ket{0}^{\otimes b} \otimes I_t } \\ 
                                        &= \paren{\bra{0}^{\otimes a + b} \otimes I_{s+t}} S (U_1 \otimes U_2) S^{\dagger} \paren{\ket{0}^{\otimes a+b} \otimes I_{s+t}}
    \end{align*}

    Next, we look at the subnormalization and error terms. 

    \begin{align*}
        \norm{A_1 \otimes A_2 - \alpha \beta \tilde{A}_1 \otimes \tilde{A}_2}_2 &\leq \norm{(\alpha \tilde{A}_1 + \varepsilon_1 I_s) \otimes (\beta \tilde{A}_2 + \varepsilon_2 I_t) - \alpha \tilde{A}_1 \otimes \beta \tilde{A}_2}_2 \\ 
                                                                                &= \norm{\alpha \tilde{A}_1 \otimes \varepsilon_2 I_2 + \varepsilon_1 I_s \otimes \beta \tilde{A}_2 + \varepsilon_1 I_s \otimes \varepsilon_2 I_2}_2 \\
                                                                                &\leq \alpha \varepsilon_2 \norm{\tilde{A}_1}_2 + \beta \varepsilon_2 \norm{\tilde{A}_2}_2 + \varepsilon_1 \varepsilon_2 \\ 
                                                                                &= \alpha \varepsilon_2 + \beta \varepsilon_1 + \varepsilon_1 \varepsilon_2
    \end{align*}

    where we have used $\norm{A_1}_2 \leq \alpha \norm{\tilde{A}_1}_2 + \varepsilon_1$ and $\norm{\tilde{A}_1}_2 \leq 1$ and similarly for $A_2$. 

\end{proof}
\end{document}